\documentclass[lettersize,journal]{IEEEtran}

\usepackage[T1]{fontenc}
\usepackage{aecompl} 
\usepackage{caption}

\ifCLASSINFOpdf
   \usepackage[pdftex]{graphicx}

   \graphicspath{{{\vartheta}_{k}^v/pdf/}{{\vartheta}_{k}^v/jpeg/}}

   \DeclareGraphicsExtensions{.pdf,.jpeg,.png}
\else

   \usepackage[dvips]{graphicx}

   \graphicspath{{{\vartheta}_{k}^v/eps/}}
  
   \DeclareGraphicsExtensions{.eps}
\fi

\usepackage{amsmath}
\usepackage{subeqnarray}
\usepackage{cases}
\usepackage{cite}
\usepackage{graphicx}
\usepackage{subfigure}
\usepackage{tabularx,booktabs}
\usepackage{dingbat}
\usepackage{diagbox}
\usepackage[ruled,linesnumbered]{algorithm2e}
\usepackage{amssymb}
\usepackage{algorithmic}
\usepackage{esint}
\usepackage{color}
\usepackage{lipsum}
\usepackage{cuted}
\usepackage{stfloats}
\usepackage{multirow}
\usepackage{ntheorem}
\usepackage{color}
\usepackage{threeparttable}
\usepackage{bbm}
\usepackage{tabu}
\usepackage{bbding}
\usepackage{tikz}
\usepackage{xcolor}
\usepackage{makecell}

\usepackage[colorlinks,
            linkcolor=black,       %%修改此处为你想要的颜色
            anchorcolor=black,  %%修改此处为你想要的颜色
            citecolor=black,        %%修改此处为你想要的颜色，例如修改blue为red
            ]{hyperref}
\newcolumntype{C}{>{\centering\arraybackslash}X}
\theoremseparator{:}
\newtheorem{proposition}{\textcolor{black}{Proposition}}

\newtheorem{definition}{Definition}
\theoremstyle{nonumberplain}
\theorembodyfont{\normalfont}
\newtheorem{proof}{Proof}
\newtheorem{Proof}{Proof}

\definecolor{newextractedpurple}{RGB}{127,0,127}

\usepackage[cmintegrals]{newtxmath}

\begin{document}

% \title{\huge Favors Over Fairness: Priority-Aware Collaborative Perception for Connected and Autonomous Vehicles}
% \title{From calibration to Streaming: \\Real-Time Edge Camera Networks with \\Robust Task-Oriented Communications}
% \title{Real-Time Edge Camera Networks with \\Robust Task-Oriented Communications}
\title{\huge R-ACP: Real-Time Adaptive Collaborative Perception Leveraging Robust Task-Oriented Communications}

\author{Zhengru~Fang,~%\IEEEmembership{Student Member,~IEEE,}
Jingjing~Wang, %~\IEEEmembership{Senior Member,~IEEE,}
Yanan~Ma, 
Yihang Tao,%~\IEEEmembership{Student Member,~IEEE,}\\
\\
Yiqin~Deng, %~\IEEEmembership{Member,~IEEE, }\\ %~\IEEEmembership{\normalsize Member,~IEEE,}
%H.~Vincent~Poor,~\IEEEmembership{\normalsize Life Fellow,~IEEE}
Xianhao~Chen, %~\IEEEmembership{Member,~IEEE}
and Yuguang Fang,~\IEEEmembership{Fellow,~IEEE}%
\IEEEcompsocitemizethanks{\IEEEcompsocthanksitem Z. Fang, Y. Ma, Y. Tao, Y. Deng and Y. Fang are with the Department of Computer Science, City University of Hong Kong, Hong Kong. E-mail: \{zhefang4-c, yananma8-c, yihang.tommy\}@my.cityu.edu.hk, \{yiqideng, my.fang\}@cityu.edu.hk.
% \IEEEcompsocthanksitem Y. Zhang is with Department of Automation, Tsinghua University, China. E-mail: zya21@mails.tsinghua.edu.cn.
\IEEEcompsocthanksitem J. Wang is with the School of Cyber Science and Technology, Beihang University, China. Email: drwangjj@buaa.edu.cn.
\IEEEcompsocthanksitem X. Chen is with the Department of Electrical and Electronic Engineering, the University of Hong Kong, Hong Kong. E-mail: xchen@eee.hku.hk.}
\thanks{This work was supported in part by the Hong Kong SAR Government under the Global STEM Professorship and Research Talent Hub,  the Hong Kong Jockey Club under the Hong Kong JC STEM Lab of Smart City (Ref.: 2023-0108), and the Hong Kong Innovation and Technology Commission under InnoHK Project CIMDA. This work of J. Wang was partly supported by the National Natural Science Foundation of China under Grant No. 62222101 and No. U24A20213, partly supported by the Beijing Natural Science Foundation under Grant No. L232043 and No. L222039, partly supported by the Natural Science Foundation of Zhejiang Province under Grant No. LMS25F010007. The work of Y. Deng was supported in part by the National Natural Science Foundation of China under Grant No. 62301300. The work of X. Chen was supported in part by the Research Grants Council of Hong Kong under Grant 27213824 and CRS HKU702/24, in part by HKU-SCF FinTech Academy R\&D Funding, and in part by HKU IDS Research Seed Fund under Grant IDS-RSF2023-0012. \textit{(Corresponding author: Yiqin Deng})}

% \thanks{A preliminary version has been submitted for the IEEE Global Communications Conference (GLOBECOM 2024) (We add this version in our submission).}
}
% \markboth{IEEE Journal on Selected Areas in Communications}%
% {Shell \MakeLowercase{\textit{et al.}}: Bare Demo of IEEEtran.cls for IEEE Communications Society Journals}
%to develop an effective timely information sharing  and robust task-oriented communication strategy to optimize online calibration and efficient feature sharing for collaborative multi-view perception (CoMP).  
% we propose a robust task-oriented communication (R-TOCOM) strategy to optimize online self-calibration and efficient feature sharing for \underline{R}eal-time \underline{A}daptive \underline{C}llaborative \underline{P}erception (R-ACP)
\maketitle
\begin{abstract}
Collaborative perception enhances sensing in multi-robot and vehicular networks by fusing information from multiple agents, improving perception accuracy and sensing range. However, mobility and non-rigid sensor mounts introduce extrinsic calibration errors, necessitating online calibration, further complicated by limited overlap in sensing regions.
Moreover, maintaining fresh information is crucial for timely and accurate sensing. To address calibration errors and ensure timely and accurate perception, we propose a robust task-oriented communication strategy to optimize online self-calibration and efficient feature sharing for \underline{R}eal-time \underline{A}daptive \underline{C}ollaborative \underline{P}erception (R-ACP).
Specifically, we first formulate an Age of Perceived Targets (AoPT) minimization problem to capture data timeliness of multi-view streaming. Then, in the calibration phase, we introduce a channel-aware self-calibration technique based on re-identification (Re-ID), which adaptively compresses key features according to channel capacities, effectively addressing calibration issues via spatial and temporal cross-camera correlations. In the streaming phase, we tackle the trade-off between bandwidth and inference accuracy by leveraging an Information Bottleneck (IB)-based encoding method to adjust video compression rates based on task relevance, thereby reducing communication overhead and latency. Finally, we design a priority-aware network to filter corrupted features to mitigate performance degradation from packet corruption. Extensive studies demonstrate that our framework outperforms five baselines, improving multiple object detection accuracy (MODA) by 25.49\% and reducing communication costs by 51.36\% under severely poor channel conditions. Code will be made publicly available: \href{https://github.com/fangzr/R-ACP}{github.com/fangzr/R-ACP}.
\end{abstract}
% \vspace{-4mm}
% \vspace{-3mm}

\section{Introduction}
\subsection{Background}
\IEEEPARstart{C}{ollaborative} perception systems are increasingly prevalent in fields such as IoT systems\cite{wang2024generative,9735326}, connected and autonomous driving\cite{10976336,Chen2024,hu2025cpguardnewparadigmmalicious}, unmanned aerial vehicles\cite{10937373,fang2025taskoriented,hou2025splitfederated}, and sports analysis \cite{yang2022traffic,he2020multi}. They offer significant advantages over single-agent systems by mitigating blind spots, reducing occlusions, and providing comprehensive coverage through multiple perspectives \cite{10557621}, which is especially valuable in cluttered or crowded environments. However, these benefits also introduce considerable challenges. The increased number of cameras demands higher network bandwidth, and more fine-grained synchronization. Besides, synchronized data transmission with high inference accuracy necessitates precise calibration and efficient communication \cite{wang2022rt}. Therefore, balancing network resource management for real-time collaborative perception becomes essential.

\begin{figure}[t]
  \centering
  \subfigure[Unpredictable accidents can alter a UGV's camera extrinsic parameters.]{
    \includegraphics[width=0.46\textwidth]{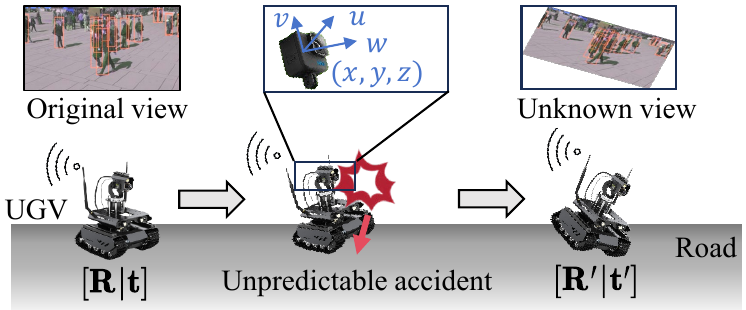}\label{fig:ugv_collaborative_perception}
  }
  \subfigure[Incorrect extrinsic parameters cause errors in collaborative perception.]{
    \includegraphics[width=0.48\textwidth]{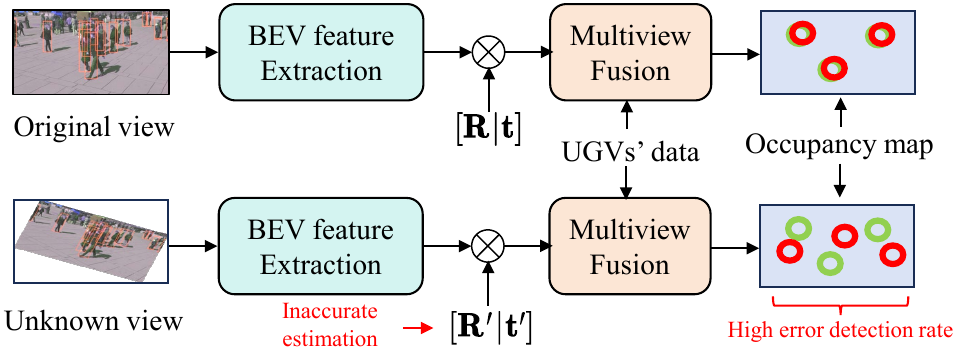}\label{fig:high_error_rate}
  }
  \caption{Effect of unpredictable accidents involving UGVs on camera extrinsic parameters and perception error rates. }
  \label{fig:ugv_perception_error}
  \vspace{-3mm}
\end{figure}

Camera calibration is the process of determining the intrinsic and extrinsic parameters of vision sensors to ensure accurate perception across different viewpoints in multi-camera networks \cite{yang2024robust}. Intrinsic calibration deals with internal characteristics like focal length and lens distortion, while extrinsic calibration defines the physical position and orientation (rotation \(\mathbf{R}\) and translation \(\mathbf{t}\)) of the camera relative to a reference frame. This enables effective alignment across cameras.
\mbox{Fig. \ref{fig:ugv_perception_error}} illustrates a scenario where multiple unmanned ground vehicles (UGVs) equipped with vision sensors collaborate to track moving objects, reducing the impact of obstacles and enhancing perception accuracy. In smart cities, UGVs can effectively prevent elderly falls, provide real-time abnormal behavior alerts (e.g., crime or terrorist attacks), and assist in search and rescue missions in hazardous areas. As shown in \mbox{Fig. \ref{fig:ugv_perception_error}(a)}, accidents can disrupt extrinsic parameters \([\mathbf{R} | \mathbf{t}]\), resulting in an ``unknown view''. In Fig. \ref{fig:ugv_perception_error}(b), these errors impact collaborative perception, leading to inaccurate Bird's Eye View (BEV) mapping when parameters \(\left[ \mathbf{R}^{'}|\mathbf{t}^{'} \right]\) are incorrect. Thus, efficient extrinsic calibration is vital for perception accuracy.  

% Add explaination why we did not use kalman fliter
Traditional calibration methods—pattern-based \cite{tabb2019calibration}, manual measurements \cite{ozuysal2004manual}, and feature-based approaches \cite{9495137}—rely on predefined calibration objects and work well in controlled settings. However, they are impractical for dynamic, large-scale deployments due to the need for specific targets, time-consuming processes, and inconsistent natural features. Multi-camera calibration also faces bandwidth limitations and network variability since collaboration requires data exchanges among sensors.
Kalman filter-based methods provide online calibration without predefined objects \cite{khodarahmi2023review} but rely on linear assumptions about system dynamics. These assumptions often fail in multi-UGV systems due to non-linear motions, especially in rotations and velocity changes, making them less effective in rapidly changing scenarios. Transmitting large volumes of image data for collaborative calibration consumes significant bandwidth, and network delays or packet loss disrupt synchronization, complicating cross-camera feature matching \cite{fang2024pacp,fang2024pib,fang2025ton}, particularly when adding or adjusting cameras.
To address these challenges, we propose an adaptive calibration approach leveraging spatial and temporal cross-camera correlations during deployment. By incorporating re-identification (Re-ID) technology \cite{yang2022traffic}, our method achieves higher key-point matching accuracy than traditional edge-detection techniques. The Re-ID-based calibration effectively handles non-linearities in UGV motion and varying camera perspectives by using both global appearance and fine-grained local features, enhancing robustness under changing conditions. Additionally, adaptive feature quantization based on channel capacity reduces communication overhead, maintaining high-precision calibration without extra sensors or specific calibration objects, making it suitable for dynamic and large-scale deployments.

% To address these challenges, we propose an adaptive calibration approach that leverages spatial and temporal cross-camera correlations during deployment. By incorporating re-identification (Re-ID) technology, which offers higher matching accuracy than traditional edge-detection methods, and dynamically adjusting feature quantization based on channel capacity, our method optimizes calibration under varying network conditions. This ensures robust calibration even as cameras are dynamically adjusted or added due to target movement .

After meeting the accuracy requirements of vision sensing through self-calibration, we need to address how to guarantee the timely data transmission in multi-camera networks, which is essential for monitoring dynamic environments. For example, timely data in life critical signal monitoring can be life-saving \cite{yang2022mixed}, and promptly detecting abnormal behaviors in public safety can prevent crimes or ensure traffic safety \cite{abedi2024safety}. Therefore, perception tasks like object detection rely on both data accuracy and its timeliness, as stale information can result in poor decision-making when immediate responses are required.
The Age of Information (AoI) measures data timeliness by tracking the time since the latest packet was received \cite{fangAgeInformationEnergy2022,Wu2025MNET}. Traditional AoI assumes homogeneous data sources of equal importance and consistent quality, simplifying timeliness evaluation. However, this cannot align with multi-camera networks, where cameras have varying fields of views (FOVs) and data quality due to different positions and environmental factors. While He \textit{et al.} \cite{he2018minimizing} considered AoI in multi-camera perception, they did not adequately model camera coverage or account for overlapping fields of view, where variations in sensing accuracy affect multi-view fusion performance. To fill this gap, we propose a novel age-aware metric for multi-camera networks that reflects both data timeliness and source quality. This enhanced metric guides optimizations for sources with higher priorities, ultimately improving overall perception performance.

Limited network bandwidth and high redundancy in video streaming increase transmission overhead \cite{fang2024pacp}. Traditional systems transmit vast amounts of raw data without considering task relevance, causing latency that degrades real-time perception. Additionally, channel limitations and multi-user interference can corrupt transmitted features, but existing task-oriented communication seldom addresses transmission robustness, often assuming ideal channels.
Task-oriented communication offers efficiency by focusing on task-relevant data and ignoring redundancy \cite{shao2023task}, prioritizing compact representations for tasks like object detection. However, without robustness considerations, these methods remain vulnerable to channel impairments.
Traditional solutions like Automatic Repeat reQuest (ARQ) protocols enhance reliability through retransmissions but introduce significant overhead and latency, unsuitable for real-time applications\cite{10268059}. Therefore, there is a need to develop robust task-oriented communication methods that withstand channel impairments without extra latency.
The Information Bottleneck (IB) method \cite{fang2025ton} aligns with this approach by encoding the most relevant features while enhancing resilience to data corruption. By leveraging robust task-oriented communication, we can optimize network resources and improve multi-camera network performance, ensuring efficient and effective real-time collaborative perception even under bandwidth constraints and poor channel conditions.

% In this paper, we propose a calibration-driven task-oriented communication framework for multi-camera networks. Our contributions include an adaptive calibration method using spatial and temporal cross-camera correlations to ensure accurate perception during deployment under bandwidth constraints. For streaming, we develop an AoI-based communication model to optimize bandwidth and bitrate, ensuring timely updates. We also employ submodular optimization to dynamically select active cameras, enhancing network efficiency. Finally, we integrate an IB-based communication strategy to reduce overhead by selectively transmitting task-relevant data. Our framework improves perception accuracy and network efficiency under constrained conditions, outperforming conventional methods.
\subsection{State-of-the-Art}
\label{sec:Related Work}
\subsubsection{Multi-Camera Networks}  
Multi-camera networks enhance perception by providing comprehensive coverage and reducing occlusions through multiple views from distributed cameras. Yang \textit{et al.} \cite{YANG2023103982} introduced the edge-empowered cooperative multi-camera sensing system for traffic surveillance, leveraging edge computing and hierarchical re-identification to minimize bandwidth usage while maintaining vehicle tracking accuracy. Liu \textit{et al.} \cite{liu2023efficient} proposed a Siamese network-based tracking algorithm that enhances robustness against occlusion and background clutter in intelligent transportation systems. For pedestrian detection, Qiu \textit{et al.} \cite{qiu20223d} improved detection accuracy by using multi-view information fusion and data augmentation to address occlusion challenges.  
Guo \textit{et al.} \cite{guo2021optimal} explored wireless streaming optimization for 360-degree virtual reality video, focusing on joint beamforming and subcarrier allocation to reduce transmission power. However, existing research has not fully leveraged spatial and temporal correlations among multiple camera views to optimize coverage. Fang \textit{et al.} \cite{fang2025ton} developed a collaborative perception framework to leverage correlations among frames and perspectives through a prioritization mechanism, but the priorities cannot be adjusted for dynamically perceived targets. Additionally, the effect of data timeliness, particularly AoI, on collaborative perception accuracy remains underexplored.

\subsubsection{Visual Sensor Calibration}

Calibration is essential for ensuring accuracy in multi-camera networks. Traditional methods include pattern-based, manual measurement, and feature-based techniques. \emph{Pattern-based methods}, like using checkerboard patterns \cite{tabb2019calibration}, are impractical in dynamic environments due to the unavailability of specific targets. \emph{Manual measurement based methods}, requiring physical measurements of camera positions \cite{ozuysal2004manual}, are time-consuming and unsuitable for rapidly changing settings, such as connected and autonomous driving. \emph{Feature-based methods} match natural features across overlapping views \cite{9495137}, but inconsistent features and limited overlap reduce their reliability. In mobile robotic systems, frequent recalibration is often needed due to unpredictable conditions. However, real-time calibration transmission consumes significant bandwidth, which can degrade perception accuracy. Collaborative perception, requiring multi-view data, further increases communication resource demands, complicating precise calibration, especially when environmental changes necessitate camera adjustments. Motion-based techniques, like Su \textit{et al.} \cite{10218989}, estimate transformations through sensor motions, but accuracy may be limited without collaborative vehicle assistance. Yin \textit{et al.} \cite{10195910} introduced a targetless method, combining motion and feature-based approaches, improving accuracy but often requiring feature alignment, which increases data transmission. Thus, how to optimize communication protocols is crucial for managing calibration overhead and maintaining perception accuracy. Integrating both calibration techniques and communication strategies is essential for achieving real-time, multi-camera calibration in dynamic deployments.

\subsubsection{Task-Oriented Communications}

Recent advancements in task-oriented communication have shifted the focus from bit-level to semantic-level data transmission. Wang \textit{et al.} designed a semantic transmission framework for sharing sensing data from the physical world to Metaverse\cite{10007890}. The proposed method can achieve the sensing performance without data recovery. Meng \textit{et al.} \cite{10370739,10422886} proposed a cross-system design framework for modeling robotic arms in Metaverse, integrating Constraint Proximal Policy Optimization (C-PPO) to reduce packet transmission rates while optimizing scheduling and prediction. Kang \textit{et al.} \cite{kang2022personalized} explored semantic communication in UAV image-sensing, designing an energy-efficient framework with a personalized semantic encoder and optimal resource allocation to address efficiency and personalization in 6G networks. Wei \textit{et al.} \cite{10225550} introduced a federated semantic learning (FedSem) framework for collaborative training of semantic-channel encoders, leveraging the information bottleneck theory to enhance rate-distortion performance in semantic knowledge graph construction. Shao \textit{et al.} \cite{shao2023task} proposed a task-oriented framework for edge video analytics, focusing on minimizing data transmission by extracting compact task-relevant features and utilizing temporal entropy modeling for reduced bitrate.
% Additionally, Shao \textit{et al.} also \cite{shao2024task} developed TOCOM-V2I for vehicle-to-infrastructure cooperative perception, reducing bandwidth by transmitting only task-relevant information and employing spatial-aware feature selection and hierarchical entropy modeling to enhance perception accuracy and efficiency. 
However, for multi-camera networks, considering correlations between cameras and task-oriented priorities allows for further data compression, optimizing transmission efficiency based on varying perceptual and transmission needs.

\subsection{Our Contributions}
Multi-camera networks enhance real-time collaborative perception by leveraging multiple views. Our contributions are summarized as follows.

\begin{itemize}
  \item We propose a novel robust task-oriented communication strategy for \underline{R}eal-time \underline{A}daptive \underline{C}ollaborative \underline{P}erception (R-ACP), which optimizes calibration and feature transmission across calibration and streaming phases. It enhances perception accuracy while managing communication overhead. We also formulate the Age of Perceived Targets (AoPT) minimization problem to ensure both data quality and timeliness.
  
  \item We introduce a channel-aware self-calibration technique based on Re-ID, which adaptively compresses key-point features based on channel capacity and leverages spatial and temporal cross-camera correlations, improving calibration accuracy by up to 89.39\%.

  \item To balance bandwidth and inference accuracy, we develop an Information Bottleneck (IB)-based encoding method to dynamically adjust video compression rates according to task relevance, reducing communication overhead and latency while maintaining perception accuracy.
  
  \item To address severe packet errors or loss without retransmission in real-time scenarios, we design a priority-aware multi-view fusion network that discards corrupted data by dynamically adjusting the importance of each view, ensuring robust performance even under challenging network conditions.

  \item Extensive evaluations demonstrate that our R-ACP framework outperforms conventional methods, achieving significant improvements in multiple object detection accuracy (MODA) by 25.49\% and reducing communication costs by 51.36\% under constrained network conditions.
\end{itemize}

The remainder of this paper is organized as follows. Sec. \ref{sec:System_Model} and Sec. \ref{sec:Problem Formulation} introduce the communication and calibration models, analyze data timeliness, and formulate the optimization problem. Sec. \ref{sec:Methodology} details our methodology, focusing on Re-ID-based camera calibration, task-oriented compression using the IB principle, and adaptive \& robust streaming scheduling. Finally, Sec. \ref{sec:Performance Evaluation} evaluates our framework through simulations, demonstrating improved MODA and reduced communication costs under constrained network conditions.

\section{System Model and Preliminary}
\label{sec:System_Model} 
\begin{figure}[t]
  \centering
  \includegraphics[width=0.50\textwidth]{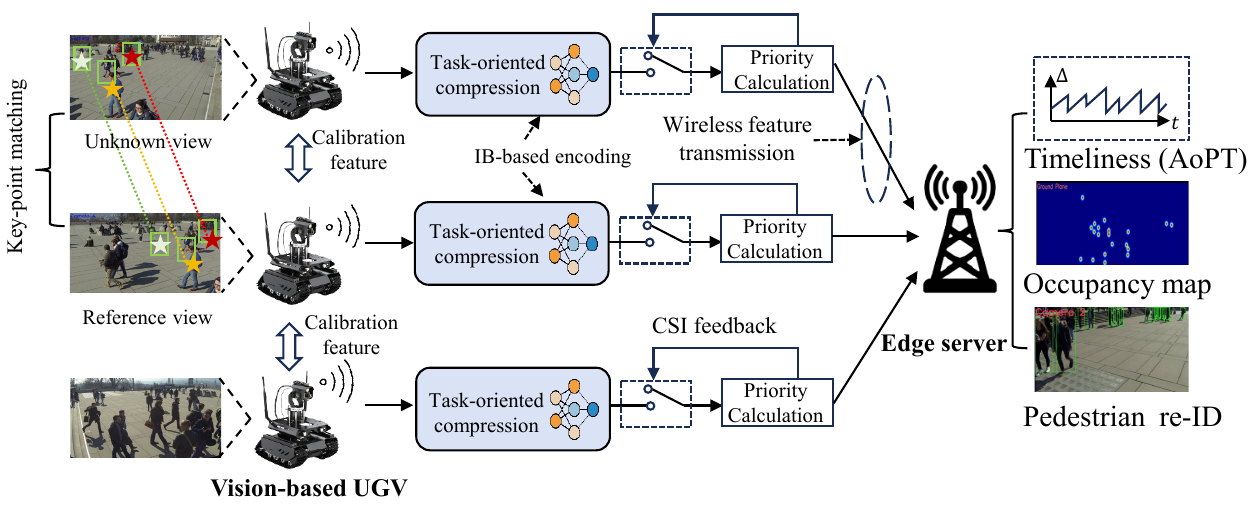}
  \caption{The system consists of several UGVs equipped with cameras, collaboratively tracking pedestrians. }
  \label{fig:system model}
  % \vspace{-5mm}
\end{figure}

\subsection{Scenario Description}
\label{Sec: Scenario Description}
As illustrated in Fig.~\ref{fig:system model}, our system consists of multiple vision-based UGVs equipped with edge cameras, denoted as \( \mathcal{K} = \{1, 2, \ldots, K\} \), that collaboratively track mobile targets, such as pedestrians, within their FOVs. The UGVs are responsible for transmitting decoded features to edge server through wireless channel, which generate a \textit{pedestrian occupancy map} and conducts \textit{pedestrian re-identification (Re-ID)} tasks. Edge servers have more powerful computing and storage capability for DNN-based downstream tasks\cite{9043503}. 
However, UGVs encounter several challenges in dynamic environments. First, unpredictable factors such as terrain changes and obstacles cause sudden variations in camera extrinsic parameters, leading to degraded perception accuracy over time. Traditional methods like Kalman filtering struggle to handle these rapid, non-linear variations due to their reliance on accurate initial states and linearity assumptions. To address this, we introduce a \textit{Re-ID-based collaborative perception} mechanism, where nearby UGVs share perceptual information, allowing real-time calibration of extrinsic parameters without the need of precise initial settings or additional sensors. Another challenge is to ensure the timeliness of the high-quality data being collected, especially in high-mobility scenarios. To address this, we introduce the \textit{Age of Perceived Targets (AoPT)} metric and formulate a new optimization problem. By adjusting the frame rate and applying \textit{Information Bottleneck (IB)-based encoding}, we reduce spatiotemporal redundancy and improve the perceptual data timeliness. The proposed approach is also extensible to other robotic platforms such as UAVs and autonomous robots. Furthermore, unpredictable channel impairments can lead to packet errors or loss during transmission. To mitigate this, we design a priority-aware network, which selectively fuses data from multiple UGVs based on channel conditions, filtering out erroneous information to ensure robust perception performance.

\subsection{Communication Model}
\label{Sec: Communication Model}
To manage communications between UGVs and an edge server, we adopt a Frequency Division Multiple Access (FDMA) scheme. The transmission capacity \( C_k \) for each UGV \( k \) is determined by the Shannon capacity formula, which depends on the signal-to-noise ratio (SNR) at the receiver:
\begin{equation}\label{eq:c_k}
C_k = B_k \log_2 \left(1 + \text{SNR}_k\right),
\end{equation}
where \( B_k \) is the bandwidth allocated to the link between UGV \( k \) and the edge server, and SNR is given by:
\begin{equation}
\text{SNR}_k = \frac{P_t G_k}{N_0 B_k},
\end{equation}
where \( P_t \) is the transmission power, \( G_k \) is the channel gain for UGV \( k \), \( N_0 \) is the noise power spectral density, and \( B_k \) is the bandwidth allocated to UGV \( k \). Besides, the transmission delay \( d_k^T \) for each camera-server connection is then determined by the amount of data to be transmitted \( D \) and the capacity \( C_k \):
\begin{equation}\label{eq:transmission_delay}
d_{k}^{T}=\frac{D}{C_k}=D\left[ B_k\log _2\left( 1+\frac{P_tG_k}{N_0B_k} \right) \right] ^{-1}.
\end{equation}

Thus, the total delay for each UGV \( k \), which includes the inference delay \( d_k^I \) at the edge server, is given by:
\begin{equation}\label{eq:total_delay}
d_k^{\text{total}} = d_k^T + d_k^I.
\end{equation}

\subsection{Camera Calibration and Multi-view Fusion}
\label{Sec: Camera Calibration and Multiview Fusion}
\begin{figure}[t]
  \centering
  \includegraphics[width=0.50\textwidth]{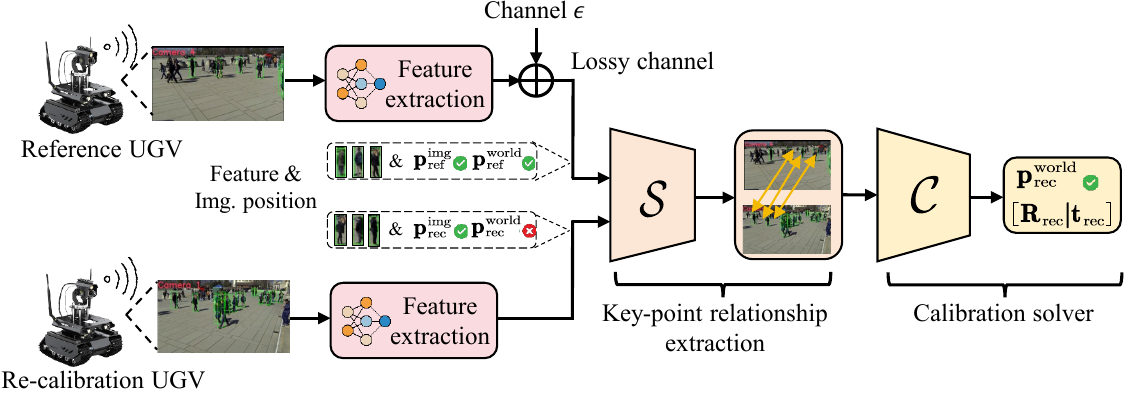}
  \caption{The flow of the self-calibration method using multiview feature sharing. }
  \label{fig:calibration_flow}
  % \vspace{-6mm}
\end{figure}
%Explain the calibration process during the deployment phase and the challenges associated with it
Our multi-UGV collaborative perception system operates in three phases: \textbf{Idle} (Phase 0), \textbf{Calibration} (Phase 1), and \textbf{Streaming} (Phase 2). During Phase 0, UGVs perform object detection without transmitting data. Phase 1 occurs when new UGVs are deployed or when existing UGVs require recalibration to improve tracking accuracy. Phase 2 begins when targets are detected, prompting real-time data transmission to the fusion node $s$.
% \begin{figure}[t]
%   \centering
%   \includegraphics[width=0.35\textwidth]{figure/fig_multiview.pdf}
%   \caption{Multi-camera networks for calibration using feature matching.}
%   \label{fig:multiview-model}
%   \vspace{-3mm}
% \end{figure}

\textbf{1) Calibration (Phase 1)}: Calibration involves estimating intrinsic and extrinsic parameters for each UGV's camera\cite{9495137}. Intrinsic parameters are defined by the intrinsic matrix \( \mathbf{K} \), which includes focal lengths \( f_x, f_y \) and principal point \( (c_x, c_y) \): 
\begin{equation}
  \mathbf{K} = \begin{pmatrix} 
  f_x & 0 & c_x \\
  0 & f_y & c_y \\
  0 & 0 & 1
  \end{pmatrix}.
\end{equation}
Extrinsic parameters represent the camera's orientation and position, encapsulated in the rotation matrix \( \mathbf{R} \) and translation vector \( \mathbf{t} \), forming the extrinsic matrix \( [\mathbf{R} | \mathbf{t}] \):
\begin{equation}
  [\mathbf{R} | \mathbf{t}] = \begin{pmatrix} 
  r_{11} & r_{12} & r_{13} & t_x \\
  r_{21} & r_{22} & r_{23} & t_y \\
  r_{31} & r_{32} & r_{33} & t_z 
  \end{pmatrix}.
\end{equation}
The full perspective transformation matrix is then:
\begin{equation}
\mathbf{P} = \mathbf{K} [\mathbf{R} | \mathbf{t}].
\end{equation}
Given a 3D point \( \mathbf{P}^{\text{world}} = [x, y, z, 1]^T \), its 2D image projection \( \mathbf{p}^{\text{img}} = [u, v, 1]^T \) is calculated by \( \mathbf{p}^{\text{img}} = \mathbf{P} \mathbf{P}^{\text{world}} \). As shown in Fig.~\ref{fig:calibration_flow}, calibration depends on sharing detected feature points between the reference camera and the re-calibration UGV through a wireless channel. When the re-calibration UGV requires external calibration, it broadcasts a request to nearby UGVs, requesting them to share their extracted features along with the corresponding image coordinates \( \mathbf{p}^{\text{img}}_{\text{ref}} \) and world coordinates \( \mathbf{P}^{\text{world}}_{\text{ref}} \). The re-calibration UGV then exploits a Key-point Relationship Extraction (KRE) network \( \mathcal{S} \) to select the UGV with the most matched points as the reference UGV. The reference UGV transmits multi-frame feature and position information through a lossy channel. Finally, the re-calibration UGV uses this data in the Calibration Solver \( \mathcal{C} \) to solve the linear equations and determine its extrinsic parameters \([ \mathbf{R}_{\text{rec}}| \mathbf{t}_{\text{rec}} ]\)\footnote{Intrinsic parameters are relatively simple to calculate, while extrinsic parameters, which define camera position and orientation, require more intricate computations. Hence, we focus on calibrating extrinsic parameters in our multi-camera network.}.

% Rotation matrix error \( e_{\mathbf{R}} \) and translation vector error \( e_{\mathbf{t}} \) are computed similarly:
% \[
% e_{\mathbf{R}} = \frac{\| \mathbf{R}_{\text{rec}} - \mathbf{R}_{\text{gt}} \|_F}{\| \mathbf{R}_{\text{gt}} \|_F} \times 100, \quad e_{\mathbf{t}} = \frac{\| \mathbf{t}_{\text{rec}} - \mathbf{t}_{\text{gt}} \|_2}{\| \mathbf{t}_{\text{gt}} \|_2} \times 100.
% \]

%Our experiments demonstrate that Re-ID aided matching significantly outperforms the other two methods, both in match accuracy and calibration quality, thereby enhancing the overall performance of the multi-camera network.

\textbf{2) Streaming (Phase 2)}: For tasks like pedestrian detection, objects are often assumed to lie on the ground plane \( z = 0 \). This assumption simplifies the projection to a 2D-to-2D transformation between views. For a ground plane point \( \mathbf{P}_{\text{ground}} = [x, y, 0, 1]^T \), the image projection becomes:
\begin{equation}
\mathbf{p}^{\text{img}} = \mathbf{P}_0 \mathbf{P}_{\text{ground}},
\end{equation}
where \( \mathbf{P}_0 \) is the simplified 3x3 perspective matrix obtained by eliminating the third column from the extrinsic matrix.

% \begin{figure}[t]
%   \centering
%   \includegraphics[width=0.24\textwidth]{figure/fig_state-trans.pdf}
%   \caption{The state transition of three phases.}
%   \label{fig:state transition}
%   \vspace{-3mm}
% \end{figure}
To evaluate data timeliness in multi-UGV collaborative perception, we calculate the proportion of time the system spends in Phases 0 (Idle) and 2 (Streaming). We assume target occurrences are independent with a constant rate \(\lambda\), forming a Poisson process. Pedestrian dwell times are modeled as a Log-normal distribution \(S \sim \text{LogNormal}(\mu_S, \sigma_S^2)\) \cite{peng2009walking}. We adopt the log-normal distribution because pedestrian dwell times are inherently non-negative. It also captures the right-skewed nature of dwell times: most are short, but only some pedestrians stay longer, reflecting real-world behavior. The Exponential distribution's memoryless property is unsuitable since pedestrian leaving probability depends on the time already spent. Since cameras must capture targets before they leave, we consider infinite servers, so each target is served immediately without queuing. Therefore, we model these phases as an \(M/G/\infty\) queue with periodic Calibration (Phase 1), where \(M\) denotes Poisson arrivals, \(G\) is a general service time distribution, and servers are infinite.

% The state transition of the three phases is shown in Fig. \ref{fig:state transition}\footnote{Transitions \(\mathrm{A}_{00}\), \(\mathrm{A}_{11}\), and \(\mathrm{A}_{22}\) represent the system maintaining its current state. The system moves to calibration (\(\mathrm{A}_{01}\)) when a new camera is added or recalibration is needed. Calibration either succeeds (\(\mathrm{A}_{12}\)), transitioning to streaming, or fails (\(\mathrm{A}_{10}\)), returning to idle. Streaming starts directly from idle (\(\mathrm{A}_{02}\)) when a target is detected and returns to idle (\(\mathrm{A}_{20}\)) when the target leaves. Calibration may also be triggered during streaming (\(\mathrm{A}_{21}\)).}, where $\mathrm{A}_{ij}$ represents the event of phase $i$ transitioning to phase $j$. 

Let \( L(t) \) represent the number of active targets in the system (i.e., the number of targets currently being captured by cameras, corresponding to Phase 2). The steady-state distribution of \( L(t) \) is Poisson with mean \( \rho = \lambda \mathbb{E}[S] =\lambda \exp\left(\mu_S + \frac{\sigma_S^2}{2}\right)\). Thus, the probability that there are \( n \) targets being captured (Phase 2) is $P(L = n) = \frac{\rho^n e^{-\rho}}{n!}$,
% \begin{equation}
% P(L = n) = \frac{\rho^n e^{-\rho}}{n!},
% \end{equation}
where \( \rho \) represents the expected number of active targets. Then, the probability that there are no targets being captured (Phase 0, Idle) is $P(L = 0) = e^{-\rho}$. Phase 1 (Calibration) is deterministic, and we assume it occurs with a fixed probability \( p_1 \). 
Let \( T^{\text{total}} \) be the total cycle time, including time spent in Phase 1. The average time spent in Phase 1 is \( T_1 = p_1 T^{\text{total}}  \). The remaining time is split between Phase 0 and Phase 2. 
If we let \( \pi_0^{(2)} \) and \( \pi_2^{(2)} \) denote the relative time spent in Phases 0 and 2 within the non-calibration portion of the cycle (i.e., after Phase 1), we have $\pi_0^{(2)} = e^{-\rho}$ and $ \pi_2^{(2)} = 1 - e^{-\rho}$. Thus, the steady-state probabilities for the three phases can be given by:
\begin{equation}
    \begin{cases}
    \pi _1 = p_1,\\
	\pi _2=(1-p_1)\cdot \pi _{2}^{(2)}=(1-p_1)\cdot (1-e^{-\rho}),\\
	\pi _0=(1-p_1)\cdot \pi _{0}^{(2)}=(1-p_1)\cdot e^{-\rho}.\\
\end{cases}
\end{equation}
Therefore, the average communication cost $\overline {C}$ can then be calculated as:
\begin{equation}
    \overline {C} = \pi_0 C_0 + \pi_1 C_1 + \pi_2 C_2,
\end{equation}
where \( C_0 \), \( C_1 \), and \( C_2 \) represent the communication costs for the Idle, Calibration, and Streaming phases, respectively. Additionally, the values of the communication costs are determined by different features in associated phases. In Sec. \ref{sec:Performance Evaluation}, the features can be transmitted successfully only when $\overline {C}$ does not exceed the communication bottleneck.

\subsection{Data Timeliness Analysis}
\label{Sec: Data Timeliness Analysis}

In this section, we first derive the classical average Age of Information (AoI), which is not sufficient for evaluating the data timeliness of a multi-source system with multiple perceived targets. Therefore, we propose a new metric, namely the Age of Perceived Targets (AoPT), in a multi-camera collaborative perception system.

\textbf{1) Age of Information (AoI) for a single UGV}: Let \( \Delta_k \) be the average sampling interval of the \( k \)th UGV's camera, and \( d_k^{\text{total}} = d_k^T + d_k^I \) be the average total delay, where \( d_k^T \) is the average processing delay and \( d_k^I \) is the average transmission delay. For a multi-source system, the average AoI for the \( k \)th UGV is then given in Proposition \ref{proposition:individual_aoi}.

\begin{proposition}
\label{proposition:individual_aoi}
The AoI for UGV \( k \) under deterministic sampling and transmission delays is given by:
\begin{equation}
\Delta_{\text{AoI},k}=\frac{\Delta _k}{2}+d_{k}^{T}+d_{k}^{I},
\end{equation}
where \( d_k^T = \frac{D}{C_k} = D \left[ B_k \log_2\left( 1 + \frac{P_t G_k}{N_0 B_k} \right) \right]^{-1} \),\( B_k \) represents the bandwidth allocated to the link between UGV \( k \) and the edge server, \( P_t \) is the transmission power, \( G_k \) is the channel gain for UGV \( k \), \( N_0 \) is the noise power spectral density, and \( D \) is the data packet size.
\end{proposition}

\begin{Proof}
The AoI at time \( t \) for UGV \( k \), denoted as \( \Delta_{\text{AoI},k}(t) \), increases linearly between updates and resets to the transmission delay \( d_k^{\text{total}} \) upon each update. Given the average sampling interval \( \Delta_k \), the AoI is \( \Delta_{\text{AoI},k} = \frac{1}{\Delta_k} \int_{t_{n-1}}^{t_n} \Delta_{\text{AoI},k}(t) \, dt \), where \( t_{n-1} \) is the time of the \( (n-1) \)th update, and \( t_n \) is the time of the \( n \)th update. Substituting \( \Delta_{\text{AoI},k}(t) = t - t_{n-1} + d_k^{\text{total}} \), we get:
\begin{equation}\label{eq: aoi_2}
\begin{aligned}
\Delta_{\text{AoI},k} &= \frac{1}{\Delta_k} \int_{t_{n-1}}^{t_n} \left( t - t_{n-1} + d_k^{\text{total}} \right) dt= \frac{1}{\Delta_k} \left[ \frac{\Delta_k^2}{2} + d_k^{\text{total}} \Delta_k \right].
\end{aligned}
\end{equation}
According to Eq. (\ref{eq:total_delay}), we have $\Delta_{\text{AoI},k} = \frac{\Delta_k}{2}  + d_k^T + d_k^I $, where \( \Delta_k \), \( d_k^T \), and \( d_k^I \) are time-averaged values representing the average sampling interval, average processing delay, and average transmission delay, respectively.
{\hfill $\blacksquare$\par}
\end{Proof}

\begin{table}[t]
\centering
\caption{{\color{black}{Age metrics: mathematical definition and key feature captured}}}
\label{tab:aoi_aopt_cmp}
\renewcommand{\arraystretch}{1.15}
\begin{tabular}{|c|c|c|}
\hline
{\color{black}{\textbf{Metric}}} & {\color{black}{\textbf{Definition}}} & {\color{black}{\textbf{Key Feature}}}\\
\hline
{\color{black}{AoI}} &
{\color{black}{$\displaystyle
\Delta_{\mathrm{AoI},k}
=\frac{\Delta_k}{2}+d_k^{T}+d_k^{I}$}} &
{\color{black}{Vanilla freshness}} \\
\hline
{\color{black}{AoII}} &
{\color{black}{$\displaystyle
\Delta_{\mathrm{AoII},k}
=\bigl(\tfrac{\Delta_k}{2}+d_k^{\text{total}}\bigr)\,
\Pr\{\hat{X}_k\!\neq\!X_k\}$}} &
\begin{tabular}[c]{@{}c@{}}{\color{black}{Freshness weighted by}}\\ {\color{black}{correctness penalty}}\end{tabular} \\
\hline
{\color{black}{AoPT}} &
{\color{black}{$\displaystyle
\Delta_{\mathrm{AoPT},k}^{\mathrm{st}}
=\mathbbm{1}_{\{g_k\ge\varepsilon_g\}}\,
g_k\!\bigl(\tfrac{\Delta_k}{2}+d_k^{\text{total}}\bigr)$}} &
\begin{tabular}[c]{@{}c@{}}{\color{black}{Freshness weighted by}}\\ {\color{black}{target relevance}}\end{tabular} \\
\hline
\end{tabular}
\vspace{-3mm}
\end{table}

{\color{black}{While the AoI effectively quantifies data timeliness in traditional sensor networks, it exhibits significant limitations when applied to multi-view collaborative perception systems:}}
\begin{itemize}
    \item[1)] {\color{black}{\textbf{AoI assumes uniform sensor contributions and identical fields of view (FoVs).} In multi-UGV systems, each camera covers different areas and contributes unevenly to global perception. Updates from less critical views are treated equally, leading to inefficiencies.}}
    \item[2)] {\color{black}{\textbf{AoI neglects perception quality, such as target visibility or occlusions.} A timely but low-quality update may reset AoI while providing little meaningful information.}}
\end{itemize}

{\color{black}{Moreover, although the Age of Incorrect Information (AoII)~\cite{aoii,10143537} extends AoI by penalizing incorrect updates, it still presents limitations in collaborative perception:}}
\begin{itemize}
    \item[1)] {\color{black}{\textbf{AoII emphasizes correctness over perceptual value.} It targets estimation errors but cannot distinguish semantically uninformative frames from valuable observations.}}
    \item[2)] {\color{black}{\textbf{AoII lacks task-driven prioritization.} It does not account for the number or relevance of perceived targets, which are crucial in multi-view perception.}}
\end{itemize}

{\color{black}{To overcome these deficiencies, we propose the \emph{Age of Perceived Targets (AoPT)}, which integrates both freshness and perceptual relevance by weighting updates according to detected target counts.}}
{\color{black}{The differences between AoI, AoII, and AoPT are summarized in Table~\ref{tab:aoi_aopt_cmp}.}}

\vspace{2mm}

\textbf{2) Definition of AoPT}: The Age of Perceived Targets (AoPT) quantifies the freshness and relevance of perception data from each UGV. Specifically, the AoPT of the \( k \)th UGV is defined as  $\Delta_{\mathrm{AoPT},k}^{\mathrm{st}} = \mathbbm{1}_{\left\{ g_k \geq \boldsymbol{\varepsilon}_{\boldsymbol{g}} \right\}} \left[ g_k \cdot \left(\frac{ \Delta_k }{2}+ d_k^{\text{total}}\right) \right]$, where \( X_t^{(k)} \) denotes the perception data frame of the \( k \)th UGV at time \( t \), \( g_k(\cdot) \) is the object recognition network outputting the number of objects in a frame\footnote{To minimize computational costs, we calculate \( g_k(\cdot) \) only at regular time intervals since the count of targets remains constant within small time slot $\tau$. Moreover, we assume that $\tau$ is comparatively longer than both $\Delta_k$ and $\Delta_T$.}, \( \Delta_k \) represents the sampling interval, \( d_k^I \) is the inference delay, and \( \boldsymbol{\varepsilon}_{\boldsymbol{g}} \) is a threshold filtering out low-quality data. This equation accounts for both the data freshness and its informational value based on target count. For simplicity, we abbreviate \( g_{{k}} = g_{{k}}\left( X_t^{({k})} \right) \). As illustrated in Fig.~\ref{fig:diff_penalty}(a), the AoI function increases linearly between updates and resets upon receiving new data. Fig.~\ref{fig:diff_penalty}(b) shows how \( g_k \) varies over time. Frames with \( g_k < \boldsymbol{\varepsilon}_{\boldsymbol{g}} \) are discarded, while those with \( g_k \geq \boldsymbol{\varepsilon}_{\boldsymbol{g}} \) are retained, contributing to the AoPT based on target count and motion dynamics, as depicted in Fig.~\ref{fig:diff_penalty}(c). In a multi-UGV collaborative perception system, it is crucial to consider the worst-case AoPT to ensure no UGV significantly lags behind. Therefore, the AoPT during the streaming phase is formulated by taking the supremum over all UGVs:
\begin{equation}\label{eq:AoPT1}
\Delta_{\mathrm{AoPT}}^{\mathrm{st}}  = \sup_{k \in \mathcal{K}} \left\{ \mathbbm{1}_{\left\{ g_k \geq \boldsymbol{\varepsilon}_{\boldsymbol{g}} \right\}} \left[ g_k \cdot \left(\frac{ \Delta_k}{2}  + d_k^{\text{total}}\right) \right] \right\} ,
\end{equation}
where \( d_k^{\text{total}} \) includes all delays such as inference and transmission. This expression reflects the system's aim to prioritize the freshest and most informative data by optimizing the worst-case AoPT scenario.
\begin{figure}[t]
  \centering
  \includegraphics[width=0.47\textwidth]{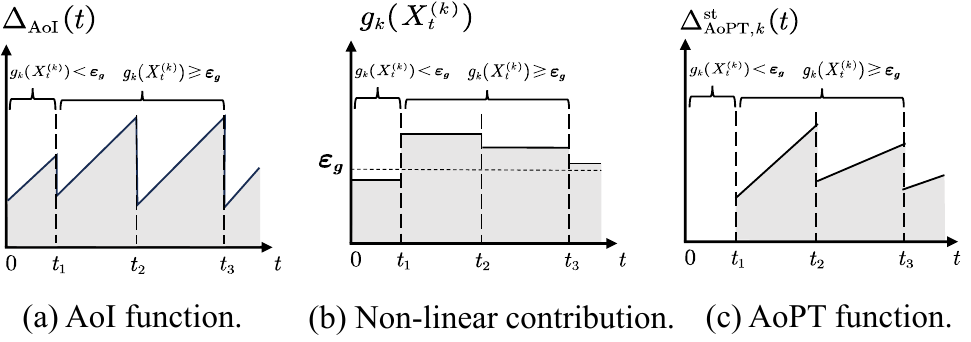}
  \caption{Illustrations of different age-based functions.}
  \label{fig:diff_penalty}
  \vspace{-3mm}
\end{figure}

\textbf{3) Impact of Calibration on AoPT}: Calibration introduces delays due to data suspension. With probability \( p_1 \), the system enters the Calibration phase (Phase 1), which increases the AoPT by a fixed duration \( T_1 \). To derive the AoPT during the Calibration phase, \( \Delta_{\mathrm{AoPT}}^{\mathrm{ca}} \), we first compute the average AoI \( \Delta_{k}^{\text{Ca}} \) during this phase. Within one calibration phase at times \( t_{n-1} \) and \( t_n \), the data update interval is the sum of the calibration time \( \Delta_T \) and the total delay \( d_{k}^{\text{total}} \). Therefore, we have $t_n - t_{n-1} = \Delta_T + d_{k}^{\text{total}}$.
The average AoI during the Calibration phase for device \( k \) is calculated as:
\begin{equation}\label{eq: aoi_calibration}
\begin{aligned}
\Delta _{k}^{\mathrm{Ca}}& =\frac{1}{t_n-t_{n-1}}\int_{t_{n-1}}^{t_n}{\left( t-t_{n-1}+d_{k}^{\mathrm{total}} \right)}dt
\\
% &= \frac{1}{\Delta _T+d_{k}^{\mathrm{total}}}\int_0^{\Delta _T+d_{k}^{\mathrm{total}}}{\left( \tau +d_{k}^{\mathrm{total}} \right)}d\tau \\
&= \frac{1}{\Delta _T+d_{k}^{\mathrm{total}}}\left[ \frac{\left( \Delta _T+d_{k}^{\mathrm{total}} \right) ^2}{2}+d_{k}^{\mathrm{total}}\left( \Delta _T+d_{k}^{\mathrm{total}} \right) \right]
\\
&= \frac{\Delta _T+3d_{k}^{\mathrm{total}}}{2}.
\end{aligned}
\end{equation}

Therefore, the AoPT during the Calibration phase is then expressed as:
\begin{equation}\label{eq:aopt_ca}
\begin{aligned}
\Delta_{\mathrm{AoPT}}^{\mathrm{ca}}
&=\mathbbm{1}_{\left\{ g_{{k}} \geqslant \boldsymbol{\varepsilon}_{\boldsymbol{g}} \right\}} \left[ g_{{k}} \cdot \Delta_{{k}}^{\text{Ca}} \right]  
&= \frac{1}{2}  g_{\widehat{k}}\cdot \left( 3d_{\widehat{k}}^{\text{total}} + \Delta_T \right) ,
\end{aligned}
\end{equation}
where \( \widehat{k} = \mathop{\arg\max}\limits_{k \in \mathcal{K}} \left\{ \mathbbm{1}_{\left\{ g_k\geqslant \boldsymbol{\varepsilon}_{\boldsymbol{g}} \right\}} \left[ g_k\left( X_t^{(k)} \right) \cdot \left( \frac{ \Delta_k}{2}  + d_k^{\text{total}} \right) \right] \right\} \) replaces the original indicate function $\mathbbm{1}_{\left\{ g_k\geqslant \boldsymbol{\varepsilon}_{\boldsymbol{g}} \right\}}$ in the following sections. 

\begin{definition}
\label{definition:calibration_aopt}
Combing Eq. (\ref{eq:AoPT1}) and Eq. (\ref{eq:aopt_ca}), AoPT over the entire cycle (including three phases) is formulated as :
\begin{equation}\label{eq:aopt_cy}
\begin{aligned}
\Delta_{\mathrm{AoPT}}^{\mathrm{cy}} &= p_1 \Delta_{\mathrm{AoPT}}^{\mathrm{ca}} + (1 - p_1) \Delta_{\mathrm{AoPT}}^{\mathrm{st}}\\
&=\frac{g_{\widehat{k}}}{2}  \left[ p_1 \Delta_T + (1 - p_1) \Delta_{\widehat{k}} + (p_1 + 2) d_{\widehat{k}}^{\text{total}} \right],
\end{aligned}
\end{equation}
where \( p_1 \) is the probability that the system enters the Calibration phase. 
\end{definition}

According to Definition \ref{definition:calibration_aopt}, the AoPT over the entire cycle is a weighted sum of the calibration and non-calibration phases.

% % \textbf{Global Monitoring and Random Target Monitoring (Tasks 1 and 2)}: 

% For global monitoring tasks (Task 1), where all cameras periodically transmit data, the AoPT is modeled as a multi-source system. Each camera's sampling interval and transmission delay contribute to the overall AoPT, as derived in Proposition \ref{proposition:aopt_definition}. For random target monitoring (Task 2), where target arrival follows a Poisson process with rate \(\lambda_i\), and the dwell time follows a LogNormal distribution \(S \sim \text{LogNormal}(\mu_S, \sigma_S^2)\), the AoPT can also be modeled similarly without the calibration phase. Both tasks lead to the same AoPT formulation.
\section{Problem Formulation}
\label{sec:Problem Formulation}

In this section, real-time multi-UGV system operate in three main phases: Idle (Phase 0), Calibration (Phase 1), and Streaming (Phase 2). The goal of optimizing these networks is to minimize AoPT during the entire cycle, ensuring that the system provides the fresh data for target perception. The objective is to reduce AoPT across all UGVs, which directly affects the real-time accuracy of multi-target detection. The optimization problem can be expressed as follows:
{\color{black}
\begin{equation}\label{OP:AoPT_minimization}
\begin{aligned}
  \mathbf{P}_1:\ & \min_{\left\{ \boldsymbol{B}, \boldsymbol{\Delta}, \Delta_T, \boldsymbol{D}, \mathbf{\Theta } \right\}} \Delta _{\mathrm{AoPT}}^{\mathrm{cy}} \\
  \quad \text{s.t.} \quad
  &(\ref{OP:AoPT_minimization}\mathrm{a})\quad \boldsymbol{\gamma}_{\mathrm{Ca}}(\mathbf{\Theta}_{\mathrm{Ca}})\succeq\gamma_{\mathrm{Ca},0}\mathbf{1}, \\
  &(\ref{OP:AoPT_minimization}\mathrm{b})\quad \boldsymbol{\gamma}_{\mathrm{St}}(\mathbf{\Theta}_{\mathrm{St}})\succeq\gamma_{\mathrm{St},0}\mathbf{1}, \\
  &(\ref{OP:AoPT_minimization}\mathrm{c})\quad \boldsymbol{B}_{\min} \preceq \boldsymbol{B} \preceq \boldsymbol{B}_{\max}, \\
  &(\ref{OP:AoPT_minimization}\mathrm{d})\quad \boldsymbol{\Delta}_{\min} \preceq \boldsymbol{\Delta} \preceq \boldsymbol{\Delta}_{\max}, \\
  &(\ref{OP:AoPT_minimization}\mathrm{e})\quad \boldsymbol{\Delta}_{T,\min} \preceq \boldsymbol{\Delta}_T \preceq \boldsymbol{\Delta}_{T,\max}, \\
  &(\ref{OP:AoPT_minimization}\mathrm{f})\quad \boldsymbol{D}_{\min} \preceq \boldsymbol{D} \preceq \boldsymbol{D}_{\max}, \\
\end{aligned}
\end{equation}
where \(\boldsymbol{B}\), \(\boldsymbol{\Delta}\), \( \boldsymbol{\Delta}_T\) and \(\boldsymbol{D}\) are K-dimensional vectors corresponding to the bandwidth, sampling intervals, calibration intervals, and data packet sizes, respectively, for each UGV \(k \in \mathcal{K}\). \(\mathbf{1}\) is an all-ones vector.} Let $\mathbf{\Theta }=\left[ \mathbf{\Theta }_{\mathrm{Ca}},\mathbf{\Theta }_{\mathrm{St}} \right]$ be the set of model parameters, where $\mathbf{\Theta }_{\mathrm{Ca}}$ represents the model parameters for feature extraction in the calibration phase, and $\mathbf{\Theta }_{\mathrm{St}}$ denotes the parameters for task-specific feature extraction in the streaming generation. Ineqs. (\ref{OP:AoPT_minimization}a) and (\ref{OP:AoPT_minimization}b) are the constraint on calibration and streaming task accuracy, respectively. Ineq. (\ref{OP:AoPT_minimization}c) is the bandwidth constraint, while Ineqs. (\ref{OP:AoPT_minimization}d), Ineq. (\ref{OP:AoPT_minimization}e) and (\ref{OP:AoPT_minimization}f) are the constraints on the sampling intervals, calibration intervals and data packet size, respectively. According to Proposition \ref{proposition:decomposition}, the original optimization problem (\ref{OP:AoPT_minimization}) can be decomposed into two subproblems in
 the following subsections.

\begin{proposition}
\label{proposition:decomposition}
The original problem of minimizing AoPT $\mathbf{P}_1$ can be decomposed into two subproblems $\mathbf{P}_2$ and $\mathbf{P}_3$, corresponding to the calibration phase and the streaming phase, respectively. $\mathbf{P}_1$ can be solved independently by different phases.
\end{proposition}

\begin{proof}
According to Eq.~(\ref{eq:aopt_cy}), the original optimization problem $\mathbf{P}_1$ aims to minimize the AoPT over the entire cycle is given by $\Delta_{\mathrm{AoPT}}^{\mathrm{cy}} = p_1 \Delta_{\mathrm{AoPT}}^{\mathrm{ca}} + (1 - p_1) \Delta_{\mathrm{AoPT}}^{\mathrm{st}}$, where
\begin{subequations}\label{eq:aopt_final}
\begin{align}
	\Delta _{\text{AoPT}}^{\text{ca}} &= \frac{1}{2}\,g_{\widehat{k}}\left( \Delta _T + 3d_{\widehat{k}}^{\text{total}} \right), \label{eq:aopt_ca_final}\\
	\Delta _{\text{AoPT}}^{\text{st}} &= g_{\widehat{k}}\left( \frac{\Delta _{\widehat{k}}}{2} + d_{\widehat{k}}^{\text{total}} \right). \label{eq:aopt_st_final}
\end{align}
\end{subequations}
It is noted that $\widehat{k}\;=\;\arg\max_{k\in\mathcal{K}}
        \Bigl\{
            \mathbbm{1}_{\{g_k\ge\varepsilon_g\}}\,
            g_k\Bigl(\tfrac{\Delta_k}{2}+d_k^{\text{total}}\Bigr)
        \Bigr\}$,
which denotes the ``bottleneck'' UGV that dominates the AoPT. By substituting Eqs.~(\ref{eq:aopt_ca_final}) and (\ref{eq:aopt_st_final}) into $\Delta_{\mathrm{AoPT}}^{\mathrm{cy}} $, we obtain:
\begin{small}
\begin{equation}
\begin{aligned}\label{eq:aopt_cy_2}
\Delta_{\mathrm{AoPT}}^{\mathrm{cy}} &= p_1 \left[ \frac{1}{2} \, g_{\widehat{k}} \left( \Delta_T + 3 d_{\widehat{k}}^{\mathrm{total}} \right) \right] + (1 - p_1) \left[ g_{\widehat{k}} \left( \frac{ \Delta_{\widehat{k}} }{2} + d_{\widehat{k}}^{\mathrm{total}} \right) \right] \\
&= g_{\widehat{k}} \left[ \frac{p_1}{2} \Delta_T + \left( \frac{3 p_1}{2} + 1 - p_1 \right) d_{\widehat{k}}^{\mathrm{total}} + \frac{ (1 - p_1) }{2} \Delta_{\widehat{k}} \right] \\
&= g_{\widehat{k}}\Bigg[ \underset{\text{Calibration\ rate}}{\underbrace{\frac{p_1}{2}\Delta _T}}+\underset{\text{Transmission\ performance}}{\underbrace{\left( \frac{p_1}{2}+1 \right) d_{\widehat{k}}^{\text{total}}+\frac{\left( 1-p_1 \right)}{2}\Delta _{\widehat{k}}}} \Bigg] .
\end{aligned}
\end{equation}
\end{small}
We observe that the variables $\Delta_T$ in the calibration phase are independent of the variables $\Delta_{\widehat{k}}$ and $d_{\widehat{k}}^{\mathrm{total}}$ in the streaming phase. Furthermore, the resource allocations in each phase are independent. Therefore, the original problem $\mathbf{P}_1$ can be decomposed into two independent subproblems $\mathbf{P}_2$ and $\mathbf{P}_3$ regarding \textit{calibration rate} and \textit{transmission performance}, respectively. 

1) For the calibration phase, we integrate the calibration rate in Eq. (\ref{eq:aopt_cy_2}) and the constraint of accuracy in Ineqs. (\ref{OP:AoPT_minimization}a) using Lagrangian multipliers. Therefore, $\mathbf{P}_2$ is given by:
\begin{equation}
\begin{aligned}\label{OP:calibration_P2}
&\mathbf{P}_2:\  \min_{\left\{ B_{k},\Delta _T,D_{k},\mathbf{\Theta }_{\mathrm{Ca}} \right\} } \underset{\text{Calibration rate}}{\underbrace{{p_1}  \cdot \Delta _T}} \ -\ \underset{\text{Calibration accuracy}}{\underbrace{\lambda _{\mathrm{Ca}}\left( \boldsymbol{\gamma }_{\mathrm{Ca,k}}(\mathbf{\Theta }_{\mathrm{Ca}})-\gamma _{\mathrm{Ca},0} \right) }}, \\
&\text{s.t.} \quad
B_{\min}\leq B_{k}\leq B_{\max}, \quad D_{\min}\leq D_{k}\leq D_{\max}, \\
&\quad \quad \max \left( \Delta _{T,\min},\frac{D_{k}}{C_{k}} \right) \leqslant \Delta _{T} \leqslant \Delta _{T,\max},
\end{aligned}
\end{equation}
where $\lambda_{\mathrm{Ca}}$ represents the weight of calibration accuracy (Lagrange weight), and $\boldsymbol{\gamma }_{\mathrm{Ca}}(\mathbf{\Theta }_{\mathrm{Ca}})$ denotes the calibration accuracy function. {\color{black}The final constraint in (\ref{OP:calibration_P2}) ensures that the calibration interval \(\Delta_T\) is long enough to accommodate both data transmission and feature extraction, while remaining within an upper limit to preserve freshness. Specifically, \(D_k/C_k\) represents the minimal time needed to transmit a calibration packet of size \(D_k\) over a link of capacity \(C_k = B_k \log_2(1+\mathrm{SNR})\), and \(\Delta_{T,\min}\) accounts for the minimum required processing time; thus, the lower bound in (\ref{OP:calibration_P2}) reflects the maximum of these two factors. The upper bound \(\Delta_{T,\max}\) prevents excessive delays that would degrade calibration quality and perception freshness. $\mathbf{P}_2$ is solved for each \(k\in\mathcal{K}\); afterwards the worst-case index \(\widehat{k}\) is identified and substituted back into Eq. (\ref{eq:aopt_cy_2}).} In order to solve the problem $\mathbf{P}_2$, we need to design efficient feature matching algorithm in Sec. \ref{sec:m-Calibration}.

2) For the streaming phase, we integrate the transmission performance in Eq. (\ref{eq:aopt_cy_2}) and the constraint of accuracy in Ineqs. (\ref{OP:AoPT_minimization}b) using Lagrangian multipliers. Therefore, the subproblem $\mathbf{P}_3$ is given by:
\begin{equation}
\begin{aligned}\label{OP:transmission}
&\mathbf{P}_3:\ \min_{\left\{ B_{\widehat{k}},D_{\widehat{k}},\Delta _{\widehat{k}},\mathbf{\Theta }_{\mathrm{St}} \right\}} \,\,\underset{\mathrm{Transmission}\ \mathrm{performance}}{\underbrace{\left( \frac{p_1}{2} + 1 \right) d_{\widehat{k}}^{\mathrm{total}} + \frac{1 - p_1}{2}  \Delta_{\widehat{k}}}}\\
& \quad \quad \quad \quad \quad \quad \quad \quad \underset{\mathrm{Inference\ } \mathrm{performance}}{-\underbrace{{\lambda }_{\widehat{k}}\left( \boldsymbol{\gamma }_{\mathrm{St},\widehat{k}}(\mathbf{\Theta }_{\mathrm{St}})-\boldsymbol{\gamma }_{\mathrm{MOD},0} \right) }},\\
  \quad &\text{s.t.} \quad
   B_{\min}\leq B_{\widehat{k}}\leq B_{\max}\ ,
  \ D_{\min}\leq D_{\widehat{k}}\leq D_{\max}, \\
  & \quad \quad \max \left( \Delta _{\min},D_{\widehat{k}}\cdot C_{k}^{-1} \right) \leqslant \Delta _{\widehat{k}}\leqslant \Delta _{\max},
  \end{aligned}
\end{equation}
where ${\lambda }_{\widehat{k}}$ is the weight balancing transmission and inference performance\footnote{{\color{black}\(\lambda_{\mathrm{Ca}}\) and \(\lambda_{\hat{k}}\) are manually adapted according to the prevailing channel conditions. Under good channel quality (high \(C_{\widehat{k}}\)), a smaller \(\lambda\) is used to prioritize faster updates, whereas under poor channel quality (low \(C_{\widehat{k}}\)), a larger \(\lambda\) is selected to ensure sufficient accuracy despite limited transmission resources.}}. Since $\mathbf{P}_3$ aims to balance the trade-off between transmission budget and inference performance, we can address this subproblem by utilizing an IB-based theoretical framework, which offers the task-specific feature encoder/decoder as discussed in Section \ref{sec: Task-Oriented Encoding with Information Bottleneck}.
% Thus, the original problem $\mathbf{P}_1$ is decomposed into two subproblems $\mathbf{P}_2$ and $\mathbf{P}_3$, corresponding to the calibration and streaming phases, respectively. Solving $\mathbf{P}_1$ is equivalent to independently solving these subproblems, as the variables and constraints are independent across the phases. 
{\hfill $\blacksquare$}
\end{proof}

\section{Methodology}
\label{sec:Methodology} 
\begin{figure}[t]
  \centering
  \includegraphics[width=0.48\textwidth]{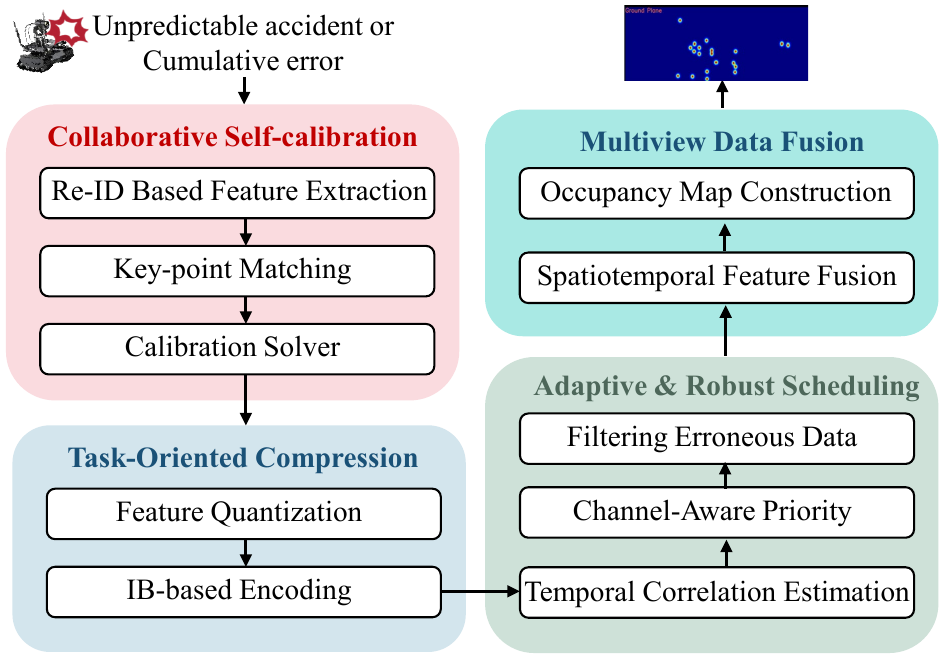}
  \caption{Framework of R-ACP.}
  \label{fig:R-ACP_flow}
  \vspace{-3mm}
\end{figure}

This section elaborates R-ACP's design: 1) Collaborative Self-calibration: R-ACP uses Re-ID to share perception data for real-time extrinsic calibration. 2) Task-Oriented Compression: After calibration, visual features are compressed for pedestrian tracking and Re-ID. 3) Adaptive \& Robust Scheduling: Features are further compressed by temporal correlation. Considering the varied packet loss rate, we calculate the channel-aware priorities and filter out erroneous data. 4) Multiview Data Fusion: It generates the occupancy map. The framework of R-ACP is shown in Fig.~\ref{fig:R-ACP_flow}.

% \begin{figure}[t]
%   \centering
%   \includegraphics[width=0.48\textwidth]{figure/fig_match_comparison.pdf}
%   \caption{Comparison of different feature matching techniques for camera calibration. (a) Facial recognition match: Limited matching points. (b) SIFT-based match: High error rate. (c) Re-ID aided match: Higher correct match rate.}
%   \label{fig:match_comparison}
%   \vspace{-3mm}
% \end{figure}
\subsection{Collaborative Self-calibration}
\label{sec:m-Calibration}

This section addresses subproblem \(\mathbf{P}_2\) from Sec. \ref{sec:Problem Formulation}, focusing on minimizing calibration transmission rate while maximizing calibration accuracy through Re-ID based method.
\subsubsection{Comparative Analysis of Various Matching Techniques}  
\begin{figure}[t]
  \centering
  \subfigure[Rotation error vs comm. cost.]{
    \includegraphics[width=4.3 cm]{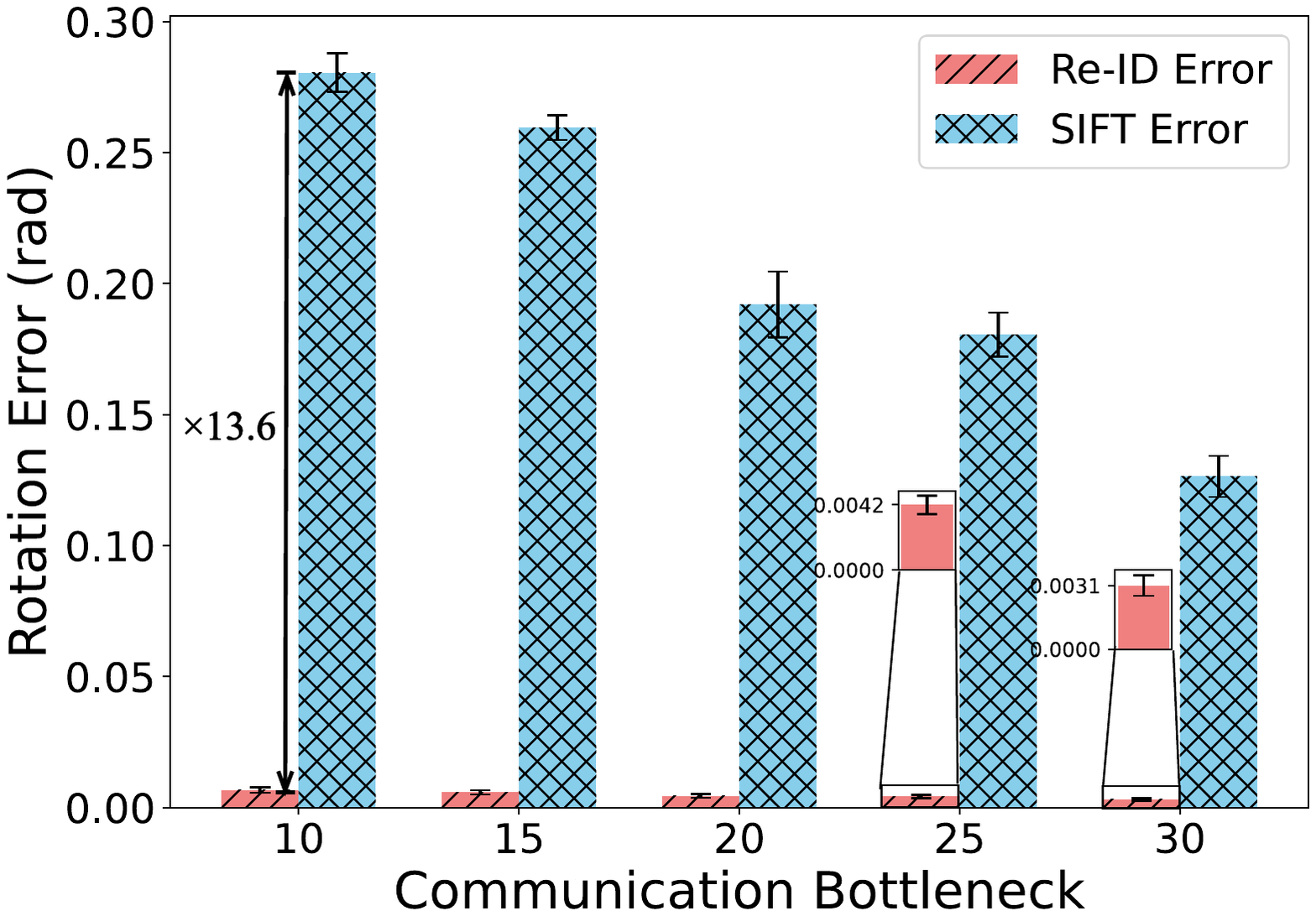}
    \label{fig:rotation_error_vs_cost}
  }
  \hspace{-0.5cm}  % 调整子图间距
  \subfigure[Translation error vs comm. cost.]{
    \includegraphics[width=4.3 cm]{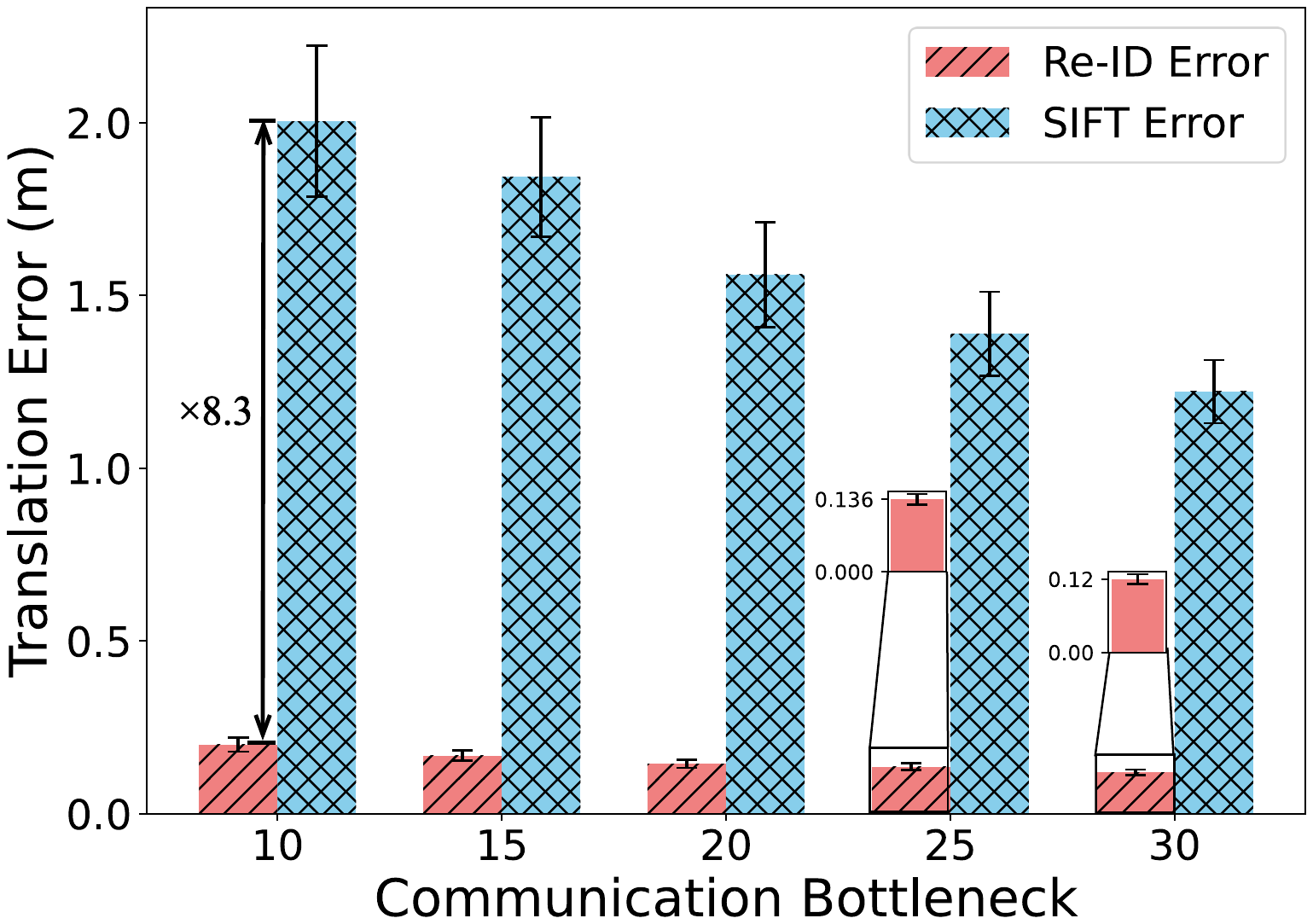}
    \label{fig:translation_error_vs_cost}
  }
  \caption{Calibration accuracy vs communication bottleneck for different errors. Fig. \ref{fig:rotation_error_vs_cost} shows the rotation error, while \mbox{Fig. \ref{fig:translation_error_vs_cost}} demonstrates the translation error.}
  \label{fig:calibration_accuracy_vs_communication_bottleneck}
  \vspace{-3mm}
\end{figure}

\begin{figure}[t]
  \centering
  \subfigure[Calibration errors using Re-ID.]{
    \includegraphics[width=4.2 cm]{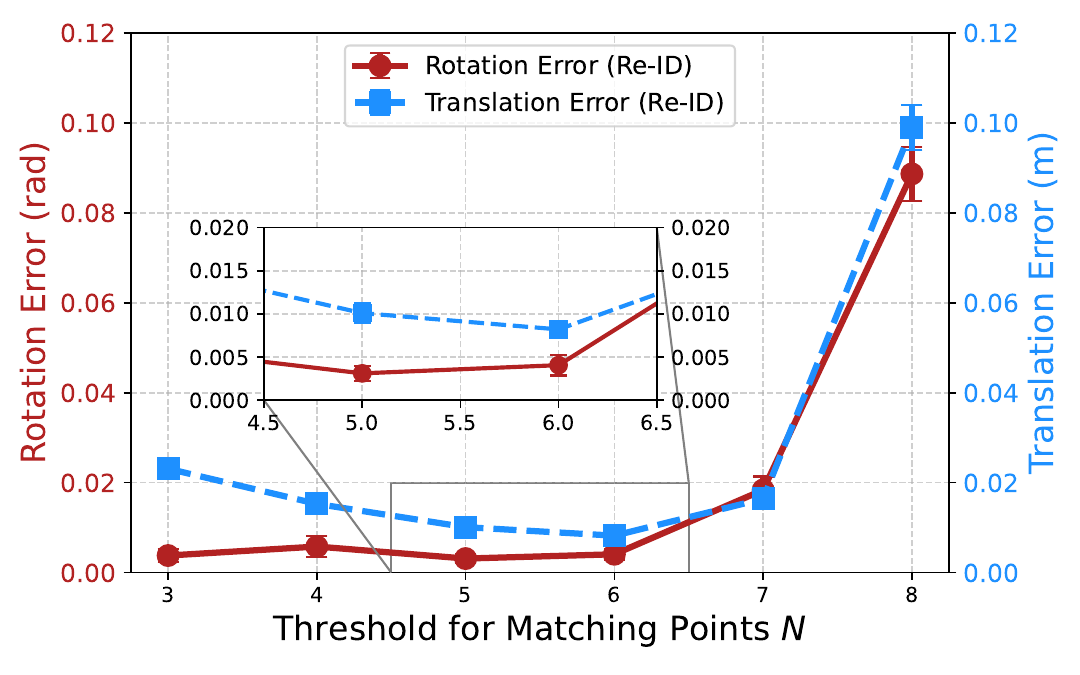}
    \label{fig:reid_errors}
  }
  \hspace{-0.5cm}  % 调整子图间距
  \subfigure[Calibration errors using SIFT.]{
    \includegraphics[width=4.2 cm]{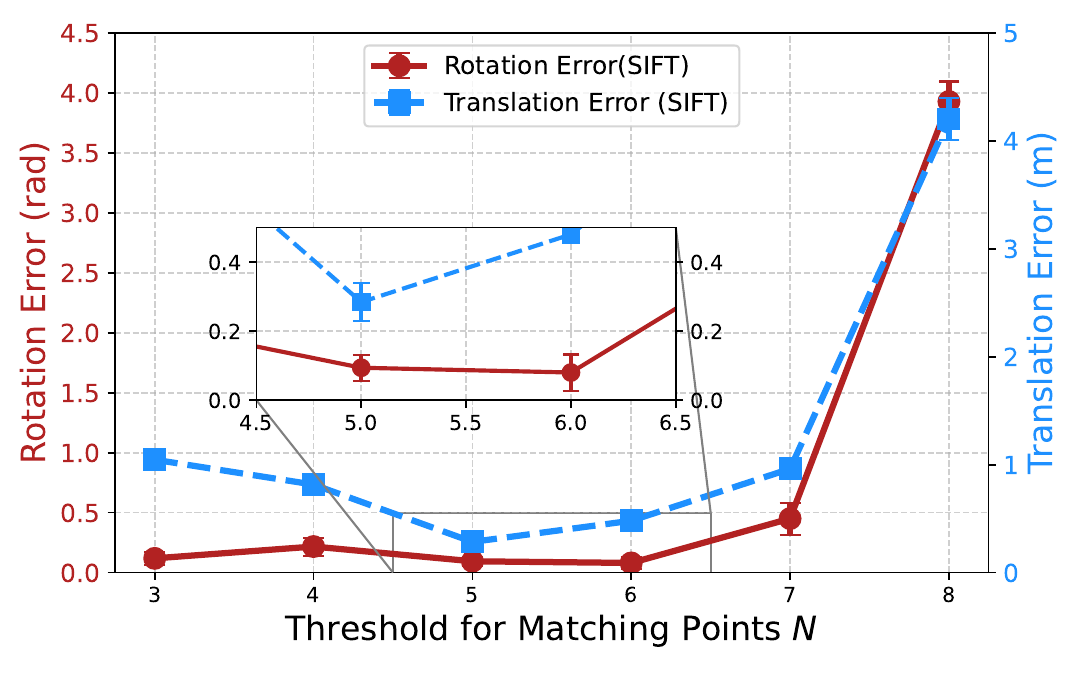}
    \label{fig:sift_errors}
  }
  \caption{Comparison of rotation and translation errors using Re-ID and SIFT methods for different thresholds $N$.}
  \label{fig:rotation_translation_comparison}
  \vspace{-3mm}
\end{figure}

As shown in Fig. \ref{fig:calibration_flow}, we need to choose a suitable algorithm to obtain the relevant key-points for the calibration solver.
Facial recognition extracts unique features to match individuals across camera views~\cite{du2022elements}. SIFT extracts distinctive, scale- and rotation-invariant key-points suitable for viewpoint matching~\cite{lowe2004distinctive}. Both are common in camera calibration. We compare them with our Re-ID-based approach, evaluating calibration accuracy using the extrinsic error \( e_{\text{extrinsic}} \) between the recalibrated extrinsic matrix \( [\mathbf{R}_{\text{rec}} | \mathbf{t}_{\text{rec}}] \) and the ground truth \( [\mathbf{R}_{\text{gt}} | \mathbf{t}_{\text{gt}}] \) by
$e_{\text{extrinsic}} = \frac{\| [\mathbf{R}_{\text{rec}} | \mathbf{t}_{\text{rec}}] - [\mathbf{R}_{\text{gt}} | \mathbf{t}_{\text{gt}}] \|_F}{\| [\mathbf{R}_{\text{gt}} | \mathbf{t}_{\text{gt}}] \|_F} \times 100\%.$

Experiments show facial recognition performs poorly, with high data transmission (856MB for two matches) and a 24.3\% extrinsic error. SIFT yields more matches but has higher errors (42.5\%). Re-ID, leveraging pedestrian attributes, reduces extrinsic error to under 0.6\%. We evaluate Re-ID and SIFT under varying communication constraints. In Fig.~\ref{fig:calibration_accuracy_vs_communication_bottleneck}, errors decrease as communication constraints ease. Re-ID consistently outperforms SIFT, with rotation error 13.6 times lower and translation error 8.3 times lower. Fig.~\ref{fig:rotation_translation_comparison} shows a matching threshold of 5 key-points yields optimal calibration, while incorrect matches reduce performance.

\subsubsection{Feature Extraction and Similarity Matching}  
\label{sec:Feature Extraction and Similarity Matching}

To improve the calibration accuracy in \(\mathbf{P}_2\), Re-ID treats pedestrians as key-points. Given bounding boxes \( \mathcal{B}_k = \{b_1, b_2, \dots, b_N\} \) from a pretrained model (YOLOV5~\cite{9780215}), the Re-ID network \( \mathcal{F}(\cdot; \mathbf{\Theta}) \) extracts feature vectors \( \mathbf{f}_i = \mathcal{F}(b_i; \mathbf{\Theta}) \in \mathbb{R}^d \). Similarity between features across views \(i\) and \(j\) is computed via Euclidean distance: $d_{\text{sim}}\left(\mathbf{f}_i, \mathbf{f}_j\right) = \| \mathbf{f}_i - \mathbf{f}_j \|_2$. To optimize calibration, distances are ranked, and the top \(N\) matches are selected to minimize matching error:
\begin{equation}
\mathcal{M}_{ij} = \underset{M \subset \mathcal{B}_i \times \mathcal{B}_j}{\mathrm{argmin}} \sum_{(b_i, b_j) \in M} d_{\text{sim}}\left( \mathcal{F}\left(b_i\right), \mathcal{F}\left(b_j\right) \right),
\end{equation}
where $\mathcal{M}_{ij}$ represents key-point matching, focusing on the most similar features to minimize extrinsic error \( e_{\text{extrinsic}} \).

\subsubsection{Quantization and Communication Cost}  
\label{sec:Quantization}

To reduce calibration rate in \(\mathbf{P}_2\), feature vectors \( \mathbf{f}_i \) are quantized based on channel quality. The communication cost of calibration is $C_{1} = \sum_{i=1}^{N} q_i \cdot d$, where \( q_i \) is the quantization level and \( d \) is dimensionality. Quantization adapts to SNR to meet available capacity \( C_{\text{ava}} \)\cite{10368103}:
\begin{equation}
q_i = \arg\min_{q_i} \left\{ q_i \cdot d \mid C_{1} \leq C_{\text{ava}} \right\}.
\end{equation}
Higher SNR allows finer quantization; lower SNR uses coarser quantization to reduce overhead. Combining Re-ID and adaptive quantization, we optimize calibration accuracy and rates, efficiently addressing \(\mathbf{P}_2\) in multi-UGV systems.

\subsection{Task-Oriented Compression}
\label{sec: Task-Oriented Encoding with Information Bottleneck}

The Information Bottleneck (IB) method offers a principled framework to balance the compression of input camera data \(X\) and the task-relevant information for inference about \(Y\)\cite{fang2024prioritized}. The IB objective is formulated as:
\begin{equation}\label{eq:IB-def}
\max_{\Theta} \quad I(Z; Y) - \lambda I(X; Z),
\end{equation}
where \(I(Z; Y)\) is the mutual information between the compressed representation \(Z\) and the target variable \(Y\), capturing inference performance, and \(I(X; Z)\) is the mutual information between \(X\) and \(Z\), representing the amount of information retained for transmission. The parameter \(\lambda\) controls the trade-off between the compression and inference accuracy.

\begin{proposition}
\label{proposition: P3-is-IB}
The optimization problem \(\mathbf{P}_3\) is equivalent to the Information Bottleneck (IB) problem defined in Eq.~(\ref{eq:IB-def}); specifically, minimizing transmission delay and sampling interval corresponds to minimizing \(I(X; Z)\), while maximizing inference performance corresponds to maximizing \(I(Z; Y)\).
\end{proposition}

\begin{proof}
First, the total delay \(d_{\widehat{k}}^{\text{total}} = d_{\widehat{k}}^T + d_{\widehat{k}}^I\) includes the transmission delay \(d_{\widehat{k}}^T = \frac{D_{\widehat{k}}}{C_{\text{ava}}}\), where \(D_{\widehat{k}}\) is the data packet size and \(C_{\text{ava}}\) is the available channel capacity.
Since \(D_{\widehat{k}}\) is proportional to the entropy \(H(Z)\) of the transmitted feature \(Z\), and for deterministic encoders \(H(Z|X) = 0\), we have \(D_{\widehat{k}} = H(Z) = I(X; Z)\). Therefore, we have \(d_{\widehat{k}}^{\text{total}}\propto I(X; Z)\).
Given \(C_{\text{ava}}\), the sampling interval \(\Delta_{\widehat{k}}\) satisfies \(\Delta_{\widehat{k}} \geq \frac{D_{\widehat{k}}}{C_{\text{ava}}} = \frac{I(X; Z)}{C_{\text{ava}}}\), which implies \(\frac{I(X; Z)}{\Delta_{\widehat{k}}} \leq C_{\text{ava}}\), showing that \(\Delta_{\widehat{k}} \propto I(X; Z)\) if \(C_{\text{ava}}\) is fixed. Thus, minimizing \(d_{\widehat{k}}^{\text{total}}\) and \(\Delta_{\widehat{k}}\) under communication constraints corresponds to minimizing \(I(X; Z)\).
On the other hand, inference performance depends on how much information \(Z\) retains about \(Y\), quantified by \(I(Z; Y)\). Maximizing inference performance corresponds to maximizing \(I(Z; Y)\).
Therefore, the objective of \(\mathbf{P}_3\) can be reformulated as $\min \left[ \alpha \cdot I(X; Z) - \beta \cdot I(Z; Y) \right]$, where \(\alpha\) and \(\beta\) are positive constants derived from the weights and coefficients in \(\mathbf{P}_3\), which aligns with the IB problem defined in Eq.~(\ref{eq:IB-def}).
{\hfill $\blacksquare$\par}
\end{proof}

\subsubsection{Variational Approximation for IB}

Due to the computational complexity of directly estimating mutual information, we employ a variational approximation method to derive a lower bound for \(I(Z; Y)\) and an upper bound for \(I(X; Z)\). The variational approach is based on approximating the conditional distribution \(p(Y|Z)\) with a simpler distribution \(q(Y|Z)\), parameterized by \(\Theta_d\), and approximating \(p(Z|X)\) with \(q(Z|X)\), parameterized by \(\Theta_{con}\). Using the standard definition of mutual information, we start with:
\[
I(Z; Y) = \mathbb{E}_{p(Y,Z)} \left[ \log \frac{p(Y|Z)}{p(Y)} \right],
\]
and introduce the KL-divergence between \(p(Y|Z)\) and \(q(Y|Z)\):
\begin{equation}\label{eq:KL}
D_{KL} \left[ p(Y|Z) || q(Y|Z) \right] = \mathbb{E}_{p(Y|Z)} \left[ \log \frac{p(Y|Z)}{q(Y|Z)} \right] \geq 0.
\end{equation}
This leads to the inequality $\mathbb{E}_{p(Y|Z)} \left[ \log p(Y|Z) \right] \geq \mathbb{E}_{p(Y|Z)} \left[ \log q(Y|Z) \right]$. Therefore, the lower bound of mutual information:
\begin{equation}\label{eq: IZY}
I(Z; Y) \geq \mathbb{E}_{p(Y,Z)} [\log q(Y | Z)] + H(Y),
\end{equation}
where \(H(Y)\) is the entropy of \(Y\). For the second part \( I(X; Z) \), we derive an upper bound using variational approximations due to the complexity of directly minimizing the term. Since entropy \( H(Z | X) \geq 0 \), we can establish the following inequality $\lambda \sum_{k=1}^{K} I(X^{(k)}; Z^{(k)}) \leq \lambda \sum_{k=1}^{K} H(Z^{(k)})$, where \( H(Z^{(k)}) \) is the entropy of the compressed feature \( Z^{(k)} \). To refine this upper bound, we incorporate latent variables \( V^{(k)} \) as side information to encode the quantized features. Thus, we obtain:
\begin{equation}\label{eq:joint_entropy}
\lambda \sum_{k=1}^{K} H(Z^{(k)}) \leq \lambda \sum_{k=1}^{K} H(Z^{(k)}, V^{(k)}),
\end{equation}
where the joint entropy \( H(Z^{(k)}, V^{(k)}) \) represents the communication cost. Moreover, we apply the non-negativity property of KL-divergence to establish a tighter upper bound. The joint entropy \( H(Z^{(k)}, V^{(k)}) \) can be bounded using variational distributions \( q(Z^{(k)} | V^{(k)}; \Theta_{con}^{(k)}) \) and \( q(V^{(k)}; \Theta_{l}^{(k)}) \). Specifically, we have:
\begin{equation}\label{ineq:joint_entropy}
\begin{aligned}
H(Z^{(k)}, V^{(k)}) \leq & \mathbb{E}_{p(Z^{(k)}, V^{(k)})} \left[ -\log q(Z^{(k)}|V^{(k)}; \Theta_{con}^{(k)}) \right. \\
& \left. \times q\left(V^{(k)}; \Theta_{l}^{(k)}\right) \right].
\end{aligned}
\end{equation}
where \( \Theta_{con}^{(k)} \) and \( \Theta_{l}^{(k)} \) are the learnable parameters of the variational distributions \( q(Z^{(k)} | V^{(k)}) \) and \( q(V^{(k)}) \), respectively. These parameters are optimized to approximate the true distributions, minimizing the communication cost while preserving essential feature relations for inference. Substituting the result from Eq. (\ref{ineq:joint_entropy}) into Eq. (\ref{eq:joint_entropy}), we derive the final upper bound for \( I(X^{(k)}; Z^{(k)}) \) as follows:
\begin{equation}\label{ineq:final_bound}
\begin{aligned}
I\left(X^{(k)};Z^{(k)}\right) \leq & \mathbb{E}_{p(Z^{(k)}, V^{(k)})} \left[ -\log q\left(Z^{(k)}|V^{(k)}; \Theta_{con}^{(k)}\right) \right. \\
& \left. \times q(V^{(k)}; \Theta_{l}^{(k)}) \right].
\end{aligned}
\end{equation}

This upper bound simplifies the minimization process of \( I(X; Z) \), providing a feasible method for reducing transmission cost during network training.

\subsubsection{Loss Function Design}

We design the loss function to optimize the IB objective. The loss function \(\mathcal{L}_{1}\) is constructed to minimize the upper bound of \(I(X; Z)\) and maximize the lower bound of \(I(Z; Y)\):
\begin{equation}
\begin{aligned}\label{eq:L1}
\mathcal{L}_{1} &= \sum_{k=1}^{K} \underset{\mathrm{The}\ \mathrm{upper} \ \mathrm{bound} \ \mathrm{of} \ -I\left( Z^{(k)};Y^{(k)} \right)}{\underbrace{\mathbb{E} \left[ -\log q(Y^{(k)}|Z^{(k)}; \Theta_d^{(k)}) \right]} }+ \\
&\lambda \sum_{k=1}^{K} \underset{\mathrm{The} \ \mathrm{upper} \ \mathrm{bound} \ \mathrm{of} \ I\left( X^{(k)};Z^{(k)} \right)}{\underbrace{\mathbb{E} \left[ -\log q(Z^{(k)} | X^{(k)}; \Theta_{con}^{(k)}) \right]}}.
\end{aligned}
\end{equation}
This loss function optimizes both compression and inference accuracy by minimizing \(I(X; Z)\) while maximizing \(I(Z; Y)\), achieving an optimal balance between transmission cost and task performance.

\subsection{Adaptive and Robust Streaming Scheduling}
\label{sec:rate_control}

In this section, we develop an adaptive streaming scheduling framework to efficiently manage dynamic packet loss rates in multi-view environments. We first reduce temporal redundancy by estimating correlations across multiple frames and then implement a prioritization strategy for robust feature transmission under varying packet loss conditions.

\subsubsection{Correlation Estimation by Multiple Frames}

To optimize the transmission efficiency, we leverage the temporal dependencies between consecutive frames. The feature representation at time \( t \), denoted by \( \hat{z}_t^{(k)} \), is estimated using the preceding frames as side information. This estimation is modeled by the variational distribution \( q(\hat{z}_t^{(k)} | \hat{z}_{t-1}^{(k)}, \dots, \hat{z}_{t-\tau}^{(k)}; {\Theta}_p^{(k)}) \), where \( {\Theta}_p^{(k)} \) represents the parameters of the network in the $k$th UGV. 

We assume that this conditional distribution follows a Gaussian distribution \( q(\hat{z}_t^{(k)} | \hat{z}_{t-1}^{(k)}, \dots, \hat{z}_{t-\tau}^{(k)}; {\Theta}_p^{(k)}) = \mathcal{N}(\mu_t^{(k)}, \sigma_t^{(k)}) \ast \mathcal{U} \), where \( \mu_t^{(k)} \) and \( \sigma_t^{(k)} \) are the predicted mean and variance, respectively, and \( \mathcal{U} \) models the quantization noise added during transmission. By exploiting temporal correlations across frames, this model reduces the entropy of \( \hat{z}_t^{(k)} \), thereby minimizing the required transmission bitrate. To further reduce communication overhead, we align the predicted distribution \( q(\hat{z}_t^{(k)} | \hat{z}_{t-1}^{(k)}, \dots, \hat{z}_{t-\tau}^{(k)}; {\Theta}_p^{(k)}) \) with the true conditional distribution \( p(\hat{z}_t^{(k)} | \hat{z}_{t-1}^{(k)}, \dots, \hat{z}_{t-\tau}^{(k)}) \) by minimizing the cross-entropy loss between them. The loss function is given by:

\begin{small}
\begin{equation*}
\begin{aligned}\label{eq:L2}
\mathcal{L}_{2} = \sum_{k=1}^{K}\sum_{t=1}^{N} \left( p(\hat{z}_t^{(k)} | \hat{z}_{t-1}^{(k)}, \dots, \hat{z}_{t-\tau}^{(k)}) \log \frac{p(\hat{z}_t^{(k)} | \hat{z}_{t-1}^{(k)}, \dots, \hat{z}_{t-\tau}^{(k)})}{q(\hat{z}_t^{(k)} | \hat{z}_{t-1}^{(k)}, \dots, \hat{z}_{t-\tau}^{(k)}; {\Theta}_p^{(k)})} \right),
\end{aligned}
\end{equation*}
\end{small}
which quantifies the KL divergence between the true distribution and the variational approximation. By minimizing this divergence, we exploit the temporal redundancy, thereby reducing the amount of data that needs to be transmitted.

\subsubsection{Robust Multi-View Fusion under Dynamic Packet Loss}

In multi-UGV sensing systems, fluctuating communication capacity leads to unpredictable packet loss, dropping critical data and harming inference accuracy. We propose a robust multi-view fusion method that adapts to dynamic packet loss by prioritizing important features and assigning potential losses to lower-priority data. Firstly, we compute the priority of each feature map using average pooling. For feature maps \( \hat{z}_{t}^{(k)} \) from camera \( k \) at time \( t \), the priority \( p_{t}^{(k)} \) is calculated as:
\begin{equation}
    p_{t}^{(k)} = \frac{1}{C H W} \sum_{c=1}^C \sum_{h=1}^H \sum_{w=1}^W \hat{z}_{t,c,h,w}^{(k)},
\end{equation}
where \( \hat{z}_{t,c,h,w}^{(k)} \) represents the element value of \( \hat{z}_{t}^{(k)} \) at time $t$, channel \( c \), height \( h \), and width \( w \). Besides, \( C \), \( H \), and \( W \) are the dimensions of the feature map, and \( \hat{z}_{t}^{(k)} \) has dimensions \( T \times C \times H \times W \). Higher average values \( p_{t}^{(k)} \) indicate more important features. Based on the dynamic packet loss rate, we allocate potential losses to features with the lowest priorities. We sort the feature maps \( \hat{z}_{t}^{(k)} \) by their priorities \( p_{t}^{(k)} \) and assign a mask \( m_{t}^{(k)} \in \{0,1\} \), where \( m_{t}^{(k)} = 1 \) if the feature is likely to be received (higher priority) and \( m_{t}^{(k)} = 0 \) otherwise. The masked features are then defined as $\tilde{z}_{t}^{(k)} = m_{t}^{(k)} \hat{z}_{t}^{(k)}$, where \( \tilde{z}_{t}^{(k)} \) represents the masked feature map corresponding to \( \hat{z}_{t}^{(k)} \). The fusion function \( \mathcal{G} \) incorporates both the masked features \( \tilde{z}_{t}^{(k)} \) and the masks \( m_{t}^{(k)} \), allowing the network to adjust for missing data due to packet loss. The fusion is formulated as:
\begin{equation}
    \hat{y}_t = \mathcal{G}\left( \left\{ \tilde{z}_{t}^{(k)},\ m_{t}^{(k)} \right\}_{k=1}^{K};\ {\Theta_r} \right),
\end{equation}
where \( {\Theta_r} \) represents the fusion parameters, and \( \hat{y}_t \) is the fused output (occupancy map) at time \( t \). Since our goal is to minimize the inference error while accounting for dynamic packet loss rates, the total loss function integrates inference accuracy and a penalty for losing important features:
\begin{equation}
    \mathcal{L}_3 = \mathbb{E} \left[ \sum_{t=1}^{N} d\left( y_t, \hat{y}_t \right) \right] + \alpha_d \sum_{t=1}^{N} \sum_{k=1}^{K} \left( 1 - m_{t}^{(k)} \right),
\end{equation}
where \( y_t \) is the ground truth at time \( t \), \( d(\cdot) \) measures inference error, and the second term penalizes the loss of important features. The weight \( \alpha_d \) balances robustness and accuracy. Therefore, the total loss is given by:
\begin{equation}
    \mathcal{L} = \mathcal{L}_1 + \alpha_2 \mathcal{L}_2 + \alpha_3 \mathcal{L}_3,
\end{equation}
where \( \mathcal{L}_1 \) and \( \mathcal{L}_2 \) are other loss terms related to bitrate minimization and inference accuracy, and \( \alpha_2 \) and \( \alpha_3 \) are weights for balancing these components.

\section{Performance Evaluation}
\label{sec:Performance Evaluation}

\begin{figure}[t]
  \centering
  \includegraphics[width=1.0\linewidth]{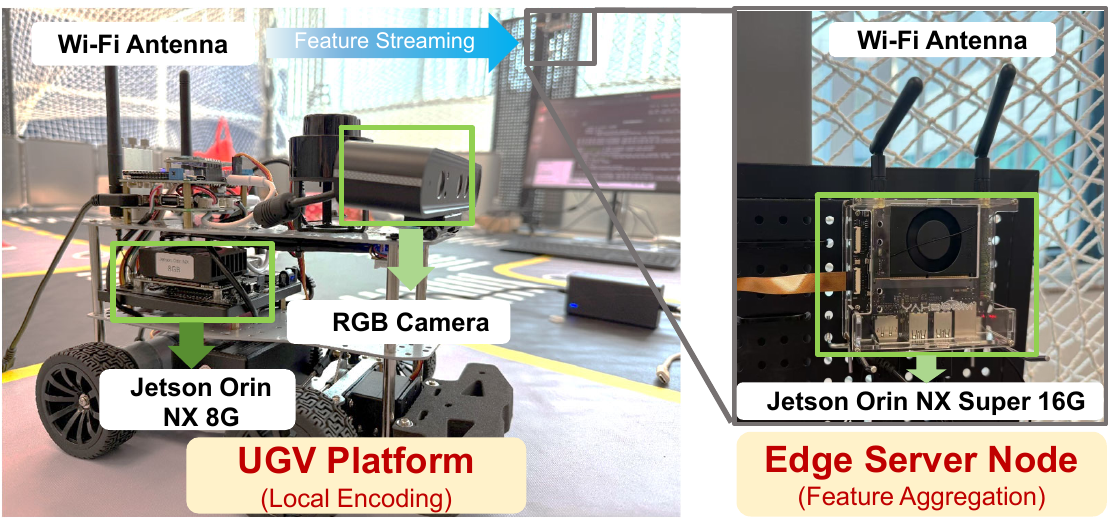}
  \caption{{\color{black}Our algorithm is deployed on a UGV-edge server platform. The UGV node captures RGB images and performs local encoding, transmitting features via Wi-Fi to the edge server node for aggregation and further processing.}}
  \label{fig:hardware_platform}
  \vspace{-3mm}
\end{figure}

\begin{figure}[t]
  \centering
  \subfigure[Error vs. $\Delta _{{{k}}}^{\text{Ca}}$ for 10kb.]{
    \includegraphics[width=4.2 cm]{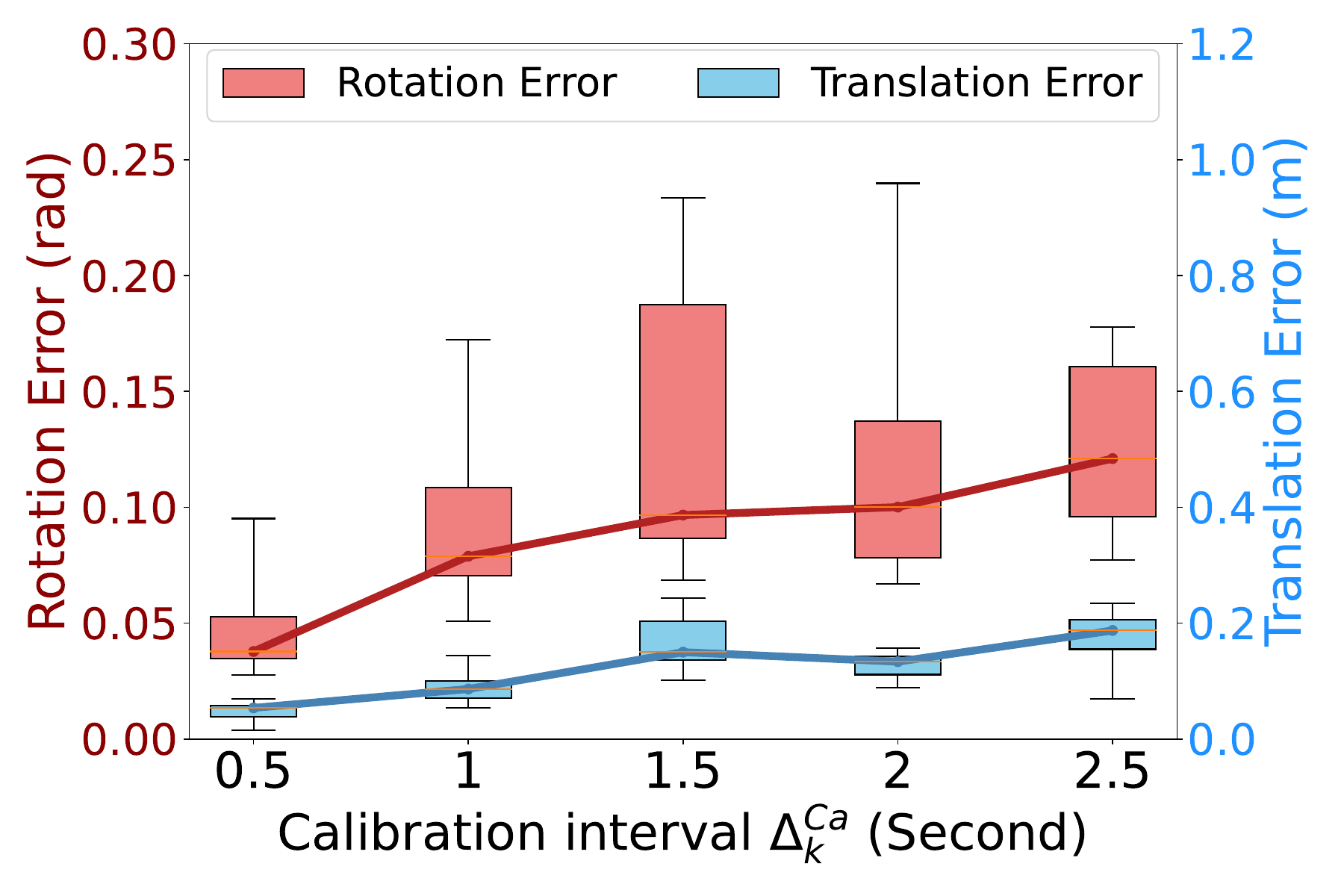}\label{fig:error_10kb}
  }
  \hspace{-0.5cm}  % Adjust this value to bring the subfigures closer
  \subfigure[Error vs. $\Delta _{{{k}}}^{\text{Ca}}$ for 30kb.]{
    \includegraphics[width=4.3 cm]{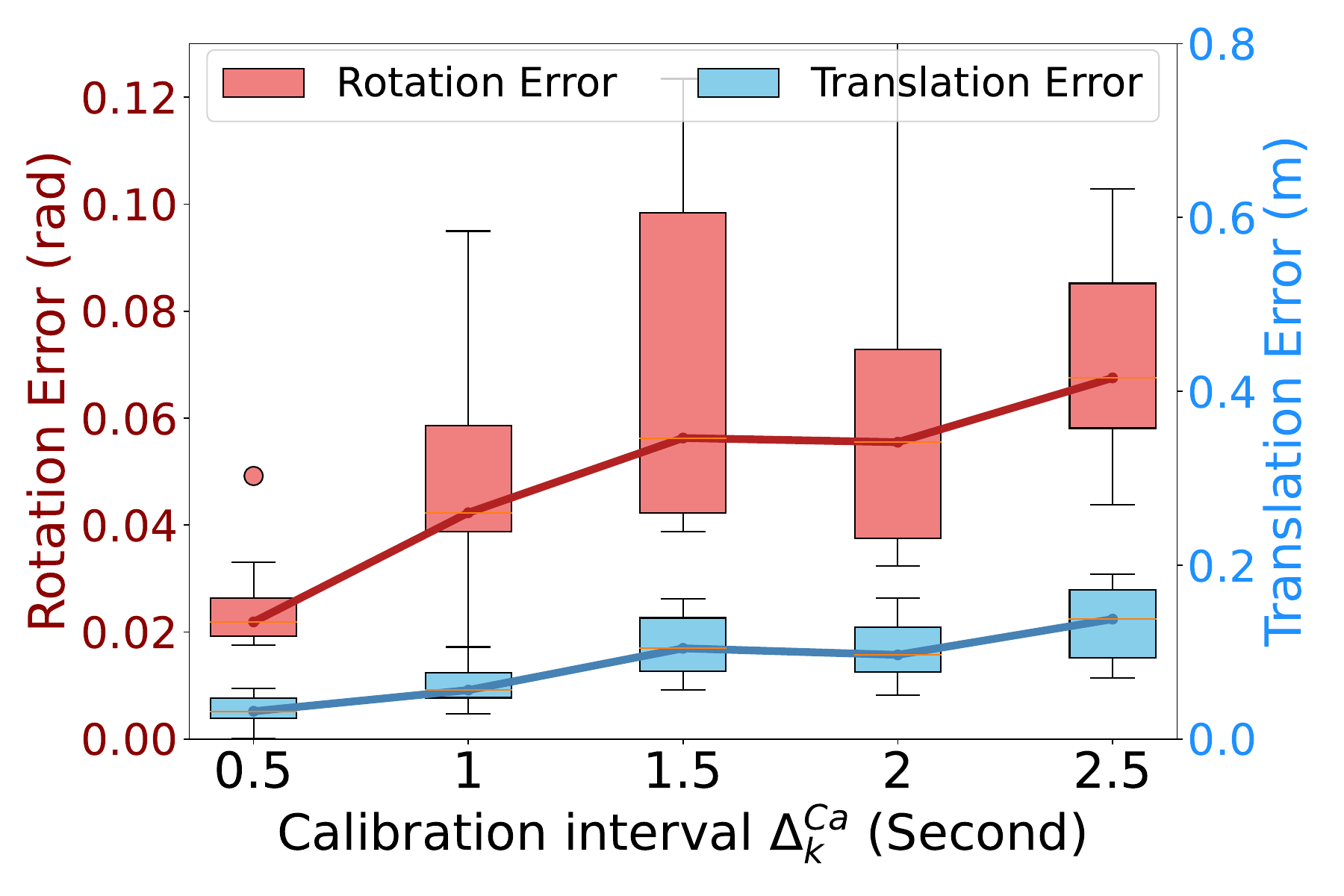}\label{fig:error_30kb}
  }
  \caption{Error vs. calibration interval $\Delta _{{{k}}}^{\text{Ca}}$ for different data sizes.}
  \label{fig:error_comparison}
  \vspace{-3mm}
\end{figure}

\begin{table}[t]
\centering
\caption{Latency of Re-ID calibration pipeline on the UGV platform with Nvidia Jetson unit.}
\label{tab:latency_analysis}
\renewcommand{\arraystretch}{1.2}
\begin{tabular}{|c|c|c|}
\hline
\textbf{Operation} & \textbf{Description} & \textbf{Latency} \\
\hline
Pedestrian Detection & YOLOv5s inference & $10\pm2$ ms \\
\hline
Feature Extraction & OSNet-x1.0  & $60\pm15$ ms\\
\hline
Feature Matching & Euclidean distance matching & $3\pm1$ ms\\
\hline
\textbf{Total Latency} & Detection + Extraction + Matching & ${73\pm18}$ ms\\
\hline
\end{tabular}
\end{table}

\begin{figure}[t]
  \centering
  \subfigure[Rotation Error vs. MODA.]{
    \includegraphics[width=4.2 cm]{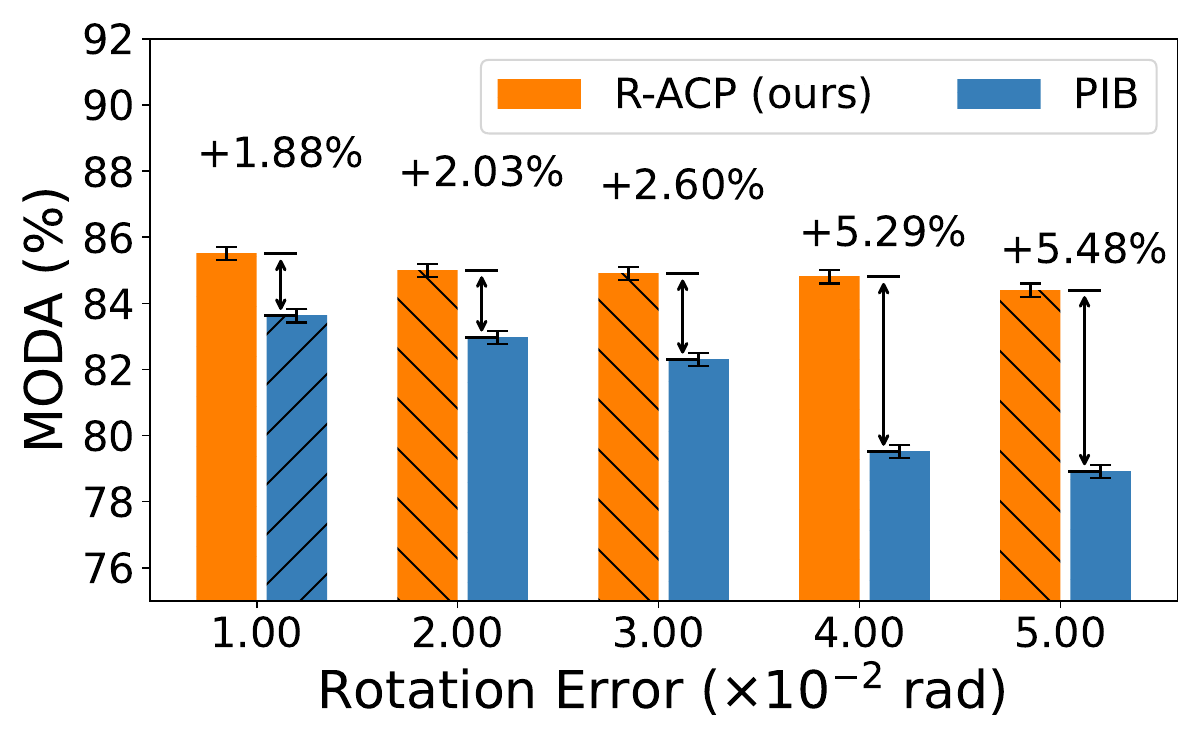}
    \label{fig:rotation_error_moda}
  }
  \hspace{-0.5cm}  % Adjust subfigure spacing
  \subfigure[Translation Error vs. MODA.]{
    \includegraphics[width=4.3 cm]{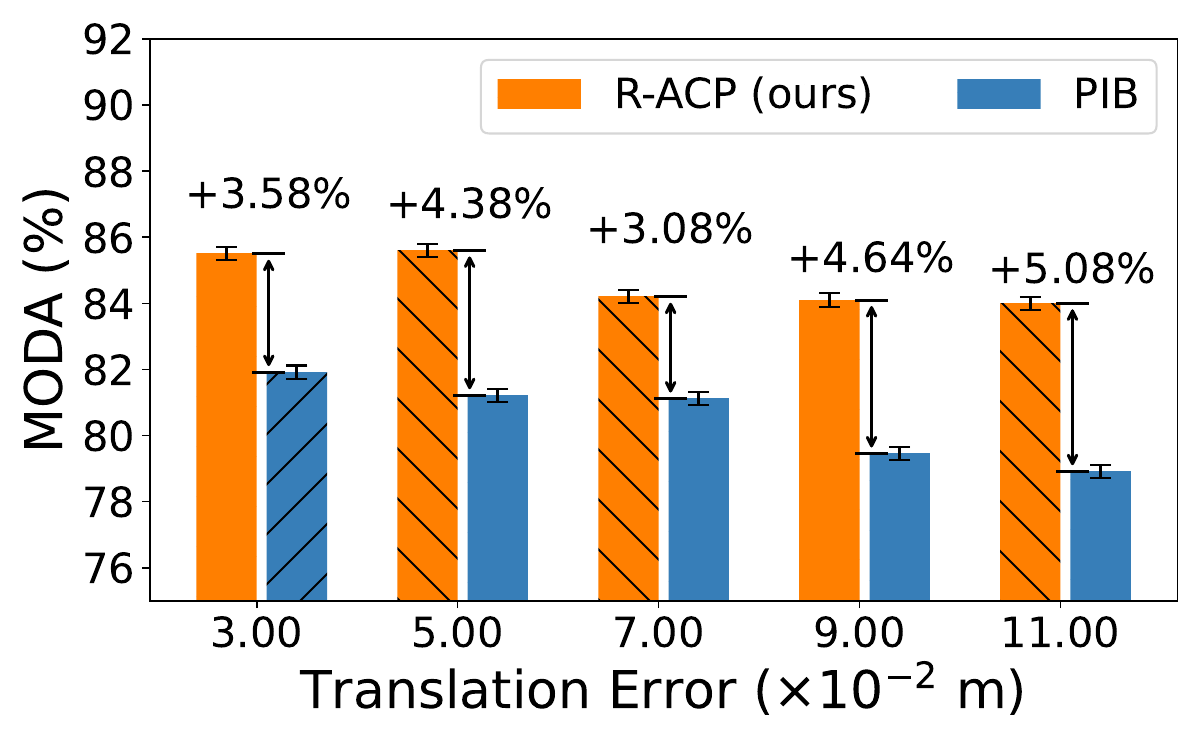}
    \label{fig:translation_error_moda}
  }
  \caption{Comparison of MODA performance with respect to rotation and translation errors.}
  \label{fig:error_comparison}
  \vspace{-3mm}
\end{figure}

\begin{figure}[t]
  \centering
  \subfigure[Cameras 0-3.]{
    \includegraphics[width=4.3cm]{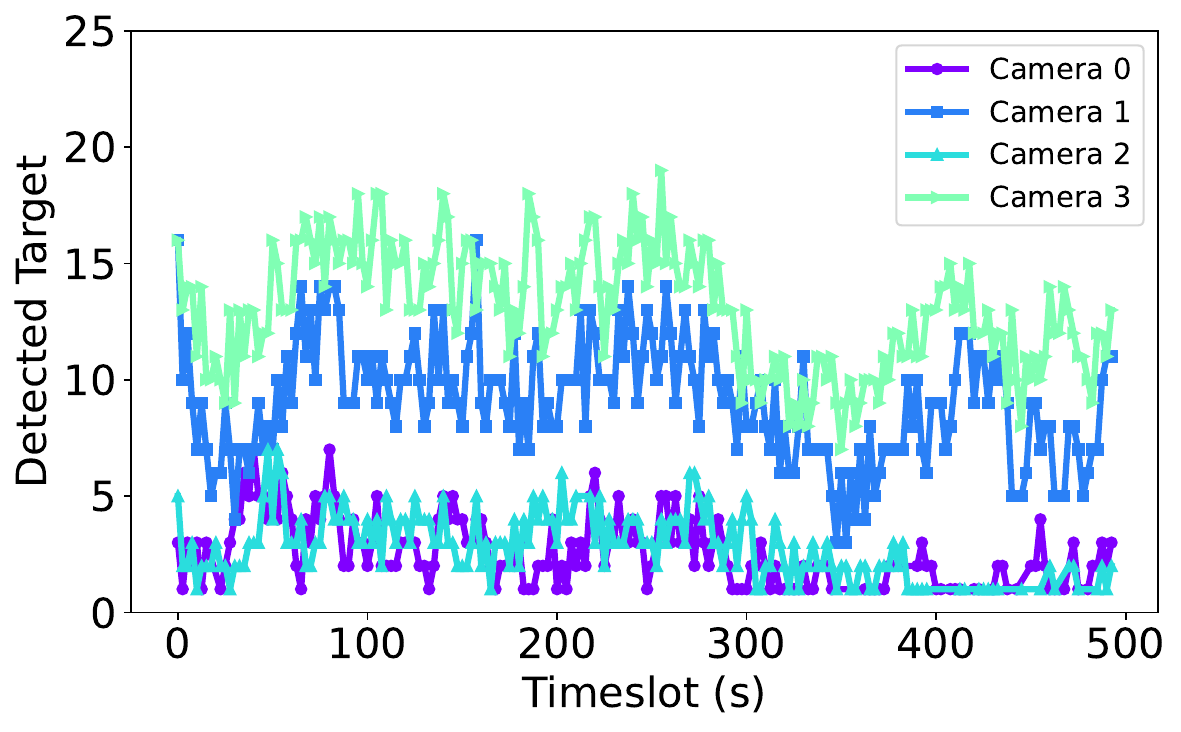}\label{fig:detected_targets_camera_0_1_2_3}
  }
  \hspace{-0.5cm}  % 调整子图间距
  \subfigure[Cameras 4-6.]{
    \includegraphics[width=4.3cm]{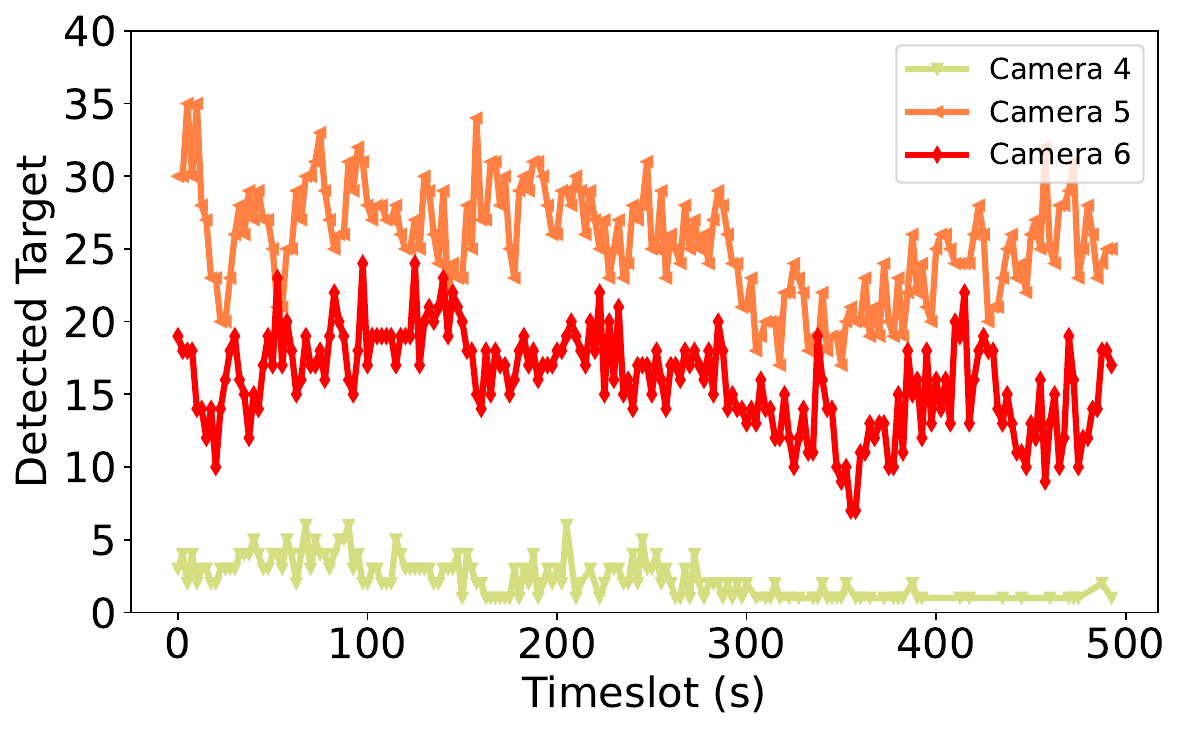}\label{fig:detected_targets_camera_4_5_6}
  }
  \caption{Detected targets over time for different cameras.}
  \label{fig:detected_targets_comparison}
  \vspace{-3mm}
\end{figure}

\begin{figure}[t]
  \centering
  \includegraphics[width=1.0\linewidth]{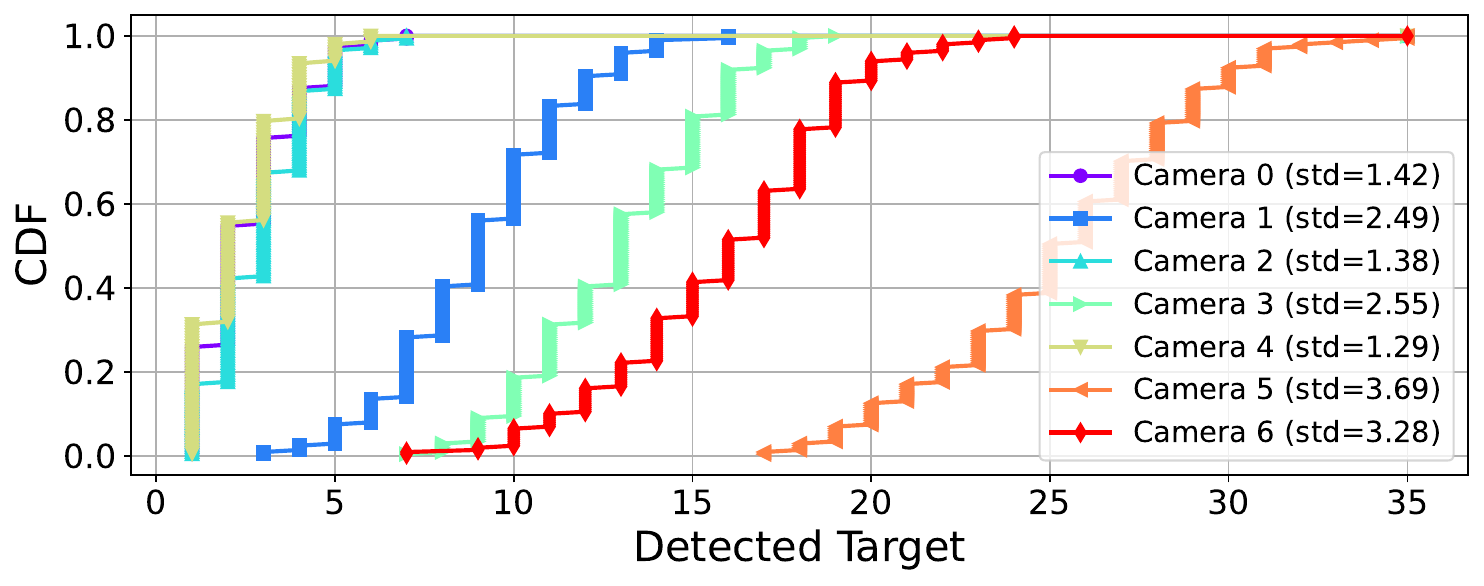}
  \caption{Cumulative distribution function (CDF) of detected targets per camera.}
  \label{fig:cdf_detected_target_per_camera}
  \vspace{-3mm}
\end{figure}

\begin{figure}[t]
  \centering
  \subfigure[Camera 2]{
    \includegraphics[width=4.2cm]{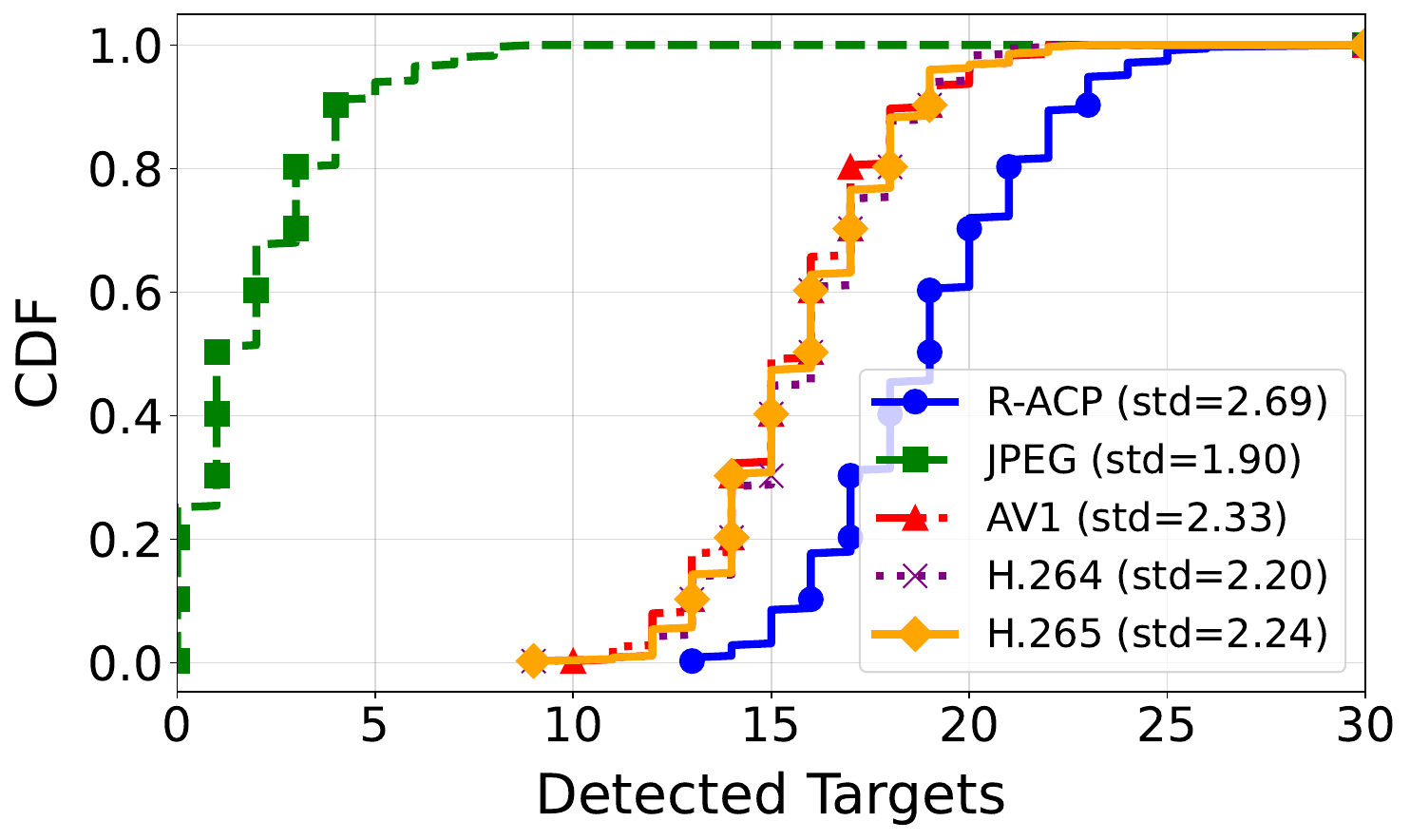}
    \label{fig:cdf_c2}
  }
  \hspace{-0.3cm}
  \subfigure[Camera 5]{
    \includegraphics[width=4.2cm]{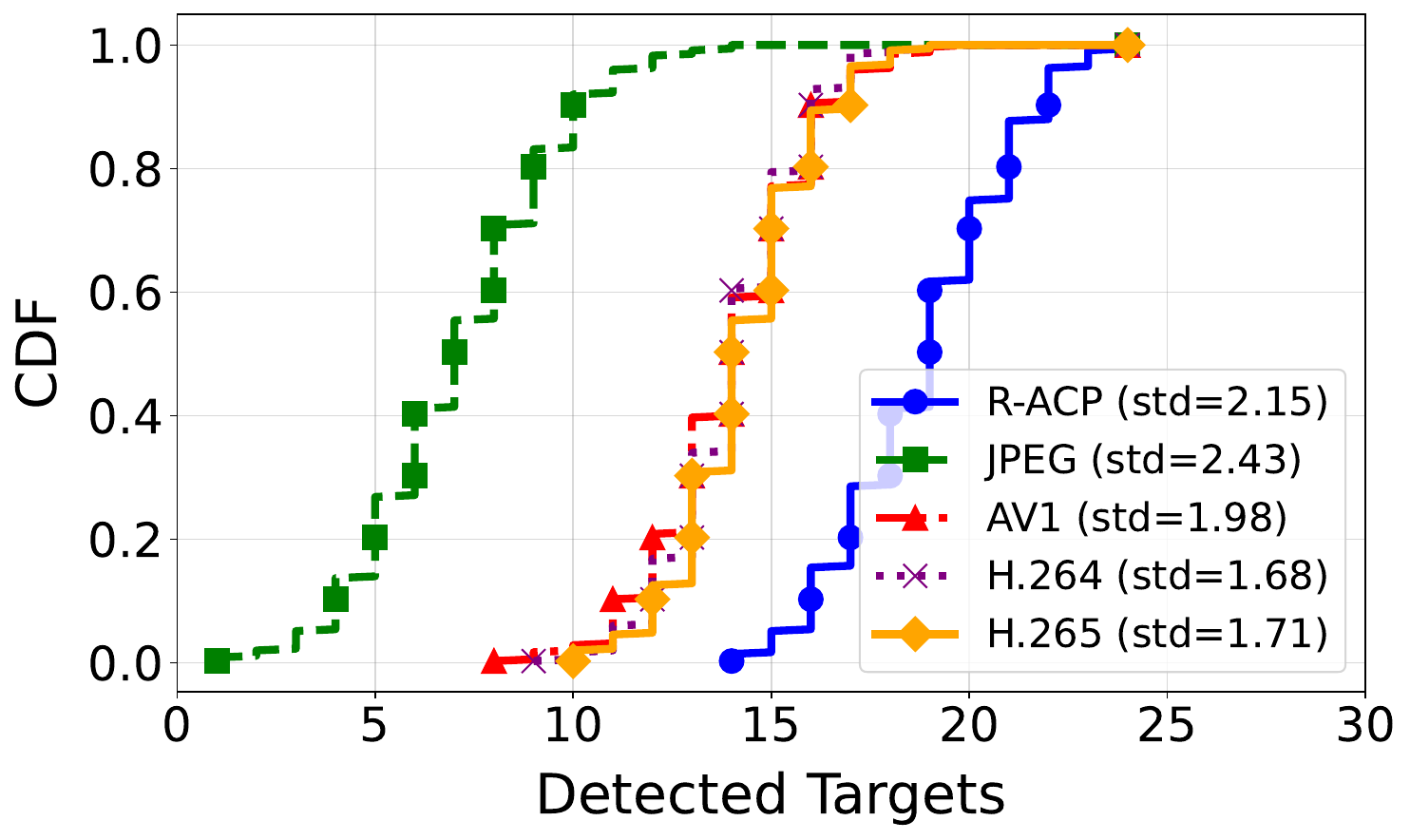}
    \label{fig:cdf_c5}
  }
  \vspace{-3mm}
  
  \subfigure[Camera 6]{
    \includegraphics[width=4.2cm]{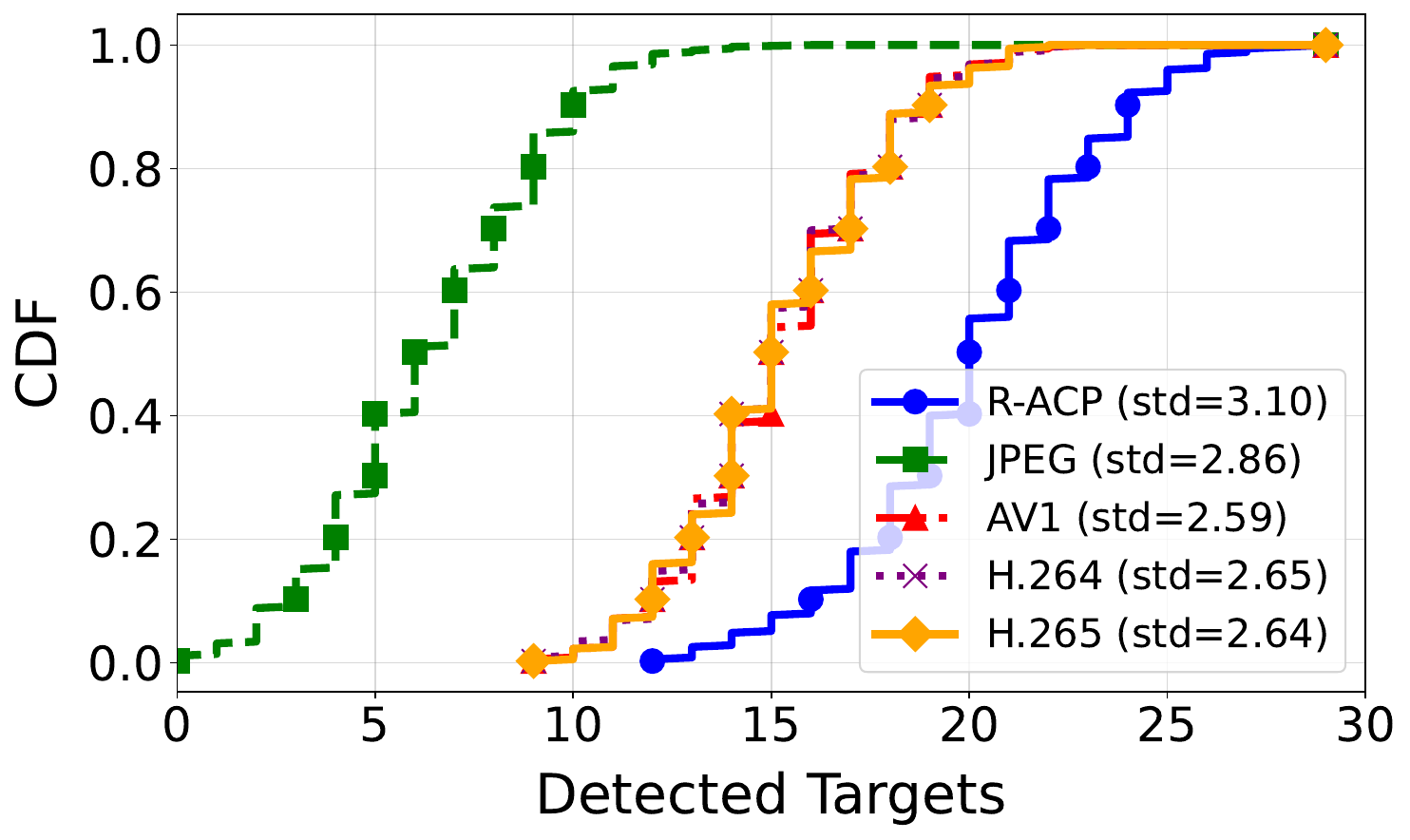}
    \label{fig:cdf_c6}
  }
  \hspace{-0.3cm}
  \subfigure[Camera 7]{
    \includegraphics[width=4.2cm]{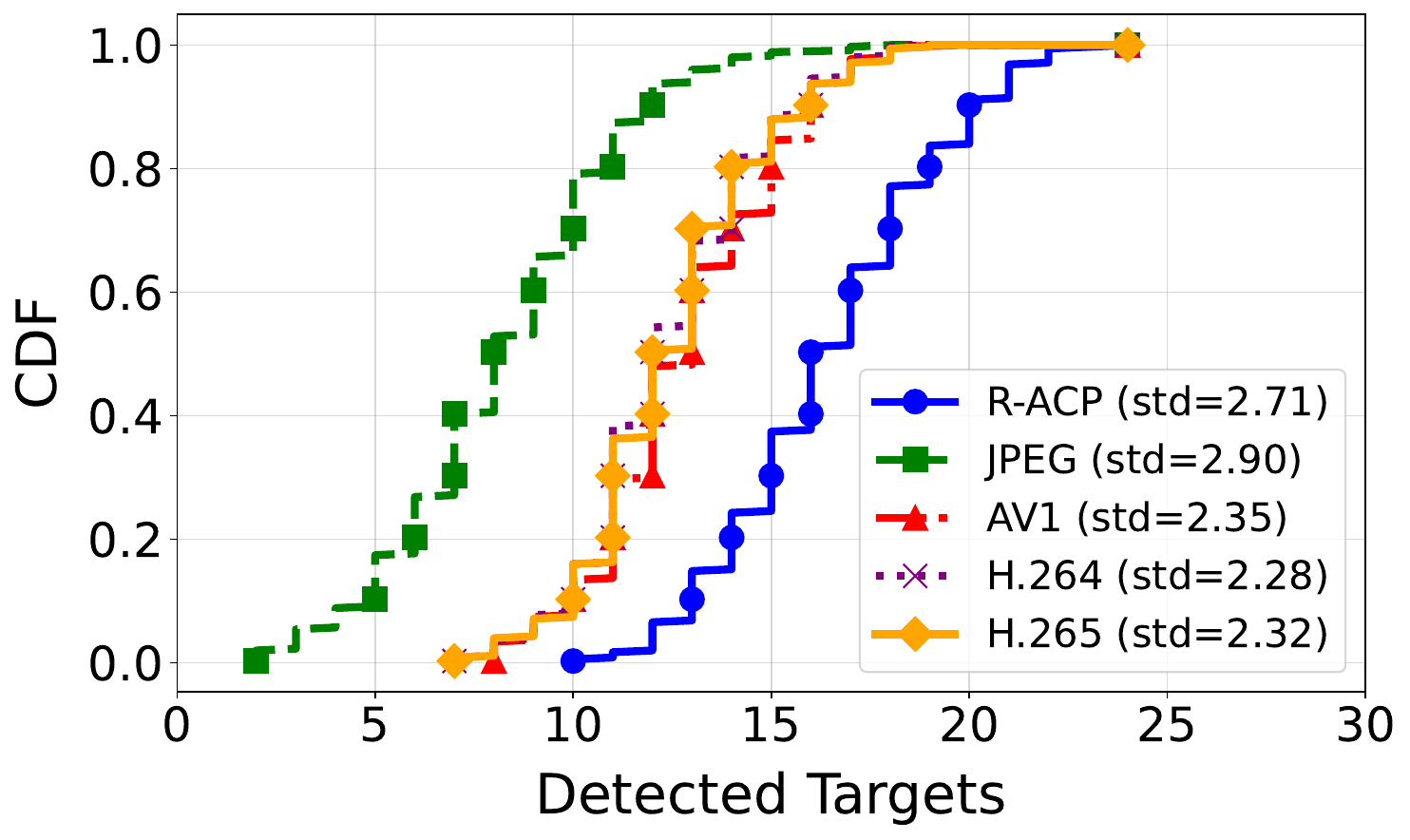}
    \label{fig:cdf_c7}
  }
  \caption{{\color{black}CDF of detected targets under different compression methods for cameras 2, 5, 6, and 7.}}
  \label{fig:cdf_all_cams}
  \vspace{-3mm}
\end{figure}

\subsection{Simulation Setup}
We set up simulations to evaluate our R-ACP framework, aimed at predicting pedestrian occupancy in urban settings using multiple cameras. These simulations replicate a city environment with variables like signal frequency and device density affecting the outcomes. We simulate a communication system operating at a 2.4 GHz frequency with path loss exponent of 3.5, and an 8 dB shadowing deviation to model wireless conditions. To assess congestion levels, devices emitting 0.1 Watts interfere at densities of 10 to 100 per 100 square meters. The bandwidth is set to 2 MHz. We use the \textit{Wildtrack} dataset from EPFL, featuring high-resolution images from seven cameras capturing pedestrian movements in a public area\cite{chavdarova2018wildtrack}. Each camera provides 400 frames at 2 frames per second, totaling over 40,000 bounding boxes for more than 300 pedestrians. In our simulations, the positions of these cameras correspond to the positions of UGVs. {\color{black}Moreover, we also evaluate the R-ACP framework in our hardware platform. As shown in Fig.~\ref{fig:hardware_platform}, the platform consists of a UGV node equipped with an RGB camera and a Jetson Orin NX 8G module for local feature encoding. The extracted features are transmitted over Wi-Fi to an edge server node with a Jetson Orin NX Super 16G for feature aggregation and inference. This setup enables efficient multi-agent collaborative perception under communication constraints.}

The primary metric of inference performance is MODA, which assesses the system’s ability to accurately detect pedestrians based on missed and false detections. We also examine the rate-performance tradeoff to understand how communication overhead affects calibration and multi-view perception. For comparative analysis, we consider five baselines, including video coding and image coding, as follows. 
\begin{itemize}
    \item \textbf{PIB}\cite{fang2025ton}:  A collaborative perception framework that enhances detection accuracy and reduces communication costs by prioritizing and transmitting only useful features.
    \item \textbf{JPEG}\cite{wallace1992jpeg}: A widely used image compression standard employing lossy compression algorithms to reduce image data size, commonly used to decrease communication load in networked camera systems.
    \item \textbf{H.264}\cite{H264}: Known as Advanced Video Coding (AVC) or MPEG-4 Part 10, the standard that significantly enhances video compression efficiency, allowing high-quality video transmission at lower bit rates.
    \item \textbf{H.265}\cite{bossen2012hevc}: Also known as High Efficiency Video Coding (HEVC), the standard that offers up to 50\% better data compression than its predecessor H.264 (MPEG-4 Part 10), while maintaining the same video quality.
    \item \textbf{AV1}\cite{han2021technical}: AOMedia Video 1 (AV1), an open, royalty-free video coding format developed by the Alliance for Open Media (AOMedia), and designed to succeed VP9 with improved compression efficiency. 
\end{itemize}

In Fig. \ref{fig:error_10kb}, the effect of varying calibration intervals \( \Delta _{\mathrm{AoPT}}^{\mathrm{ca}} \) on rotation and translation errors under a 10KB key-point feature constraint is shown. As the calibration interval increases, both errors rise. Similarly, Fig. \ref{fig:error_30kb} shows error trends for a 30KB key-point feature size, where increasing \( \Delta _{\mathrm{AoPT}}^{\mathrm{ca}} \) also results in higher errors. Comparing the two figures, a larger communication bottleneck (30KB) yields more granular features and significantly improves calibration accuracy compared to the smaller 10KB bottleneck. {\color{black}As shown in Table \ref{tab:latency_analysis}, our numerical results show that the average per-frame latency is approximately $73\pm18$ ms when processing frames with detected pedestrians. Specifically, pedestrian detection consumes $10\pm2$ ms, total feature extraction takes $60\pm15$ ms, and feature matching adds $3\pm1$ ms overhead. These results demonstrate that our Re-ID module is capable of operating at approximately $14$ frames per second (FPS) on the Jetson Orin, even without aggressive optimization (e.g., TensorRT or quantization). Therefore, the proposed method remains feasible for real-time deployment in dynamic edge environments.}

In Fig. \ref{fig:rotation_error_moda}, the relationship between rotation error and MODA is presented, showing that as rotation error increases, the accuracy of multi-UGV collaborative perception improves. R-ACP achieves up to 5.48\% better MODA than the PIB baseline for the highest rotation error. Similarly, Fig. \ref{fig:translation_error_moda} shows that larger translation errors result in higher MODA scores, with R-ACP outperforming PIB by up to 5.08\%. Both figures indicate that R-ACP maintains superior calibration and perception accuracy under larger errors, enhancing collaborative multi-UGV perception.

In Fig. \ref{fig:detected_targets_comparison}, the detected target counts are plotted over time for different cameras. In Fig. \ref{fig:detected_targets_camera_0_1_2_3}, cameras 0 to 3 detect a moderate number of targets, with cameras 2 and 3 identifying significantly more targets than others. Fig. \ref{fig:detected_targets_camera_4_5_6} shows that cameras 4, 5, and 6 detect a higher number of targets overall, respectively, with camera 5 detecting the most, followed by camera 6. These variations highlight that UGVs equipped with cameras and with more detected targets, such as UGVs 3, 5, and 6, have a higher priority in calculating AoPT. Fig. \ref{fig:cdf_detected_target_per_camera} shows the cumulative distribution function (CDF) of detected targets across 500 time slots for each camera employing H.265 compression. The results indicate considerable variations in the number of detected targets across UGVs, with UGVs 5 and 6 capturing the most targets within their FOVs, further confirming their higher priority in terms of data timeliness when computing AoPT. {\color{black}Under a constrained channel condition of 30KB/s, as shown in Fig.~\ref{fig:cdf_all_cams}, R-ACP consistently yields higher detection rates than traditional codecs. 
This suggests that R-ACP is better suited for preserving occupancy-critical visual information in bandwidth-limited edge deployments.}

\begin{figure*}[t]
  \centering
  \subfigure[Perception result from \textbf{FOV 1}.]{
    \includegraphics[width=0.48\textwidth]{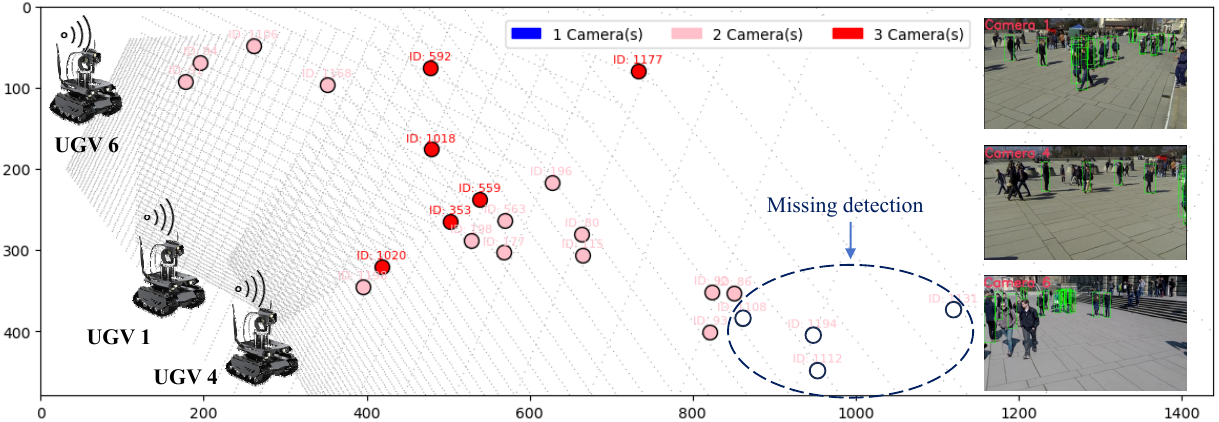}
    \label{fig:roi_1}
  }
  \hspace{-0.2cm}  % 调整子图间距
  \subfigure[Perception result from \textbf{FOV 2}.]{
    \includegraphics[width=0.48\textwidth]{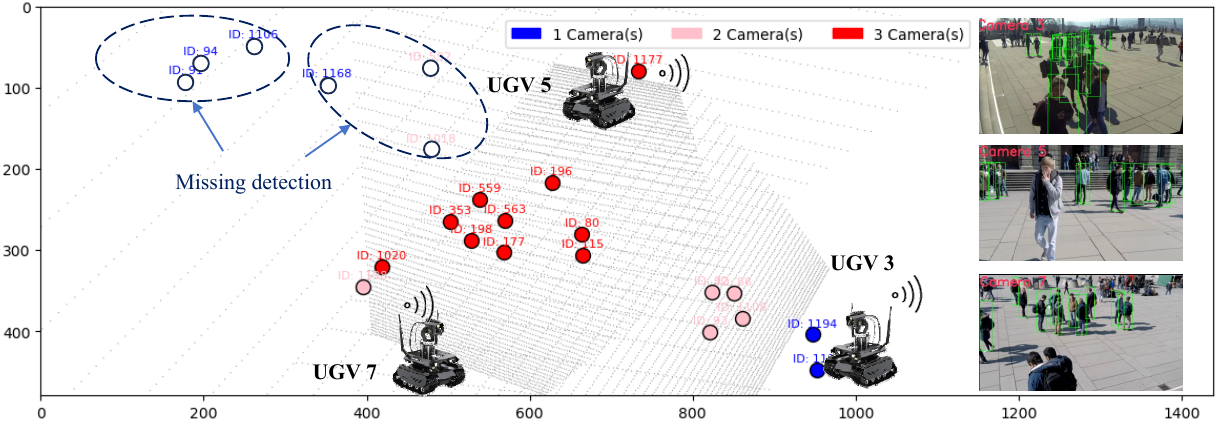}
    \label{fig:roi_2}
  }
  \subfigure[Perception result from \textbf{FOV 3}.]{
    \includegraphics[width=0.48\textwidth]{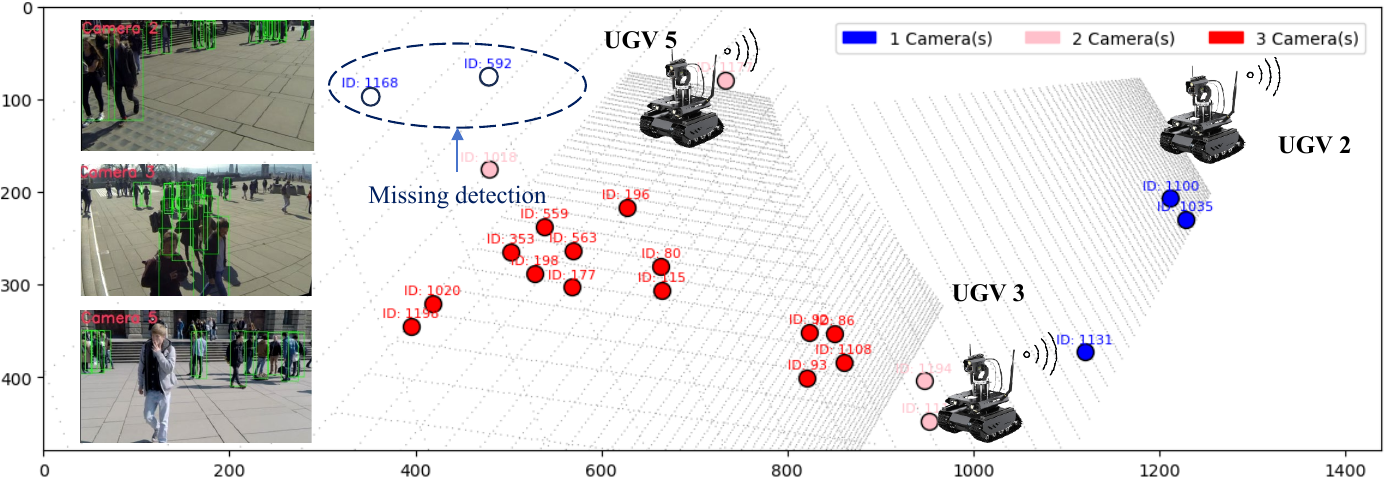}
    \label{fig:roi_3}
  }
  \hspace{-0.2cm}  % 调整子图间距
  \subfigure[Perception result from \textbf{all UGV cameras}.]{
    \includegraphics[width=0.48\textwidth]{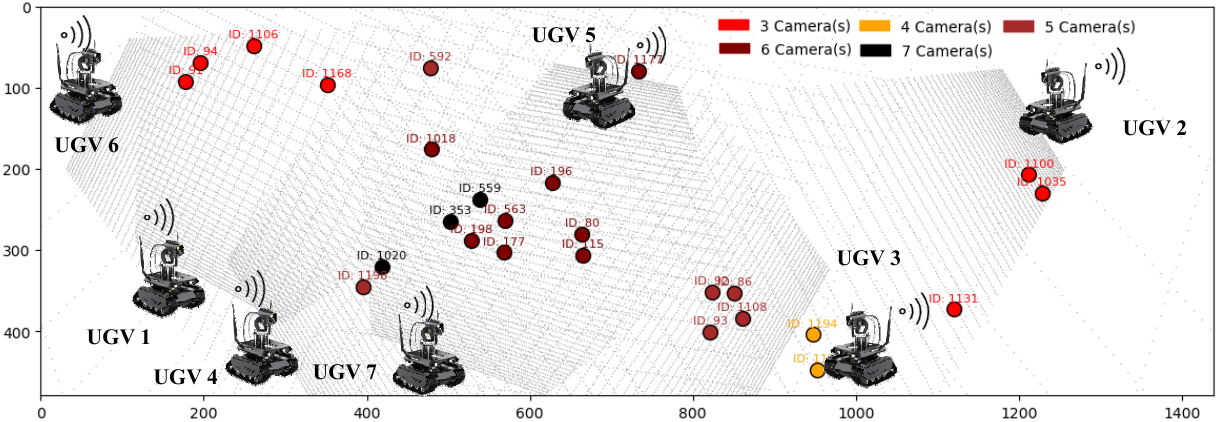}
    \label{fig:roi_all}
  }
  \caption{Comparison of perception results from different FOVs using all UGV cameras. Figs. \ref{fig:roi_1} to \ref{fig:roi_3} show the results for individual FOVs, while Fig. \ref{fig:roi_all} shows the result from using all UGV cameras collaboratively.}
  \label{fig:roi_comparison}
  \vspace{-3mm}
\end{figure*}

Fig. \ref{fig:roi_comparison} visualizes the results of collaborative perception by multiple UGVs within a 12m×36m area, represented as a 480×1440 grid with a resolution of 2.5 cm\(^2\). In this experimental setup, seven wireless edge cameras work together to perceive the area, and contour lines are used to represent the perception range of each camera. The denser the contour lines, the closer the target is to the camera, which correlates with higher perception accuracy. The comparison between individual FOVs in Figs. \ref{fig:roi_1}, \ref{fig:roi_2}, and \ref{fig:roi_3} illustrates the variability in pedestrian detection when only one or a few UGVs contribute to perception. In these figures, certain areas show missing detections due to limited coverage, with only one or two cameras detecting some targets. In contrast, Fig. \ref{fig:roi_all} demonstrates the advantage of using all UGVs collaboratively, covering a larger FOV and significantly reducing missed detections. Pedestrians are detected more accurately and with better resolution when multiple cameras provide complementary perspectives, leading to improved overall system performance. Additionally, the perception accuracy depends on the proximity of the cameras to the targets, as indicated by the density of the contour lines. The closer a target is to the cameras, the higher the resolution and accuracy of its detection, emphasizing the importance of strategic UGV positioning and multi-UGV collaboration for real-time monitoring tasks.

\begin{figure}[t]
  \centering
  \subfigure[Detected targets across FOVs.]{
    \includegraphics[width=4.2 cm]{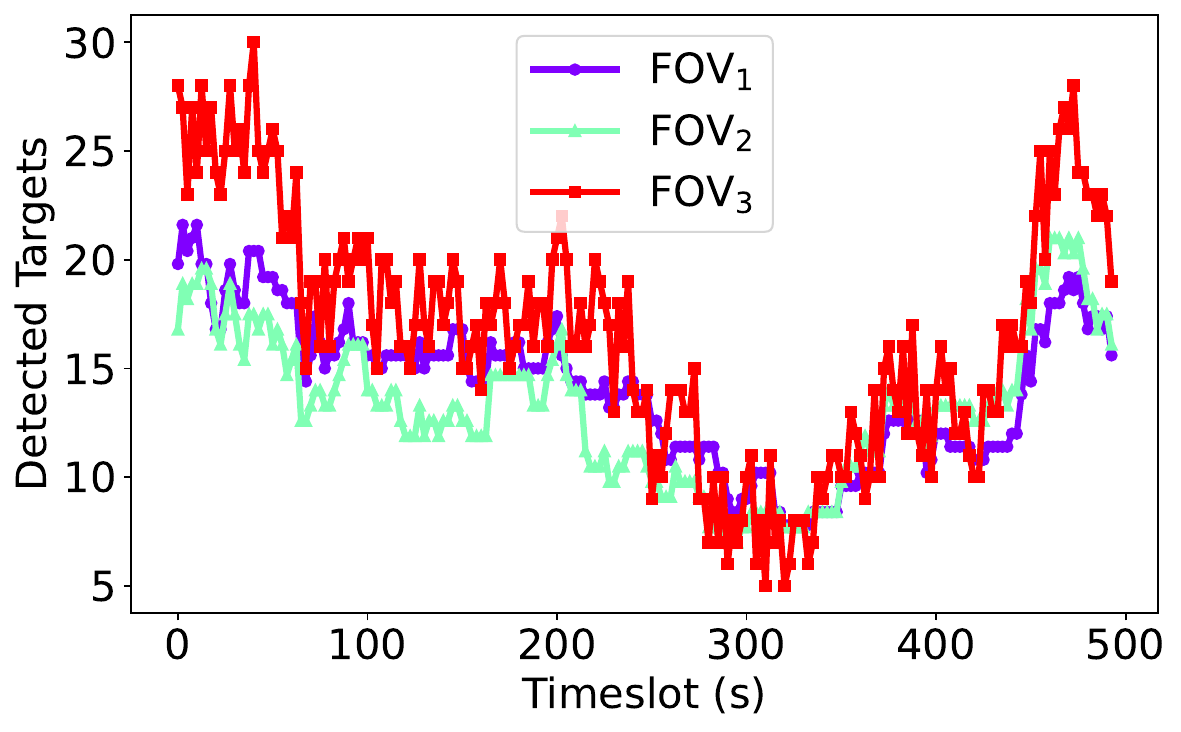}
    \label{fig:detected_targets}
  }
  \hspace{-0.5cm}  % 调整子图间距
  \subfigure[AoPT across FOVs.]{
    \includegraphics[width=4.3 cm]{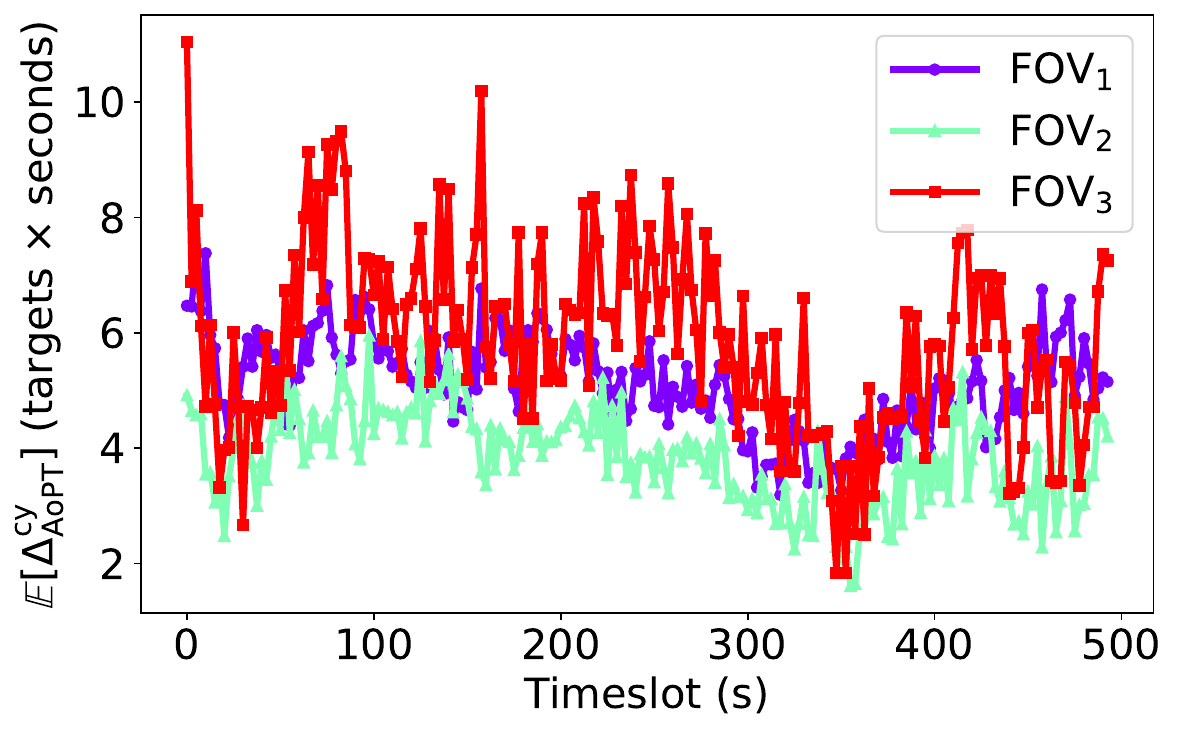}
    \label{fig:aopt_fovs}
  }
  \caption{Comparison of detected targets and AoPT for different FOVs over time. Fig. \ref{fig:detected_targets} shows the detected targets for FOVs 1-3, while Fig. \ref{fig:aopt_fovs} illustrates the AoPT values.}
  \label{fig:fov_comparison}
  \vspace{-3mm}
\end{figure}

\begin{table}[t]
    \centering
    \caption{Impact of Varying FOVs and Communication Costs on AoPT and Collaborative Perception Accuracy.}
    \begin{tabular}{@{}>{\centering\arraybackslash}p{1.5cm} >{\centering\arraybackslash}p{1.8cm} >{\centering\arraybackslash}p{2.3cm} >{\centering\arraybackslash}p{2.0cm}@{}}
        \toprule
        \textbf{FOV Num.} & \textbf{Comm. Cost} & \textbf{AoPT (Target $\times$ s)} & \textbf{MODA (\%)} \\ 
        \midrule
        \textbf{FOV 1} & 15.36 KB & 6.53±0.64  & 63.15 \\
        \textbf{FOV 1} & 18.69 KB & 8.13±1.02  & 64.90 \\
        \midrule
        \textbf{FOV 2} & 14.81 KB & 7.14±0.95 & 52.34\\
        \textbf{FOV 2} & 18.69 KB & 8.01±1.04 & 56.49 \\
        \midrule
        \textbf{FOV 3} & 15.47 KB & 6.91±0.86 & 66.46 \\
        \textbf{FOV 3} & 18.73 KB & 8.31±1.26 & 67.09 \\     
        \midrule
        \textbf{FOVs 1-3} & 17.07 KB & 9.27±1.15 & 84.14 \\
        \textbf{FOVs 1-3} & 26.62 KB & 10.13±1.40 & 85.86 \\       
        \bottomrule
    \end{tabular}
    \label{tab:fusion_camera}
    \vspace{-3mm}
\end{table}

Table \ref{tab:fusion_camera} illustrates the impact of different FOVs and communication costs on AoPT and collaborative perception accuracy. As communication costs increase, the amount of feature data transmitted between UGVs rises, leading to more accurate target detection and higher numbers of detected targets. Consequently, both AoPT and MODA values improve, highlighting that the network should allocate more resources to UGVs associated with these FOVs to enhance the timeliness and precision of perception data. The table also demonstrates that different FOVs capture varying numbers of targets, which suggests that optimizing resource distribution should prioritize nodes with higher AoPT values to maximize system performance. Additionally, when data from multiple FOVs is combined, the system achieves its highest MODA scores, but this also leads to increased communication costs and higher AoPT values. This indicates that while multi-FOV collaboration improves overall perception accuracy, it also necessitates a careful balance of resource allocations to manage the increased communication demands efficiently. 

Fig. \ref{fig:fov_comparison} illustrates the dynamic variations in detected targets and AoPT across different FOVs over time. In Fig. \ref{fig:detected_targets}, we observe the fluctuation of the detected target counts in FOVs 1, 2, and 3 as targets move within and out of the cameras' fields of view. Fig. \ref{fig:aopt_fovs} shows the corresponding AoPT values, where higher detected target counts result in increased AoPT, indicating the system’s need for more channel resources to maintain data timeliness. Since the targets are continuously moving, both the number of detected targets and AoPT are dynamic over time. As the communication constraint and camera sampling interval remain fixed, the larger the number of detected targets, the greater the AoPT becomes, indicating that the system will require more resources to ensure timely transmission. This illustrates the relationship between the dynamic nature of target movement and the system’s resource allocation strategy for maintaining efficient data timeliness.

\begin{figure}[t]
  \centering
  \subfigure[AoPT vs. Comm. Bottleneck.]{
    \includegraphics[width=4.3 cm]{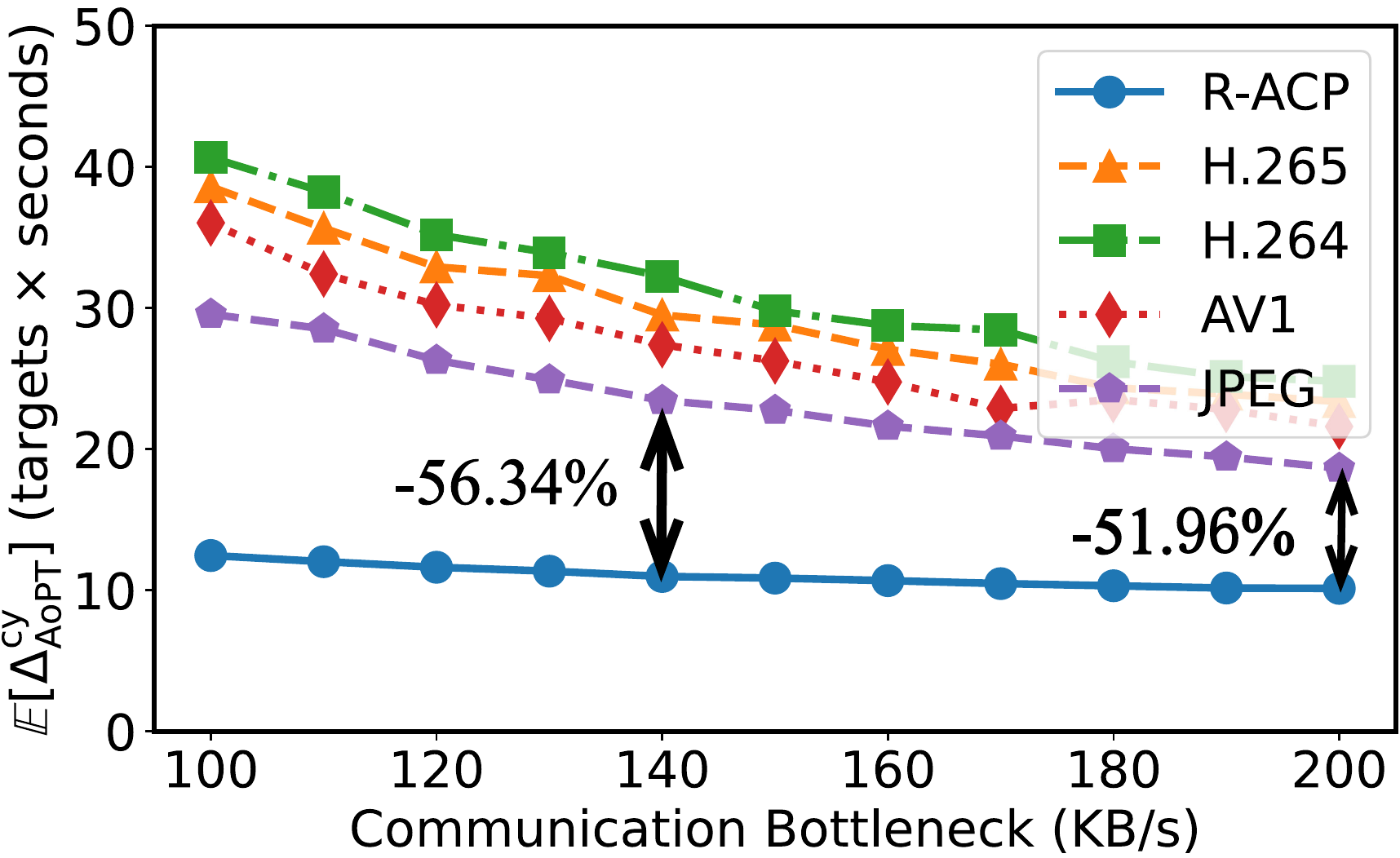}\label{fig:aopt_vs_capacity}
  }
  \hspace{-0.5cm}  % 调整此值以控制子图之间的间距
  \subfigure[AoPT vs. Sampling Interval.]{
    \includegraphics[width=4.3 cm]{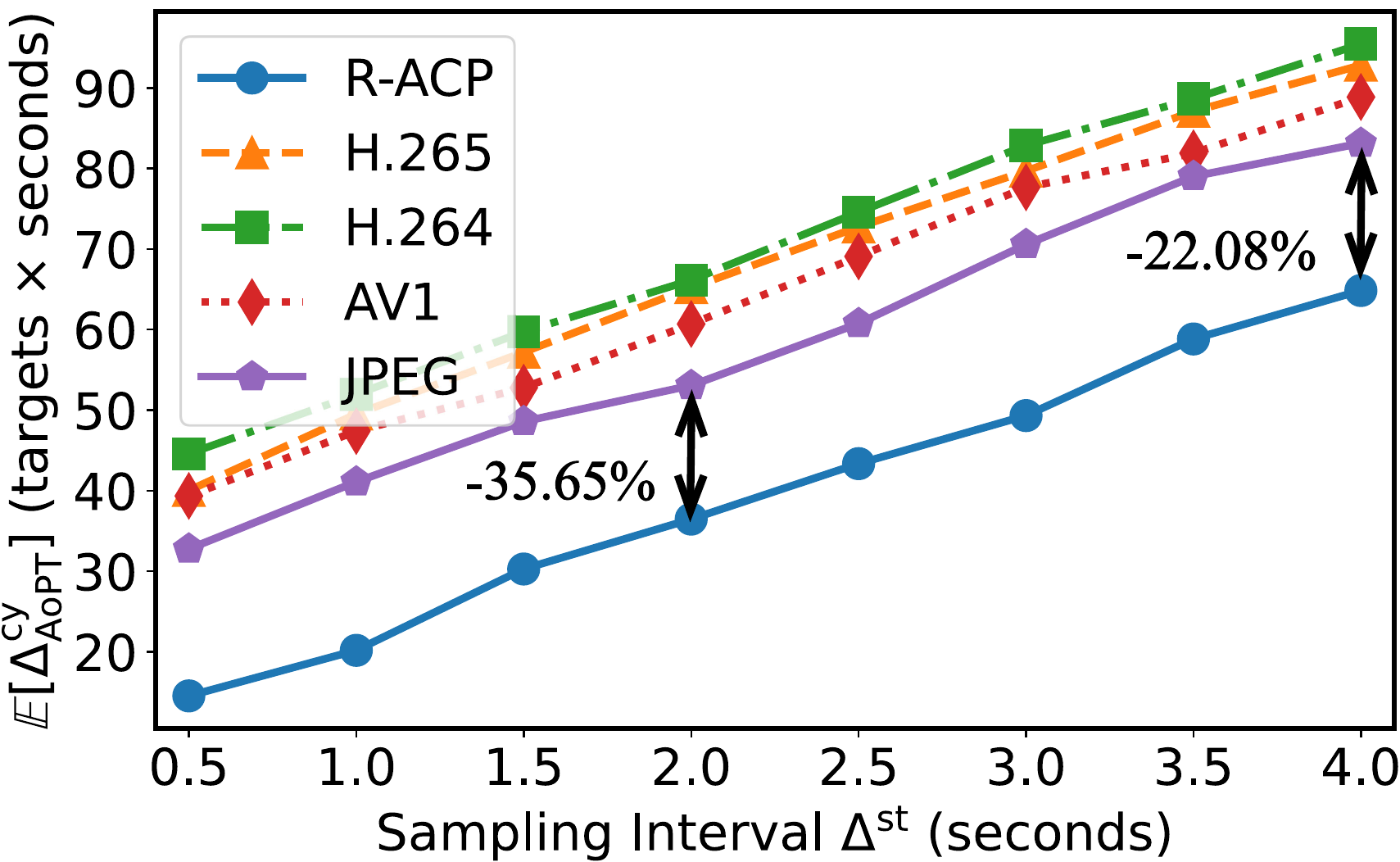}\label{fig:aopt_vs_sampling_interval}
  }
  \caption{AoPT vs. different parameters. Fig. \ref{fig:aopt_vs_capacity} shows the relationship between AoPT and capacity. Fig. \ref{fig:aopt_vs_sampling_interval} demonstrates the relationship between AoPT and sampling interval.}
  \label{fig:aopt_comparison}
  \vspace{-3mm}
\end{figure}

Fig. \ref{fig:aopt_comparison} shows the relationship between AoPT and different system parameters, i.e., communication bottleneck and sampling interval. In Fig. \ref{fig:aopt_vs_capacity}, as the communication bottleneck increases, AoPT decreases due to reduced transmission latency, improving data timeliness across all methods. R-ACP achieves up to 51.96\% lower AoPT compared to H.265, H.264, and AV1. With larger communication capacity, R-ACP enhances data freshness. Fig. \ref{fig:aopt_vs_sampling_interval} shows AoPT increasing with longer sampling intervals due to less frequent data updates. R-ACP consistently maintains lower AoPT across different sampling intervals, reducing it by up to 22.08\% compared to baselines, demonstrating its effectiveness in minimizing AoPT even with lower sampling frequencies.

\begin{figure}[t]
  \centering
  \subfigure[Packet loss rate 0.1.]{
    \includegraphics[width=4.3cm]{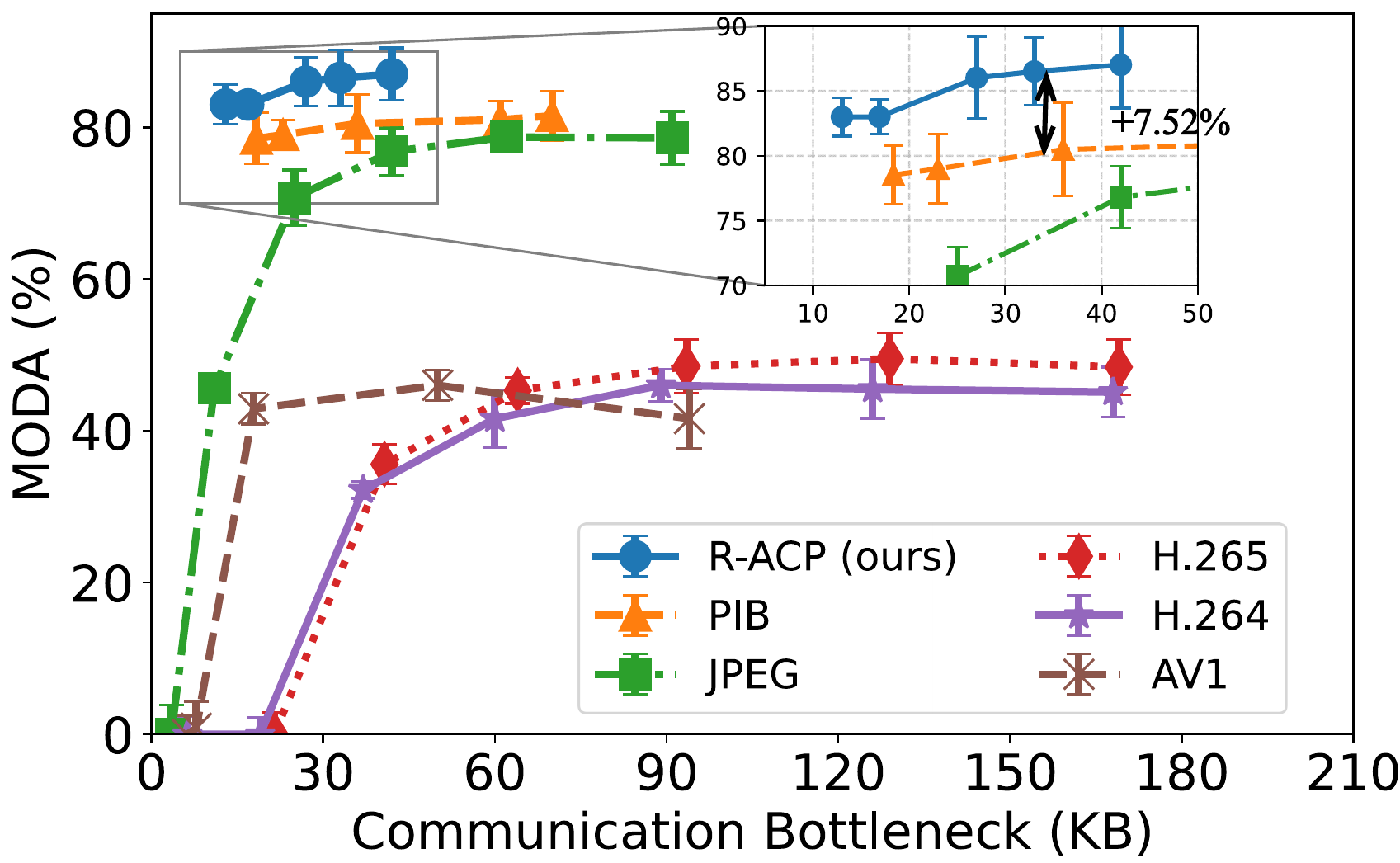}
    \label{fig:dropout_0_1}
  }
  \hspace{-0.5cm}  % Adjusting subfigure spacing
  \subfigure[Packet loss rate 0.2.]{
    \includegraphics[width=4.3cm]{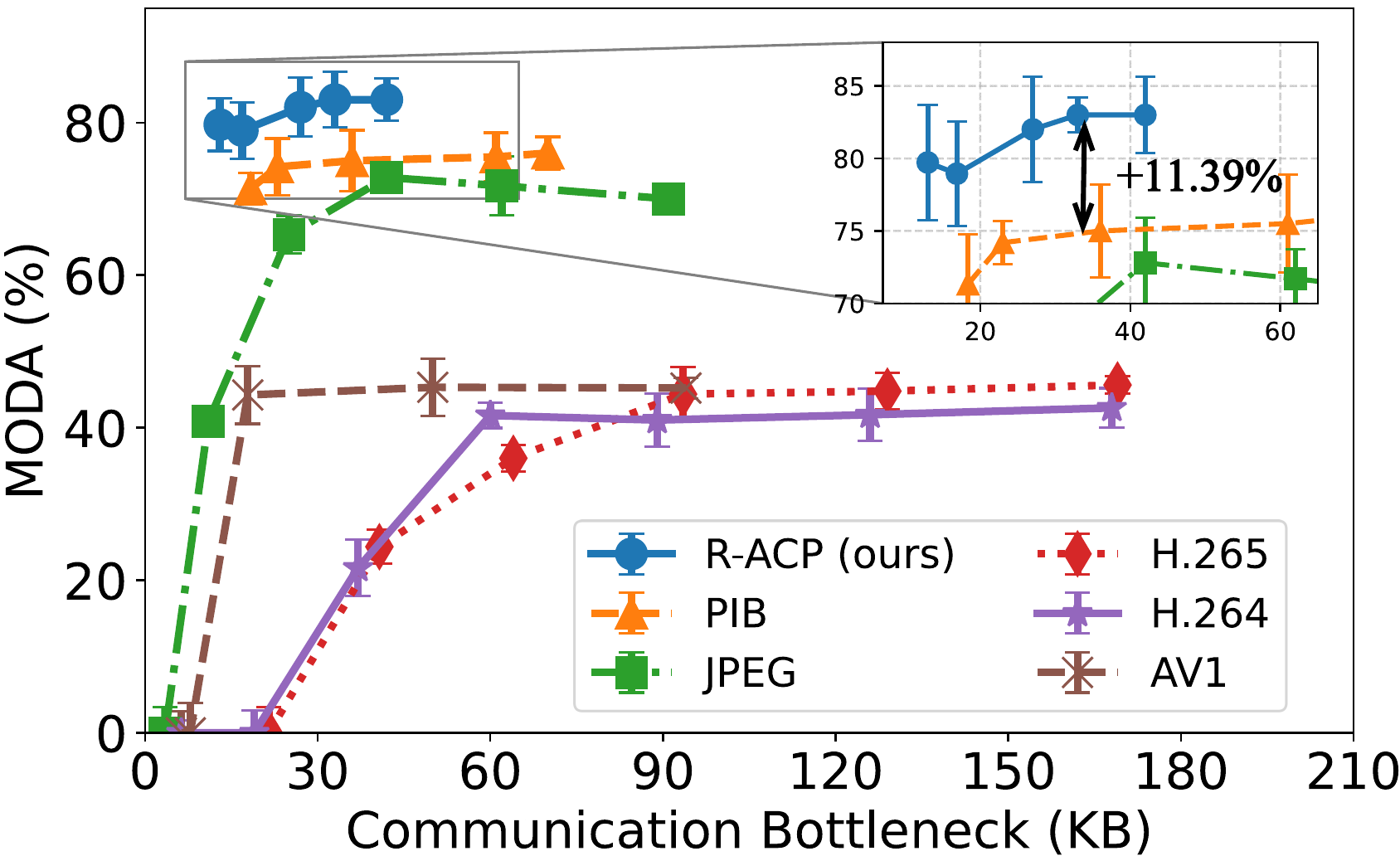}
    \label{fig:dropout_0_2}
  }
  \subfigure[Packet loss rate 0.3.]{
    \includegraphics[width=4.3cm]{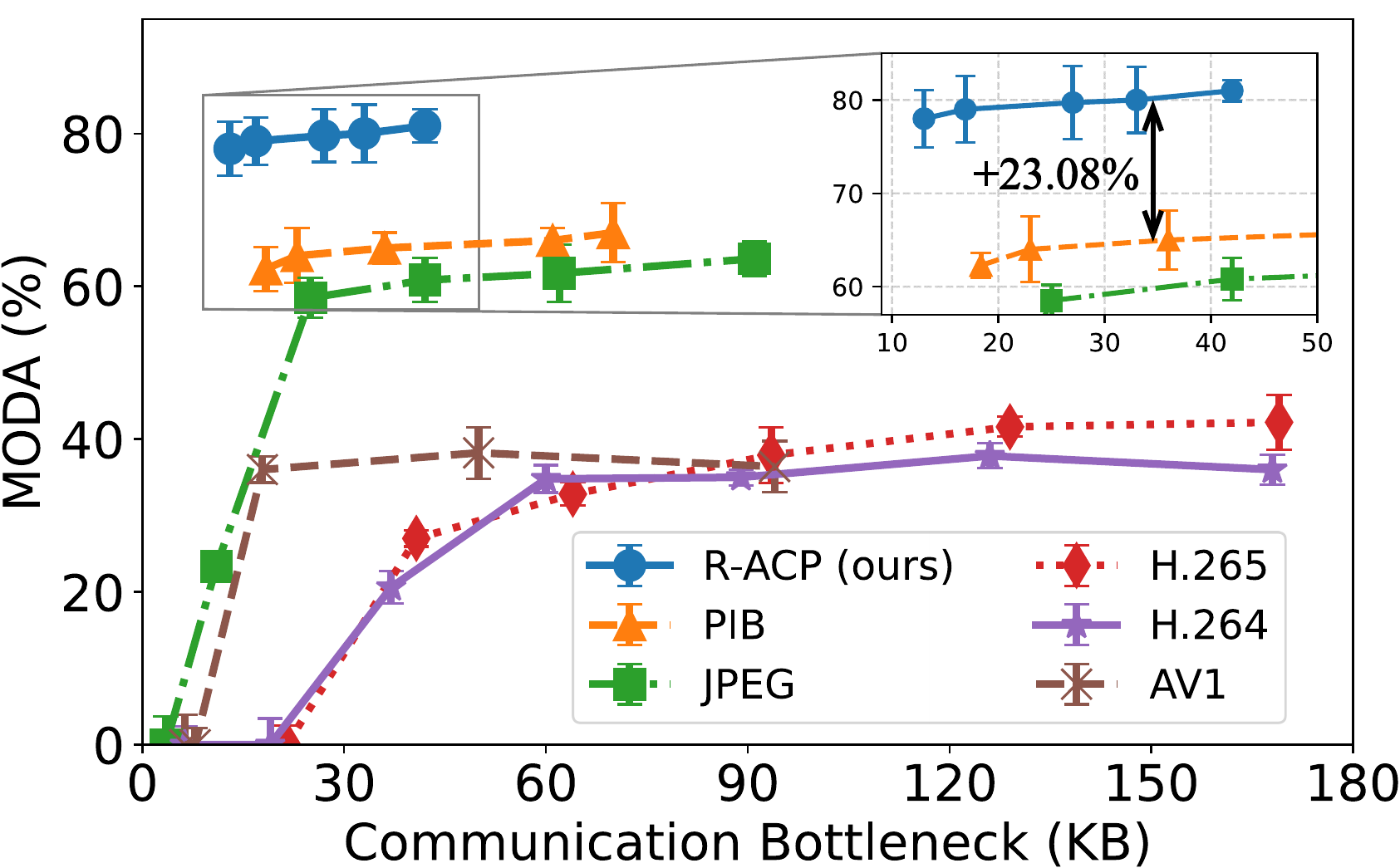}
    \label{fig:dropout_0_3}
  }
  \hspace{-0.5cm}  % Adjusting subfigure spacing
  \subfigure[Packet loss rate 0.4.]{
    \includegraphics[width=4.3cm]{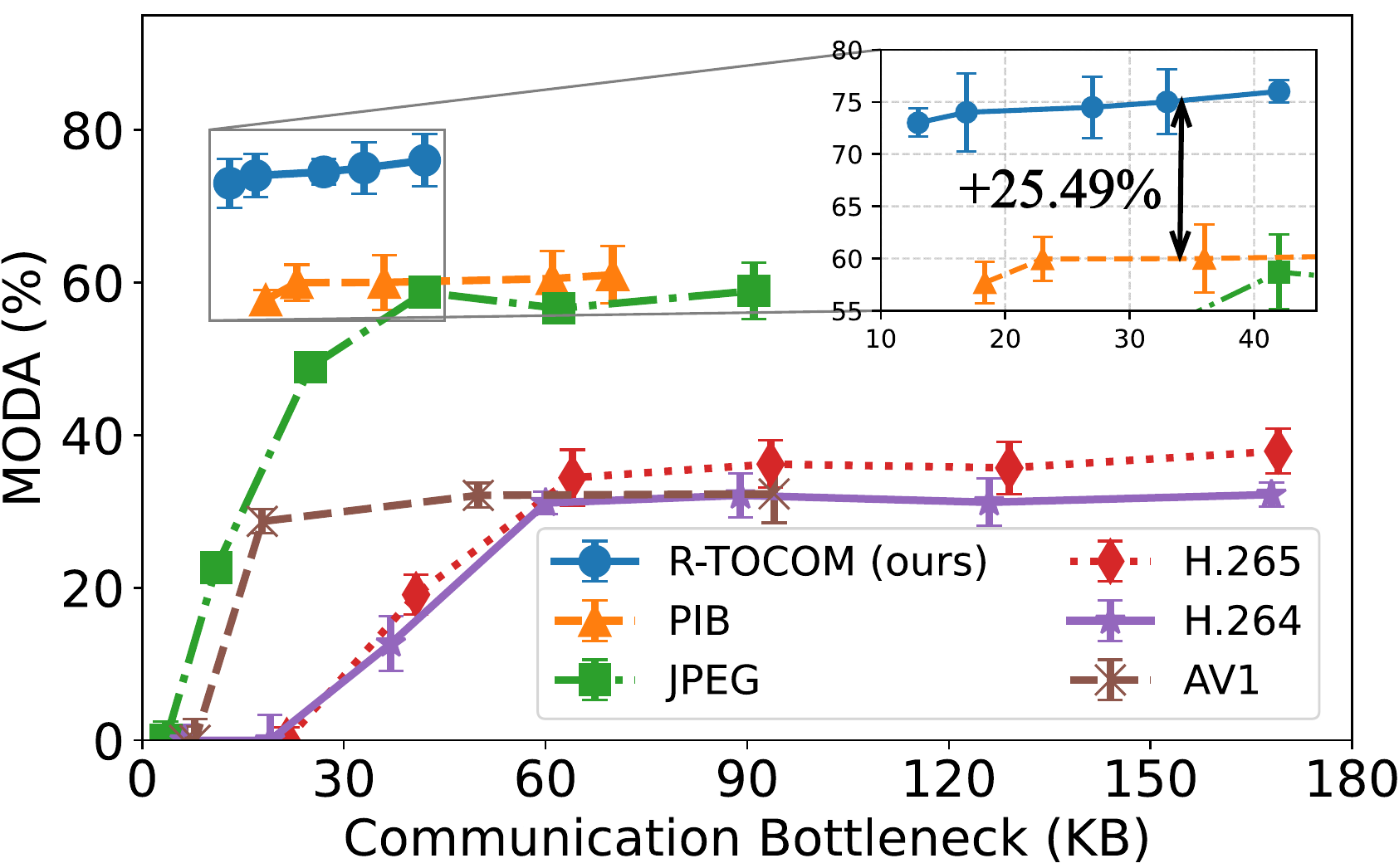}
    \label{fig:dropout_0_4}
  }
  \caption{Comparison of MODA vs. Communication Bottleneck across different packet loss rates. }
  \label{fig:dropout_comparison}
  \vspace{-3mm}
\end{figure}
Fig. \ref{fig:dropout_comparison} shows how packet loss and communication bottlenecks affect MODA. As the bottleneck increases, MODA improves across all packet loss rates, indicating higher transmission capacity enhances perception. For lower packet loss rates (0.1 and 0.2), the improvement is gradual, while at a packet loss rate of 0.3, R-ACP significantly outperforms other baselines like PIB, with at least a 23.08\% improvement. Even under severe packet loss (0.4), R-ACP maintains a notable advantage over H.265, H.264, AV1, and JPEG, achieving up to 25.49\% improvement in MODA, demonstrating robustness against high packet loss scenarios. These results highlight R-ACP's efficiency in maintaining data integrity and accuracy across various conditions.

% \begin{figure}[t]
%   \centering
%   \includegraphics[width=0.50\textwidth]{figure/ucb_communication_bottleneck_vs_total_latency.pdf}
%   \caption{The number of edge servers vs. various latencies under different communication bottlenecks.}
%   \label{fig:edge-servers-vs-latency}
%   \vspace{-3mm}
% \end{figure}
\section{Conclusion}
In this paper, we have proposed a real-time adaptive collaborative perception (R-ACP) framework by leveraging a robust task-oriented communication strategy to enhance real-time multi-view collaborative perception under constrained network conditions. Our contributions of R-ACP are twofold. First, we have introduced a channel-aware self-calibration technique utilizing Re-ID-based feature extraction and adaptive key-point compression, which significantly improves extrinsic calibration accuracy by up to 89.39\%, even with limited FOV overlap. Second, we have leveraged an Information Bottleneck (IB)-based encoding method to optimize feature transmission and sharing, ensuring data timeliness while reducing communication overhead. By intelligently compressing data and employing a priority-based scheduling mechanism for severe packet loss, R-ACP can reduce AoPT and retain perception accuracy under various channel conditions. Extensive simulation results show that R-ACP significantly outperforms traditional methods like PIB, H.265, H.264, and AV1, improving multiple object detection accuracy (MODA) by 25.49\% and decreasing communication costs by 51.36\%, particularly in high packet loss scenarios (up to 40\% packet loss rate).

%7.14检查到这

\ifCLASSOPTIONcaptionsoff
  \newpage
\fi

\bibliographystyle{IEEEtran}
\bibliography{ref,ref2}

\begin{IEEEbiography}[{\includegraphics[width=1in,height=1.25in,clip,keepaspectratio]{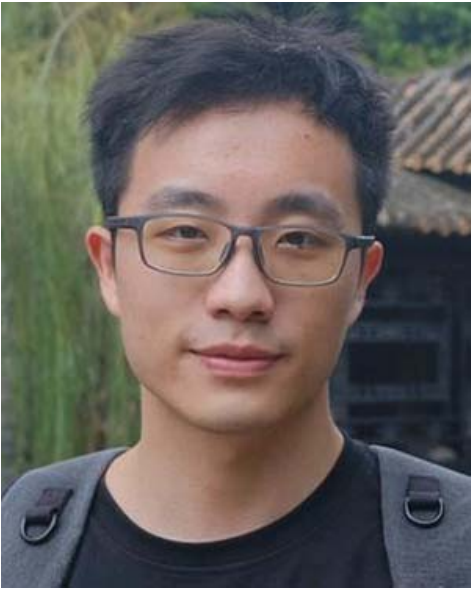}}] {Zhengru Fang} (S'20) received his B.S. degree (Hons.) in electronics and information engineering from the Huazhong University of Science and Technology (HUST), Wuhan, China, in 2019 and received his M.S. degree (Hons.) from Tsinghua University, Beijing, China, in 2022. Currently, he is pursuing his PhD degree in the Department of Computer Science at City University of Hong Kong. His research interests include collaborative perception, V2X, age of information, and mobile edge computing. He received the Outstanding Thesis Award from Tsinghua University in 2022, and the Excellent Master Thesis Award from the Chinese Institute of Electronics in 2023. His research work has been published in IEEE/CVF CVPR, IEEE ToN, IEEE JSAC, IEEE TMC, IEEE ICRA, and ACM MM, etc.
\end{IEEEbiography}

\begin{IEEEbiography}
[{\includegraphics[width=1in,height=1.25in,clip,keepaspectratio]{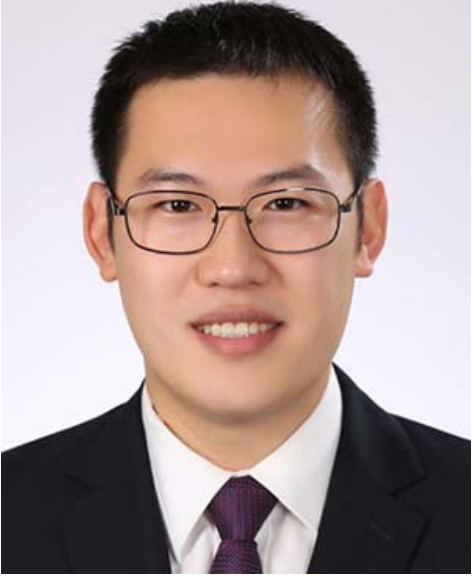}}]{\textbf{Jingjing Wang}} (S'14-M'19-SM'21) received his B.S. degree in Electronic Information Engineering from Dalian University of Technology, Liaoning, China in 2014 and the Ph.D. degree in Information and Communication Engineering from Tsinghua University, Beijing, China in 2019, both with the highest honors. From 2017 to 2018, he visited the Next Generation Wireless Group chaired by Prof. Lajos Hanzo, University of Southampton, UK. Dr. Wang is currently an associate professor at School of Cyber Science and Technology, Beihang University. His research interests include AI enhanced next-generation wireless networks, swarm intelligence and confrontation. He has published over 100 IEEE Journal/Conference papers. Dr. Wang was a recipient of the Best Journal Paper Award of IEEE ComSoc Technical Committee on Green Communications \& Computing in 2018, the Best Paper Award of IEEE ICC and IWCMC in 2019.
% \vspace{-6mm}
\end{IEEEbiography}

\begin{IEEEbiography}[{\includegraphics[width=1in,height=1.25in,clip,keepaspectratio]{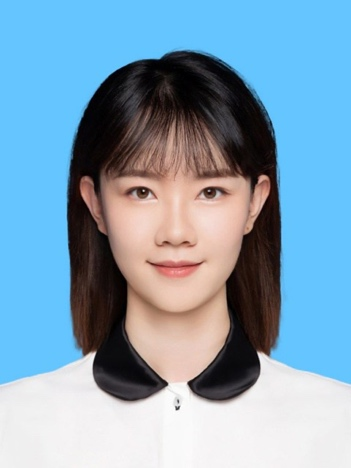}}]{Yanan Ma}( Graduate Student Member, IEEE) received the B.Eng. degree in Electronic Information Engineering (English Intensive) and the M.Eng. degree in Information and Communication Engineering from the Dalian University of Technology, Dalian, China, in 2020 and 2023. She is currently pursuing the Ph.D. degree in the Department of Computer Science at the City University of Hong Kong. Her research interests are focused on edge intelligence, wireless communication and networking.
\end{IEEEbiography}

\begin{IEEEbiography}
[{\includegraphics[width=1in,height=1.25in,clip,keepaspectratio]{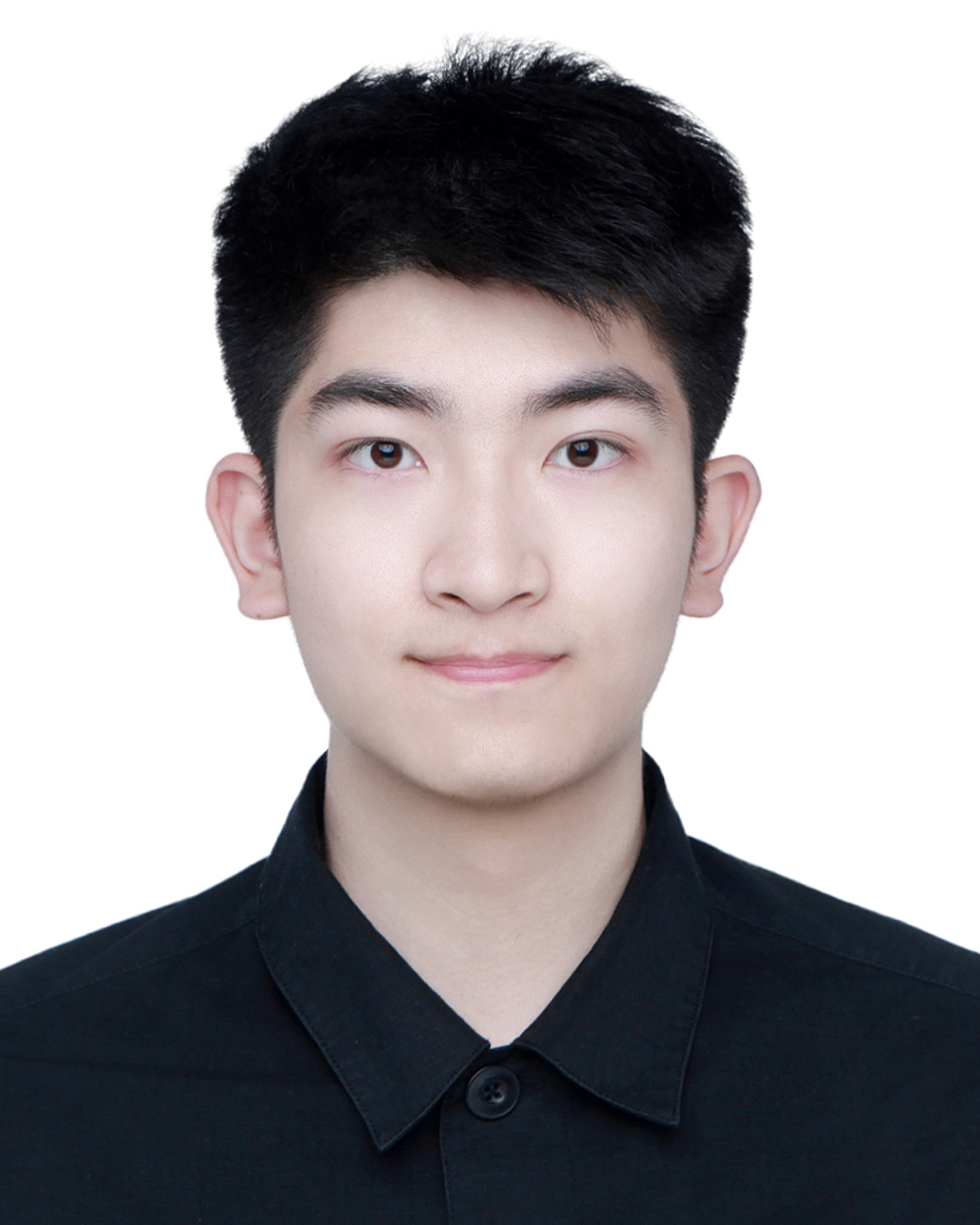}}]{\textbf{Yihang Tao}} received the B.S. degree from the School of Information Science and Engineering, Southeast University, Nanjing, China, in 2021 and received his M.S. degree from the School of Electronic Information and Electrical Engineering, Shanghai Jiao Tong University, Shanghai, China, in 2024. Currently, he is pursuing his PhD degree in the Department of Computer Science at City University of Hong Kong. His current research interests include collaborative perception, autonomous driving, and AI security.
% \vspace{-6mm}
\end{IEEEbiography}

\begin{IEEEbiography}
[{\includegraphics[width=1in,height=1.25in,clip,keepaspectratio]{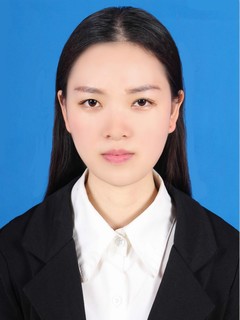}}]{\textbf{Yiqin Deng}} received her MS degree in software engineering and her PhD degree in computer science and technology from Central South University, Changsha, China, in 2017 and 2022, respectively. She is currently a Postdoctoral Researcher with the Department of Computer Science at City University of Hong Kong. Previously, she was a Postdoctoral Research Fellow with the School of Control Science and Engineering, Shandong University, Jinan, China. She was a visiting researcher at the University of Florida, Gainesville, Florida, USA, from 2019 to 2021. Her research interests include edge/fog computing, computing power networks, Internet of Vehicles, and resource management.
% \vspace{-6mm}
\end{IEEEbiography}

\begin{IEEEbiography}[{\includegraphics[width=1in,clip,keepaspectratio]{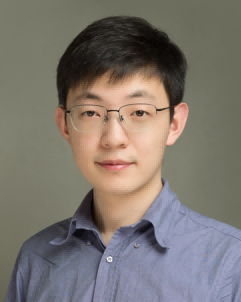}}]{Xianhao Chen}(Member, IEEE) received the B.Eng. degree in electronic information from Southwest Jiaotong University in 2017, and the Ph.D. degree in electrical and computer engineering from the University of Florida in 2022. He is currently an assistant professor at the Department of Electrical and Electronic Engineering, the University of Hong Kong, where he directs the Wireless Information \& Intelligence (WILL) Lab. He serves as a TPC member of several international conferences and an Associate Editor of ACM Computing Surveys. He received the Early Career Award from the Research Grants Council (RGC) of Hong Kong in 2024, the ECE Graduate Excellence Award for research from the University of Florida in 2022, and the ICCC Best Paper Award in 2023. His research interests include wireless networking, edge intelligence, and machine learning.
\end{IEEEbiography}

\begin{IEEEbiography}[{\includegraphics[width=1in,height=1.25in,clip,keepaspectratio]{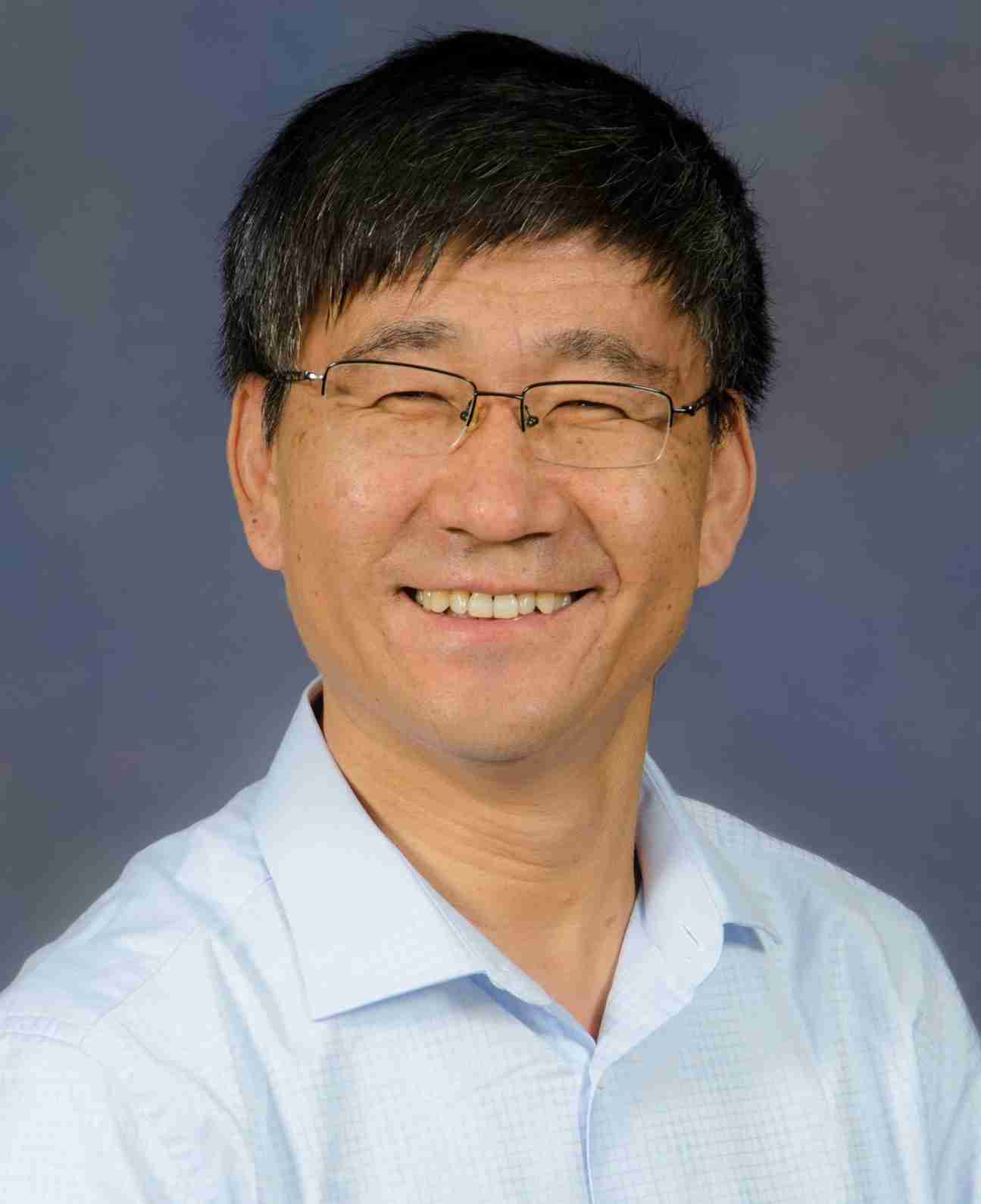}}]{Yuguang Fang}
(S'92, M'97, SM'99, F'08) received
the MS degree from Qufu Normal University, China
in 1987, a PhD degree from Case Western Reserve
University, Cleveland, Ohio, USA, in 1994, and a
PhD degree from Boston University, Boston, Massachusetts, USA in 1997. He joined the Department
of Electrical and Computer Engineering at University of Florida in 2000 as an assistant professor,
then was promoted to associate professor in 2003,
full professor in 2005, and distinguished professor in
2019, respectively. Since August 2022, he has been a Global STEM Scholar and 
Chair Professor with the Department of Computer
Science at City University of Hong Kong.

Prof. Fang received many awards including the US NSF CAREER Award
(2001), US ONR Young Investigator Award (2002), 2018 IEEE Vehicular Technology Outstanding Service Award, IEEE Communications Society
AHSN Technical Achievement Award (2019), CISTC Technical Recognition
Award (2015), and WTC Recognition Award (2014), and 2010-2011 UF
Doctoral Dissertation Advisor/Mentoring Award. He held multiple professorships including the Changjiang Scholar Chair Professorship (2008-2011),
Tsinghua University Guest Chair Professorship (2009-2012), University of
Florida Foundation Preeminence Term Professorship (2019-2022), and University of Florida Research Foundation Professorship (2017-2020, 2006-
2009). He served as the Editor-in-Chief of IEEE Transactions on Vehicular
Technology (2013-2017) and IEEE Wireless Communications (2009-2012)
and serves/served on several editorial boards of journals including Proceedings
of the IEEE (2018-present), ACM Computing Surveys (2017-present), ACM
Transactions on Cyber-Physical Systems (2020-present), IEEE Transactions
on Mobile Computing (2003-2008, 2011-2016, 2019-present), IEEE Transactions on Communications (2000-2011), and IEEE Transactions on Wireless
Communications (2002-2009). He served as the Technical Program Co-Chair of IEEE INFOCOM'2014. He is a Member-at-Large of the Board of
Governors of IEEE Communications Society (2022-2024) and the Director of
Magazines of IEEE Communications Society (2018-2019). He is a fellow of
ACM and AAAS.
\end{IEEEbiography}

\end{document}